\documentclass[acmsmall,screen]{acmart}

\usepackage{algorithm}
\usepackage{algpseudocode}
\usepackage{color,url}
\usepackage{appendix}

\newcommand{\err}{\mathrm{err}}
\newcommand{\flag}{\mathrm{flag}}
\algrenewcommand\algorithmicrequire{\textbf{Input:}}
\algrenewcommand\algorithmicensure{\textbf{Output:}}
\AtBeginDocument{%
  }

\newif\ifdraft
\draftfalse

\newcommand{\jb}[1]{{\ifdraft\color{blue}[{\bf JB:} #1]}\else{}\fi}
\newcommand{\jiawei}[1]{{\ifdraft\color{blue}[{\bf Jiawei:} #1]}\else{}\fi}

\newlength{\tableskip}

\newif\iffull
\fullfalse
\newcommand{\full}[1]{\iffull{#1}\else{}\fi}

\newif\ifchange
\changefalse

\usepackage{enumerate}

\setcopyright{none}

\begin{document}

%%
%% The "title" command has an optional parameter,
%% allowing the author to define a "short title" to be used in page headers.
\title{Probabilistic Floating-Point Round-Off Analysis via Concentration Inequalities}

%%
%% The "author" command and its associated commands are used to define
%% the authors and their affiliations.
%% Of note is the shared affiliation of the first two authors, and the
%% "authornote" and "authornotemark" commands
%% used to denote shared contribution to the research.
\author{Yichen Tao}
\email{ychtao@umich.edu}
\orcid{0009-0001-8655-3372}
\affiliation{
  \institution{University of Michigan}
  \city{Ann Arbor}
  \state{Michigan}
  \country{USA}
}

\author{Hongfei Fu}
\authornote{Corresponding author}
\email{jt002845@sjtu.edu.cn}
\orcid{0000-0002-7947-3446}
\affiliation{
  \institution{Shanghai Jiao Tong University}
  \city{Shanghai}
  \country{China}
}

\author{Jiawei Chen}
\email{chenjw@umich.edu}
\orcid{0000-0002-5461-6711}
\affiliation{
  \institution{University of Michigan}
  \city{Ann Arbor}
  \state{Michigan}
  \country{USA}
}

\author{Jean-Baptiste Jeannin}
\email{jeannin@umich.edu}
\orcid{0000-0001-6378-1447}
\affiliation{
  \institution{University of Michigan}
  \city{Ann Arbor}
  \state{Michigan}
  \country{USA}
}

% \author{Ben Trovato}
% \authornote{Both authors contributed equally to this research.}
% \email{trovato@corporation.com}
% \orcid{1234-5678-9012}
% \author{G.K.M. Tobin}
% \authornotemark[1]
% \email{webmaster@marysville-ohio.com}
% \affiliation{%
%   \institution{Institute for Clarity in Documentation}
%   \city{Dublin}
%   \state{Ohio}
%   \country{USA}
% }

% \author{Lars Th{\o}rv{\"a}ld}
% \affiliation{%
%   \institution{The Th{\o}rv{\"a}ld Group}
%   \city{Hekla}
%   \country{Iceland}}
% \email{larst@affiliation.org}

%%
%% By default, the full list of authors will be used in the page
%% headers. Often, this list is too long, and will overlap
%% other information printed in the page headers. This command allows
%% the author to define a more concise list
%% of authors' names for this purpose.
% \renewcommand{\shortauthors}{Authors}

%%
%% The abstract is a short summary of the work to be presented in the
%% article.
\begin{abstract}
Floating-point round-off errors are ubiquitous in numerically intensive programs arising in fields such as scientific computing and optimization.
As floating-point errors potentially lead to unexpected and catastrophic program failures, one must derive guaranteed round-off thresholds to ensure the correctness of these programs. 
However, deterministic round-off thresholds tend to be too conservative to be usable in practice, since they often involve large round-off errors that occur with small probability. Probabilistic thresholds relax deterministic ones by specifying that the probability of the round-off error exceeding a threshold is below a given confidence.

In this work, we propose a novel approach to probabilistic round-off analysis, by applying concentration inequalities over the Taylor expansion from FPTaylor (TOPLAS 2018).
A major obstacle in applying concentration inequalities is that the Taylor expansion involves absolute value operators that make the calculation of the expected values of 
the first order partial differential terms difficult.
Our first step to overcome this obstacle is a sound over-approximation that removes the absolute value {operators} in polynomial expressions.
Then, we show how to handle fractional expressions by a transformation into polynomial case.
Finally, we show how to improve our approach with range partitioning. 
Our approach is scalable since the key computational part is the calculation of expected values of polynomial expressions with independent variables, for which the linear and independence properties of expectation boost the computation.
Experimental results show that our approach is orders of magnitude more time efficient, while producing thresholds with comparable precision against the state of the art. 
\end{abstract}

%%
%% The code below is generated by the tool at http://dl.acm.org/ccs.cfm.
%% Please copy and paste the code instead of the example below.
%%

%\begin{CCSXML}
%<ccs2012>
% <concept>
%  <concept_id>00000000.0000000.0000000</concept_id>
%  <concept_desc>Do Not Use This Code, Generate the Correct Terms for Your Paper</concept_desc>
%  <concept_significance>500</concept_significance>
% </concept>
% <concept>
%  <concept_id>00000000.00000000.00000000</concept_id>
%  <concept_desc>Do Not Use This Code, Generate the Correct Terms for Your Paper</concept_desc>
%  <concept_significance>300</concept_significance>
% </concept>
% <concept>
%  <concept_id>00000000.00000000.00000000</concept_id>
%  <concept_desc>Do Not Use This Code, Generate the Correct Terms for Your Paper</concept_desc>
%  <concept_significance>100</concept_significance>
% </concept>
% <concept>
%  <concept_id>00000000.00000000.00000000</concept_id>
%  <concept_desc>Do Not Use This Code, Generate the Correct Terms for Your Paper</concept_desc>
%  <concept_significance>100</concept_significance>
% </concept>
%</ccs2012>
%\end{CCSXML}

%\ccsdesc[500]{Do Not Use This Code~Generate the Correct Terms for Your Paper}
%\ccsdesc[300]{Do Not Use This Code~Generate the Correct Terms for Your Paper}
%\ccsdesc{Do Not Use This Code~Generate the Correct Terms for Your Paper}
%\ccsdesc[100]{Do Not Use This Code~Generate the Correct Terms for Your Paper}

\newtheorem{remark}{Remark}

%%
%% Keywords. The author(s) should pick words that accurately describe
%% the work being presented. Separate the keywords with commas.
%\keywords{Do, Not, Us, This, Code, Put, the, Correct, Terms, for, Your, Paper}

%\received{20 February 2007}
%\received[revised]{12 March 2009}
%\received[accepted]{5 June 2009}

%%
%% This command processes the author and affiliation and title
%% information and builds the first part of the formatted document.
\maketitle

\section{Introduction}

Most non-integer numerical computations in computers are performed using floating-point arithmetic.
However, floating-point arithmetic is inherently imprecise due to its finite precision~\cite{goldberg1991every}, which introduces round-off errors at almost every execution step of a floating-point program.
While these errors are typically small in a single floating-point operation, the accumulation of round-off errors in numerically intensive programs can be non-negligible, potentially leading to catastrophic outcomes.
Ensuring that round-off errors in numerically intensive programs are under control necessitates rigorous round-off analysis, i.e., deriving guaranteed thresholds for the floating-point error incurred in numerical computation.

A typical example for the round-off analysis problem is as follows.
Consider the computation of $f(x_1, x_2, x_3) = x_1x_2 + x_3$,
where the input variables $x_1$, $x_2$ and $x_3$ are floats confined within the interval $[-1, 1]$.
Due to the finite precision nature of floating-point representation in calculating the product $x_1x_2$ and the sum $x_1x_2 + x_3$, the computed result deviates from the mathematically exact value obtained under real-number arithmetic.
A typical task of deterministic floating-point round-off analysis is to determine an upper bound on this deviation that holds for any valid input.

Rigorous analysis of floating-point round-off errors has been extensively studied in the literature.
Deterministic analysis targets guaranteed thresholds for the numerical error under any program input, and has been studied in multiple existing works.
Tools like Gappa~\cite{daumas2010certification} and PRECiSA~\cite{titolo2018abstract} adopt an  abstraction-based approach, while FPTaylor~\cite{solovyev2018rigorous} and Real2Float~\cite{magron2017certified} formulate it as an optimization problem.
Alternatively, probabilistic analysis tightens these sometimes overly-conservative thresholds, limiting them to only the most likely input ranges.

In this paper, we consider the \emph{probabilistic} analysis of floating-point errors: its objective is to derive a threshold such that the probability that the round-off error does not exceed the threshold is at least a given confidence level close to $1$. 
Probabilistic analysis is particularly motivated by the overly pessimistic nature of deterministic analysis, since it accounts for worst-case inputs, even those with negligible probability~\cite[Section 4.7]{tekriwal2023mechanized}. 
In contrast, probabilistic analysis is less impacted by these rare input scenarios, providing a more informative and practical insight when input distributions are known.
Recall the previous example of computing $x_1 x_2 + x_3$ in floating-point arithmetic.
If additional information about the distribution of the input variables is provided, such as the assumption that they are uniformly distributed on $[-1, 1]$, it is possible to obtain a ``probabilistic threshold'' for the round-off error, which is a threshold that is not exceeded with probability of at least a given confidence value (e.g., $0.99$).
A typical application motivating probabilistic analysis is the numerical computation in GPS sensor data.
Such data are often modeled by normal distributions with unbounded or large support, which renders worst-case analysis too conservative, and therefore, uninformative. 
In contrast, probabilistic analysis focuses on the bulk of probability mass, yielding thresholds that are more useful  in practice, especially for distributions with large support.

The PAF tool~\cite{constantinides2021rigorous} is the current state of the art in probabilistic round-off analysis. 
PAF traverses the abstract syntax tree of a floating-point expression, calculating probabilistic ranges and error distributions for intermediate results, and using Dempster-Shafer structures (DS-structures)~\cite{ferson2003constructing}, an interval-based data structure to represent the probability distributions.
It uses symbolic affine arithmetic for error propagation and computes conditional round-off errors by maximizing symbolic error forms over the input variable ranges.
Although interval abstraction achieves arbitrary accuracy as the number of intervals approaches infinity, the combinatorial explosion from maintaining a large number of intervals hinders the efficiency of the round-off analysis. 
For instance, in our experimental evaluation, it takes PAF more than two hours to analyze an expression with seven operations on an 18-core machine. 
An earlier tool, PrAn~\cite{lohar2019sound}, discretizes the input distribution into subdomains, applies probabilistic affine arithmetic to analyze each subdomain independently, and then merges the resulting error distributions to compute a global probabilistic error bound, which is extracted as a refined probabilistic error guarantee.
In general, PrAn is faster than PAF, but at the cost of considerably less accurate thresholds.
Hence, conducting accurate and efficient probabilistic round-off analysis remains a challenge.

We address the aforementioned challenge over arithmetic expressions with addition, subtraction, multiplication and division
by applying concentration inequalities to the probabilistic round-off analysis.  Our detailed contributions are as follows:
\begin{itemize}
\item We introduce a novel approach to probabilistic round-off analysis for polynomial arithmetic expressions with addition, subtraction, multiplication, and  division with constant denominators, given specified distributions of input variables. The main idea is to apply concentration inequalities to the Taylor expansion used by FPTaylor. 
A key technical contribution in our approach is a sound relaxation (called \emph{positive-negative decomposition}) of polynomial expressions with absolute values into those without absolute values to {avoid expensive computer algebra for computing positive and negative regions of polynomials.} 
\item We show how to transform the probabilistic round-off analysis of fractional expressions into one of polynomials, so that our approach handles division with non-constant denominators.
\item We show how one can use range partition to further improve the accuracy of our approach, and prove the correctness of our range partition refinement.
\item We implement our algorithm in a prototype named ProbTaylor in OCaml. Our prototype handles uniform, truncated normal, and Laplace distributions for input variables. A technical novelty here is a symbolic approach that efficiently and accurately computes expectation.
\end{itemize}

Evaluation over benchmarks from PAF, FPBench and two realistic examples shows that: (a) ProbTaylor is significantly more time efficient compared to PAF and PrAn; 
(b) ProbTaylor provides thresholds whose accuracy is comparable to or better than those by PAF and PrAn, and in some cases at least an order of magnitude better;
(c) With a moderate range for input variables, the probabilistic thresholds derived from ProbTaylor are significantly tighter than the deterministic ones generated by the deterministic tool FPTaylor;  
(d) Our approach is potentially scalable in handling floating point expressions with many variables and operations.

\section{Preliminaries}
\label{sec:prelim}

In this Section, we present some necessary background regarding floating-point arithmetic, first-order Taylor expansions and probability theory.

\subsection{Floating-Point Arithmetic}
\label{sec:floatingpoint}

In this work, we limit our focus to floating-point arithmetic specified by the IEEE 754 standard~\cite{kahan1996ieee}. 
A binary floating-point number is defined by a triple consisting of a sign bit ($\mathtt{sgn}$), significand bits ($\mathtt{sig}$) and exponent bits ($\mathtt{exp}$), whose numerical value can be expressed as $(-1)^{\mathtt{sgn}} \times \mathtt{sig} \times 2^{\mathtt{exp}}$.
When operations are ``correctly rounded'' and there is no overflow or exception, we have the following model for floating-point operations:
$$ x \circ_{fl} y = (x\circ y)(1 + e) + d \text{,  with } |e| \le \epsilon \text{ and } |d| \le \delta$$
where $\circ\in\{+, -, \times, / \}$ is a basic operation in real numbers, and $\circ_{fl}$ is its floating-point counterpart. 
The error variables $e$ and $d$ account for relative error due to rounding, and absolute error due to underflow, respectively. 
Their upper bounds $\epsilon$ and $\delta$ depend on the floating-point format under which the computation is carried out and the choice of rounding operator.
For instance, if the computation is done in single precision, where $\texttt{sig}$ has $23$ bits and $\texttt{exp}$ has 8 bits, by rounding to nearest, we have $\epsilon=2^{-24}$ and $\delta=2^{-150}$. 
Beyond these bounds, the exact values of $e$ and $d$ are typically considered unknown in round-off analysis. 

Given a multivariate function $f(\mathbf x)$ with variables in a vector $\mathbf{x}$ (bold symbols like $\mathbf x$ are used to denote vectors), its floating-point model $\tilde f(\mathbf x, \mathbf e, \mathbf d)$ can be derived by replacing all operations by the floating-point model mentioned above, and $\mathbf e$ and $\mathbf d$ are vectors of the error variables, i.e., $\mathbf e = (e_1, \dots, e_k)$ and $\mathbf d = (d_1, \dots, d_k)$.
The length $k$ of $\mathbf e$ and $\mathbf d$, 
is equal to the number of operations in $f(\mathbf x)$, i.e., each floating-point operation corresponds to exactly one $(e_i, d_i)$ from $\mathbf e$ and $\mathbf d$. 
Note that $|e_i| \le \epsilon$ and $|d_i| \le \delta$ for $i=1, \dots, k$.
For example, the floating-point version of $x \times y + z$ is given by $((x\times y)(1 + e_1) + d_1 + z)(1+e_2) + d_2$, where $(e_1, d_1)$ arises from the inner multiplication and $(e_2, d_2)$ 
from the outer addition.
It is guaranteed that for some actual error values $(\tilde{\mathbf{e}}, \tilde{\mathbf{d}})$ for $(\mathbf{e},\mathbf{d})$,
the deviated value $\tilde f(\mathbf x, \tilde{\mathbf e}, \tilde{\mathbf d})$ is
equal to the actual floating-point computation result of $f(\mathbf x)$. 
\jb{Very heavy wording. I'd just say it is guaranteed that there exist (e,d)... without introducting yet another notation with $\tilde e$ and $\tilde d$.}
Additionally, it holds that $\tilde f(\mathbf x, \mathbf 0, \mathbf 0) = f(\mathbf x)$.
The magnitude of this deviation at a particular program input $\mathbf{x}$ can be formalized by the following absolute round-off function $\err(f, {\mathbf x)}$: 
$$ \err(f, \mathbf x) := \left|\tilde f(\mathbf x, \tilde{\mathbf e}, \tilde{\mathbf d}) - f(\mathbf{x})\right|.$$

\subsection{First-Order Taylor Expansion}
\label{sec:taylor}

A sufficiently smooth function can be approximated by its \emph{Taylor expansion}~\cite{rudin1964principles}, which is a polynomial expression in terms of the function's derivatives at a given point.
For a $k$-ary function $f(w_1, \dots, w_k)$ that is at least twice continuously differentiable on its domain $D\subset \mathbb R^k$, its first-order Taylor expansion around a particular point $(a_1, \dots, a_k) \in D$ can be expressed as
$$ \textstyle f(\mathbf w) = f(\mathbf a) + \sum_{i=1}^k\frac{\partial f}{\partial w_i}(\mathbf a)(w_i-a_i) + \frac12\sum_{i,j = 1}^k\frac{\partial^2 f}{\partial w_i \partial w_j}(\mathbf z)(w_i-a_i)(w_j-a_j)$$
for some $\mathbf z\in D$, where $\mathbf w = (w_1, \dots, w_k)$ and $\mathbf a = (a_1, \dots, a_k)$. 

We follow FPTaylor~\cite{solovyev2018rigorous} and apply a first-order Taylor expansion to the floating-point version $\tilde f(\mathbf x, \mathbf e, \mathbf d)$ {from Section~\ref{sec:floatingpoint}}.
Henceforth, applying Taylor's theorem to $\tilde f(\mathbf x, \mathbf e, \mathbf d)$ w.r.t 
$e_1, \dots, e_k, d_1, \dots, d_k$ around zero gives us
$$  \textstyle \tilde f(\mathbf x, \mathbf e, \mathbf d) = \tilde f (\mathbf x, \mathbf 0, \mathbf 0) + \sum_{i=1}^k\frac{\partial \tilde f}{\partial e_i}(\mathbf x, \mathbf 0, \mathbf 0)e_i + R_2(\mathbf x, \mathbf e, \mathbf d).$$
For sake of brevity, set $y_i \triangleq e_i$ for $i = 1, \dots, k$ and $y_i \triangleq d_{i-k}$ for $i = k+1, \dots, 2k$.
Then, we have
$$\textstyle  R_2(\mathbf x, \mathbf e, \mathbf d) = \frac12\sum_{i,j=1}^{2k}\frac{\partial^2 \tilde f}{\partial y_i \partial y_j}(\mathbf x, \mathbf e', \mathbf d')y_iy_j + \sum_{i=1}^k \frac{\partial \tilde f}{\partial d_i}(\mathbf x, \mathbf 0, \mathbf 0)d_i$$
for some $\mathbf e', \mathbf d' \in \mathbb R^{k}$ satisfying $|e_i'| \le \epsilon$ and $|d_i'| \le \delta$ for $i = 1, \dots, k$. 
The first-order terms $\frac{\partial \tilde f}{\partial d_i}(\mathbf x, \mathbf 0, \mathbf 0)d_i$ is added to $R_2$ since $\delta \ll \epsilon^2$ for common floating-point formats. 
As is mentioned before, $\tilde f(\mathbf x, \mathbf 0, \mathbf 0) = f(\mathbf x)$. 
Therefore, we have a new expression for $\err(f, \mathbf x)$, which we may further relax using the triangular inequality:
$$ \begin{aligned} \textstyle
\err(f, D) \le \left| \sum_{i=1}^k\frac{\partial \tilde f}{\partial e_i}(\mathbf x, \mathbf 0, \mathbf 0)e_i + R_2(\mathbf x, \mathbf e, \mathbf d) \right|
\le \sum_{i=1}^k\left|\frac{\partial \tilde f}{\partial e_i}(\mathbf x, \mathbf 0, \mathbf 0)e_i \right| + \left|R_2(\mathbf x, \mathbf e, \mathbf d) \right|.
\end{aligned} $$
Throughout the paper, we denote the ``first-order error term'' $\sum_{i=1}^k\left|\frac{\partial \tilde f}{\partial e_i}(\mathbf x, \mathbf 0, \mathbf 0)e_i \right|$ by $F(\mathbf x, \mathbf e)$.
We call $|R_2(\mathbf x, \mathbf e, \mathbf d)|$ the ``second-order error term''.

\subsection{Probability Theory}
\label{sec:probth}

\full{
A \emph{probability space} is a triple $(\Omega, \mathcal{F},\mathbb{P})$ where 
(a) $\Omega$ is the sample space; 
(b) $\mathcal{F}$ is a sigma-algebra on $\Omega$ which is a subset of the powerset $2^\Omega$ of the sample space $\Omega$ that contains $\emptyset$ and is closed under the complementation and countable union operations of sets; and 
(c) $\mathbb{P}:\mathcal{F}\rightarrow [0,1]$ is a \emph{probability measure} on $\mathcal{F}$ that is countably additive and satisfies $\mathbb{P}[\Omega]=1$. An \emph{event} is an element $E\in \mathcal{F}$, for which the value
$\mathbb P[E]$ denotes the probability that the event $E$ happens~\cite{kolmogorov2018foundations}. 
A \emph{random variable} $X$ in a probability space $(\Omega,\mathcal{F},\mathbb{P})$ is a measurable function $X:\Omega\rightarrow \mathbb{R}$ measurable in the measurable space $(\Omega, \mathcal{F})$, i.e., for every $d\in \mathbb{R}$ we have the subset  $\{\omega\in\Omega\mid X(\omega)\le d\}$ ($X\le d$ for short) lies in $\mathcal{F}$. The probability distribution of a random variable $X$ is usually specified by its \emph{cumulative distribution function} (CDF) $c(X)$ whose definition is given by $c(X)(t):=\mathbb{P}[X \le t]$.
The expected value $\mathbb{E}[X]$ of a random variable $X$ in a probability space $(\Omega,\mathcal{F},\mathbb{P})$ is formally defined as the integral $\int X\,\mathrm{d}\mathbb{P}$. 
See standard textbooks~\cite{blitzstein2019introduction,ccinlar2011probability,DBLP:books/daglib/0073491} for a formal treatment of these concepts.
}

Given a \emph{probability space} $(\Omega,\mathcal{F},\mathbb{P})$ with sample space $\Omega$, set of events $\mathcal{F}$ and \emph{probability measure} $\mathbb{P}:\mathcal{F}\rightarrow [0,1]$, the \emph{probability} that an event $E\in \mathcal{F}$ happens is denoted $\mathbb P[E]$. A \emph{random variable} $X:\Omega\rightarrow \mathbb{R}$ is such that the subset $\{\omega\in\Omega\mid X(\omega)\le d\}\in\mathcal{F}$ for any $d\in \mathbb{R}$, and its \emph{expected value} $\mathbb{E}[X]$ is defined as the integral $\int X\,\mathrm{d}\mathbb{P}$.
Given a \emph{continuous} random variable $X$, its probability density function (PDF) is a Lebesgue-measurable function $g: \mathbb{R}\rightarrow [0,\infty)$ such that $c(X)(t) = \int_{-\infty}^{t}g(x)\,\mathrm{d}x$ for all reals $t$.
Given a PDF $g$ of a continuous random variable $X$, 
we have that the expected value $\mathbb{E}[X]$ is equal to $\int_{-\infty}^{\infty}x\cdot g(x)\,\mathrm{d}x$. 
Given two events $E_1, E_2 \in \mathcal{F}$ with $\mathbb P[E_2] > 0$, the \emph{conditional probability} of $E_1$ given $E_2$ is defined as $\mathbb P[E_1 \mid E_2] \triangleq \mathbb P[E_1 \cap E_2] / \mathbb P[E_2]$. 
Similarly, for a continuous random variable $X$ and an event $E\in \mathcal{F}$ with $\mathbb P[E] > 0$, the \emph{conditional expectation} of $X$ given $E$ is defined by $\mathbb E[X \mid E] {\triangleq} \left(\int_E X\mathrm d\mathbb P\right) / \mathbb P[E]$. See~\cite{blitzstein2019introduction,ccinlar2011probability,DBLP:books/daglib/0073491} for a formal treatment of these concepts.

\emph{Concentration inequalities} are used for bounding the probability that a random variable deviates largely from its majority part of probability mass.
Here, we use Markov's inequality as stated below.

\begin{theorem}[Markov's Inequality]
If $X$ is a non-negative random variable (i.e., a random variable that always takes non-negative values), then for any constant $a > 0$,
$ \mathbb P[X \ge a] \le \mathbb E [X]/a$.
\label{thm:markov}
\end{theorem}

In this work, we take the $n$-th order higher moment $X^n$ of a random variable $X$ and apply the Markov's inequality to the random variable $X^n$ to obtain 
$ \mathbb P[|X|\ge a] = \mathbb P[|X|^n \ge a^n] \le \mathbb E [|X|^n]/a^n$.
We refer to the order $n$ above as the \emph{analysis order}. 

\smallskip
\noindent{\em Problem Statement.} We consider the following probabilistic round-off analysis problem: 
\begin{itemize}
\item \textbf{Input:} (i) an arithmetic expression $f(\mathbf{x})$ that consists of addition, multiplication, subtraction and division over a vector $\mathbf{x}$ of variables, (ii) a map $\mathcal{D}$ that assigns a probability distribution $\mathcal{D}(x)$ to every variable $x$ in the vector $\mathbf{x}$, and (iii) a target confidence level $c\in (0,1)$.
\item \textbf{Output:} a threshold $U_c > 0$ such that 
$\mathbb{P}_\mathcal{D}[\err(f, \mathbf x) < U_c]> c$,
for which the probability measure $\mathbb{P}_\mathcal{D}$ corresponds to the independent joint distribution of the distributions $\mathcal{D}(x)$ for every variable $x$ in the vector $\mathbf{x}$.
\end{itemize}

\begin{example}
Consider $f(\mathbf{x}) = x_1\times x_2 + x_3$, where $x_1, x_2, x_3$ are all uniformly distributed on $[-1, 1]$. 
Floating-point computations are done in single precision, and thus  $\epsilon = 2^{-24}$ and $\delta = 2^{-150}$. 
We aim to derive a round-off error threshold $U_{0.99}$ for $f(\mathbf{x})$ that holds with probability at least $99\%$. 
\label{ex:prob} \label{ex:poly}
\end{example}

\begin{example}
Consider $g(\mathbf x) = (x_1\times x_2) / (x_3 + 5)$, where $x_1, x_2, x_3$ are all uniformly distributed on the interval $[-1, 1]$.
Floating-point computations are carried out in single precision, and thus $\epsilon = 2^{-24}$ and $\delta = 2^{-150}$.
We aim to derive a round-off error threshold $U_{0.99}$ for $g(\mathbf x)$. 
\label{ex:div}
\end{example}

\section{Overview of Our Results}
\label{sec:overview}

In this section, we perform {probabilistic} round-off error analysis on Example~\ref{ex:poly} and Example~\ref{ex:div} to illustrate our approach to the problem, for polynomials and fractional expressions, respectively.
Note that, for the sake of simplicity, the method presented in this Section has some minor differences with our algorithm described in Section~\ref{sec:algorithm}.

At a high level, our approach to the probabilistic round-off analysis is divided into three steps. 
The first step is to apply a Taylor expansion to obtain the first-order and second-order error terms as in Section~\ref{sec:taylor}. 
The second step is to relax the absolute values in the first-order term to obtain a sound over-approximation without absolute values. 
The final step is to apply Markov's inequality (Theorem~\ref{thm:markov}) to obtain the probabilistic threshold for the round-off error.

\subsection{Illustration of Example~\ref{ex:poly}}
\label{sec:exploy}

{Below we detail the steps for performing probabilistic round-off error analysis on  Example~\ref{ex:poly}.}

\smallskip
\noindent{\bf Step 1: Applying Taylor Expansion.}
We first apply a Taylor expansion to the floating-point model $\tilde f(\mathbf x, \mathbf e, \mathbf d) = (x_1x_2(1 + e_1) + d_1 + x_3)(1 + e_2) + d_2$, as shown in Section \ref{sec:taylor}.
Here $e_1, e_2 \in [-\epsilon ,\epsilon]$ correspond to the relative errors associated with normal number results arising from multiplication (resp. addition), and $d_1, d_2 \in [-\delta, \delta]$ correspond to the absolute errors associated with the subnormal number results. 
Recall that the Taylor expansion gives two error terms:
the first-order error term $F(\mathbf x, \mathbf e)$ and the second-order error term $|R_2(\mathbf x, \mathbf e, \mathbf d)|$.
For Example~\ref{ex:poly}, by computing the partial derivatives and summing them up, we derive that
$F(\mathbf x, \mathbf e)= |(x_1x_2)\cdot e_1| + |(x_1x_2 + x_3)\cdot e_2|$, and 
 $|R_2(\mathbf x, \mathbf e, \mathbf d)|=d_1 + d_2 + e_1e_2x_1x_2 + e_2d_1$.

Since the round-off error is dominated by first-order term $F(\mathbf x, \mathbf e)$ (typically by several orders of magnitude), we bound the second-order error deterministically with an existing global optimization tool (GELPIA~\cite{gelpia} in our implementation), and focus mainly on the first-order error term.

\smallskip
\noindent{\bf Step 2: Over-approximation for Absolute Values.}
Then, we analyze the probabilistic threshold of the first-order error term $F(\mathbf x, \mathbf e)$. A major obstacle here is the probabilistic evaluation of the absolute values in the first-order term.  To overcome this difficulty, we propose a sound over-approximation by introducing fresh variables to separate the positive and negative parts that an input variable can take, which we term \emph{positive-negative (PN) decomposition}.

\smallskip
\noindent\emph{PN Decomposition}. For each input variable $x$ in the vector $\mathbf{x}$, we introduce two fresh variables, $x^+$ and $x^-$, which represent the positive part $x^+ := \max\{x, 0\}$ and the negative part $x^- := \max\{-x, 0\}$, respectively. 
The probability distributions of the new variables $x^+,x^-$ are accordingly calculated as those of $\max\{x, 0\},\max\{-x, 0\}$ from the distribution $\mathcal{D}(x)$ of the variable $x$. 
We then replace each variable $x$ with $x=x^+-x^-$ in every expression within the absolute-value operators from $F(\mathbf{x},\mathbf{e})$ and expand these expressions into a summation of products of variables. After the expansion, we soundly remove the absolute values by identifying the PN parts in the expanded expressions. 
A key simplification here is that during the expansion, we can use the equality $x^+\cdot x^- = 0$ to substantially simplify the resultant expression. 

For Example \ref{ex:prob}, we have the relaxation of 
the term $|(x_1x_2)\cdot e_1|$ in $F(\mathbf x, \mathbf e)$ by:
$$\begin{aligned}
|(x_1x_2)\cdot e_1| &= |(x_1^+x_2^+ + x_1^-x_2^-)e_1 - (x_1^+x_2^- + x_1^-x_2^+)e_1|  
\le |(x_1^+x_2^+ + x_1^-x_2^-)e_1| + |(x_1^+x_2^- + x_1^-x_2^+)e_1| \\
&= (x_1^+x_2^+ + x_1^-x_2^-)|e_1| + (x_1^+x_2^- + x_1^-x_2^+)|e_1| 
\le (x_1^+x_2^+ + x_1^-x_2^- + x_1^+x_2^- + x_1^-x_2^+)\epsilon. 
\end{aligned}$$
In the relaxation above, the first equality comes from substituting each $x_i$ with $x_i^+ - x_i^-$, the second inequality is obtained via applying the triangular inequality, the third equality is by safely removing the absolute-value operator due to the PN decomposition, and the final inequality results from choosing $e_i$'s as the maximum value $\epsilon$.
The term $|(x_1x_2 + x_3) \cdot e_2|$ can be handled similarly:
$$
|(x_1x_2 + x_3)\cdot e_2| \le (x_1^+x_2^+ + x_1^-x_2^- + x_3^+ + x_1^+x_2^- + x_1^-x_2^+ + x_3^-)\epsilon.
$$
Combining the two yields the following sound over-approximation for $F(\mathbf x, \mathbf e)$:
$$F(\mathbf x, \mathbf e) \le (2(x_1^+x_2^+ + x_1^-x_2^- + x_1^+x_2^- + x_1^-x_2^+) + x_3^+ + x_3^-)\epsilon \triangleq p(\mathbf x) \cdot \epsilon.$$
Note that we write $p(\mathbf x)$ instead of $p(x_1^+,x_1^-, x_2^+, x_2^-, x_3^+, x_3^-)$ for brevity. 

\smallskip
\noindent{\bf Step 3: Applying Markov's inequality.}
Finally, we apply Markov's inequality (Theorem~\ref{thm:markov}).
For a given error threshold $u > 0$, we bound the ``threshold violation probability'' for the first-order term $F(\mathbf x, \mathbf e)$, i.e., $\mathbb P[F(\mathbf x, \mathbf e) \ge u]$, with analysis order $n=2$:
$$  \mathbb P[F(\mathbf x, \mathbf e) \ge u] = \mathbb P[(F(\mathbf x, \mathbf e))^2 \ge u^2] \le \mathbb E[(F(\mathbf x, \mathbf e))^2]/u^2 \le \epsilon^2\cdot \mathbb E[(p(\mathbf x))^2]/u^2.$$
Therefore, to derive a threshold $u$ that guarantees a threshold violation probability below $1-c$, it suffices to have $\frac{\epsilon^2\cdot \mathbb E[(p(\mathbf x))^2]}{u^2} \le 1 - c$, which indicates that $u = \epsilon \cdot \sqrt{\mathbb E[(p(\mathbf x))^2] / (1-c)}$ is a valid threshold.
Setting $c=0.99$ for Example~\ref{ex:prob}, we obtain 
$u\approx 6.74 \times 10^{-7}$ so that $\mathbb P[F(\mathbf x, \mathbf e) \ge u]\le 0.01$. 

Note that our approach simplifies the calculation of the expected values {related to} the over-approximation $p(\mathbf{x})$ by the linear and independence properties of expectation. That is, our approach calculates the expected value of each finite product 
$\prod_i x_i^{a_i}$ as $\prod \mathbb{E}(x_i^{a_i})$
and takes a summation over these expected values {(possibly multiplied with corresponding coefficients)}. 

For the second-order term, we apply an existing global optimization tool, GELPIA~\cite{gelpia}, to obtain a deterministic upper bound.
In this example, $R_2(\mathbf x, \mathbf e, \mathbf d) = d_1 + d_2 + e_1e_2x_1x_2 + e_2d_1$. 
By feeding the input ranges of the variables to GELPIA, we get the upper bound $|R_2(\mathbf x, \mathbf e, \mathbf d)| \le 3.56 \times 10^{-15}\triangleq\theta$, which is neglectable compared with the first-order threshold $u$ for $c=0.01$ determined as above. Finally, we conclude that for Example~\ref{ex:prob}, $\mathbb P[\err(f, \mathbf x) \ge 6.74 \times 10^{-7} + \theta]\le 0.01$. 
A refined version of the analysis of this example with partition of value ranges is presented in Section~\ref{sec:partition}.

\subsection{Illustration of Example~\ref{ex:div}}
\label{sec:exdiv}

Below we demonstrate how the three-step procedure can be applied to the fractional 
expression $g(\mathbf x) = (x_1 \times x_2) / (x_3 + 5)$ in Example~\ref{ex:div}.
Most {details} 
are analogous to those employed in the analysis of Example~\ref{ex:poly}, with the exception of certain adjustments in {Step 2} to handle denominators.

\smallskip
\noindent {\bf Step 1: Applying Taylor Expansion.}
The floating-point model of $g(\mathbf x)$ is derived by replacing each operation with its floating-point model:
$\tilde g(\mathbf x, \mathbf e, \mathbf d) = \frac{x_1x_2(1 + e_1) + d_1}{(x_3 + 5)(1 + e_2) + d_2}(1 + e_3) + d_3$.
\full{$\tilde g(\mathbf x, \mathbf e, \mathbf d) = ((x_1x_2(1 + e_1) + d_1) / ((x_3 + 5)(1 + e_2) + d_2))(1 + e_3) + d_3$.}
Subsequently, the first-order and second-order error terms can be computed by applying Taylor expansion. 
Specifically, the first-order term is 
$ F(\mathbf x, \mathbf e) = \frac{|x_1x_2e_1| + |x_1x_2e_2| + |x_1x_2e_3|}{x_3 + 5}$.
\full{$F(\mathbf x, \mathbf e) = (|x_1x_2e_1| + |x_1x_2e_2| + |x_1x_2e_3|)/(x_3 + 5)$.}
\full{$$ 
F(\mathbf x, \mathbf e) = \left| \frac{x_1x_2}{x_3 + 5} e_1\right| + \left| -\frac{x_1x_2}{x_3 + 5} e_2\right| + \left| \frac{x_1x_2}{x_3 + 5} e_3\right| = \frac{1}{x_3 + 5} (|x_1x_2e_1| + |x_1x_2e_2| + |x_1x_2e_3|).
$$}
The absolute value around $x_3 + 5$ can be dropped since it is positive, because we know that $x_3\in [-1, 1]$.
The second-order term $|R_2(\mathbf x, \mathbf e, \mathbf d)|$ is omitted here for brevity.
As before, a deterministic upper bound for $|R_2(\mathbf x, \mathbf e, \mathbf d)|$ is obtained using off-the-shelf tools.

\smallskip
\noindent {\bf Step 2: Over-approximation for Absolute Values.}
We apply the relaxation technique for polynomials employed over Example~\ref{ex:poly} 
to the expressions $|x_1x_2e_1|$, $|x_1x_2e_2|$, $|x_1x_2e_3|$ in the numerator of $F(\mathbf x, \mathbf e)$.
{We} obtain $|x_1x_2e_1| \le(x_1^+x_2^+ + x_1^-x_2^- + x_1^+x_2^- + x_1^-x_2^+)\epsilon$, and analogous bounds for the other two { expressions}.
Consequently, we derive the following over-approximation for $F(\mathbf x, \mathbf e)$:
$$  F(\mathbf x, \mathbf e) \le 3(x_1^+x_2^+ + x_1^-x_2^- + x_1^+x_2^- + x_1^-x_2^+)\epsilon/(x_3 + 5).$$
Define $p(\mathbf x) { \triangleq} 3(x_1^+x_2^+ + x_1^-x_2^- + x_1^+x_2^- + x_1^-x_2^+)$ and $Q(\mathbf x) {\triangleq} x_3 + 5$, we have $F(\mathbf x, \mathbf e) \le {p(\mathbf x) \cdot \epsilon}/ { Q(\mathbf x)}$.

\smallskip
\noindent {\bf Step 3: Applying Markov's inequality.}
Before applying Markov's inequality, we first perform the following transformation. For a fixed threshold $u$, since $Q(\mathbf x)>0$, {we have}
$$\mathbb P[F(\mathbf x, \mathbf e) \ge u] \le \mathbb P\left[{p(\mathbf x) \cdot \epsilon}/{Q(\mathbf x)} \ge u \right] = \mathbb P[p(\mathbf x) \cdot \epsilon - Q(\mathbf x) \cdot u > 0].$$
Denote $p(\mathbf x) \cdot \epsilon - Q(\mathbf x) \cdot u$ by $K_u(\mathbf x)$ and define $\mu = \mathbb E[K_u(\mathbf x)]$.
When $\mu<0$~\footnote{The condition that $\mu < 0$ can be guaranteed by the values of $u$ we have. This will be explained in more details in Section~\ref{sec:div}.} and for analysis order $n=2$, the relaxation proceeds as
$$   \mathbb P[K_u(\mathbf x) \ge 0] = \mathbb P[K_u(\mathbf x) - \mu \ge -\mu] \le \mathbf [(K_u(\mathbf x) - \mu)^2 \ge \mu^2] \le {\mathbb E[(K_u(\mathbf x) - \mu)^2]}/{\mu^2},$$
where the final inequality follows by an application of Markov's inequality. The validity of the preceding steps is established in more detail in Theorem~\ref{thm:div-flag}.

Therefore, it suffices to have a threshold $u$ where ${\mathbb E[(K_u(\mathbf x) - \mu)^2]}/{\mu^2} \le 1 - c$. 
Plugging in the previous definitions and values, we have $\mu = 0.75\epsilon - 5u$, the inequality can be simplified into ${(u^2/3 - 11\epsilon^2/48)}/{(9\epsilon^2/16 - 15u\epsilon/2 + 25u^2)} < 0.01.$
The largest value of $u$ satisfying this inequality would be $u^* \approx 7.68\times 10^{-8}$.

By providing the second-order error term  $|R_2(\mathbf x, \mathbf e, \mathbf d)|$ along with the variable ranges to GELPIA, we obtain a deterministic upper bound: $|R_2(\mathbf x, \mathbf e, \mathbf d)| \le 1.47\times 10^{-14} \triangleq\delta$.
Consequently, for Example~\ref{ex:div}, we conclude that the total round-off error satisfies $\mathbb P[\err(g, \mathbf x) \ge 7.68\times 10^{-8} + \delta] \le 0.01$.

\section{Threshold Synthesis Algorithms}
\label{sec:algorithm}

In this section, we present our algorithms to address the probabilistic floating-point analysis problem (cf. the end of Section~\ref{sec:prelim}). 
A sequence of algorithms are presented, each building on its predecessor to either incorporate a wider range of inputs, or improve the precision of the analysis.
We first demonstrate the positive-negative (PN) decomposition technique (Section~\ref{sec:posneg}), which relaxes the absolute values and eliminates explicit occurrences of error variables. 
We then describe two algorithms for handling division-free input functions (functions that do not involve division by non-constant expressions) in Section~\ref{sec:poly}: the \emph{Naive Markov} (NM) algorithm and the \emph{Central-Moment-Based} (CMB) algorithm. 
The NM algorithm (Section~\ref{sec:polynm}) performs a direct application of Markov's inequality following the PN decomposition, offering a lightweight analysis.
The CMB algorithm (Section~\ref{sec:polycmb}) improves upon the NM algorithm by applying Markov's inequality with the central moment. 
This algorithm is extended to handle fractional expressions involving non-constant denominators (Section~\ref{sec:div}).
Moreover, the CMB algorithm enables further refinement via range partition (Section~\ref{sec:partition}).

We introduce some terminologies for this section. A \emph{term} is a product of the form $c x_{j_1}^{a_1}\dots x_{j_m}^{a_m}$, where $c$ is a constant that acts as the coefficient, the $a_i$'s are non-negative integers that act as the exponents, and each $x_{j_i}$ is a variable in the vector $\mathbf{x}$. 
Given a term $c x_{j_1}^{a_1}\dots x_{j_m}^{a_m}$, the product $x_{j_1}^{a_1}\dots x_{j_m}^{a_m}$ is called the \emph{monomial part} of the term, and each $x_{j_i}^{a_i}$ in the monomial is called a \emph{factor}. 
By convention, terms have non-zero coefficients.
An arithmetic expression over the variables $\mathbf{x}$ is in \emph{polynomial form} if it is a finite sum $\sum_i t_i$ 
where each $t_i$ is a term.

\subsection{Positive-Negative Decomposition}
\label{sec:posneg}

Given $k$ polynomial functions 
$h_1(\mathbf x), \ldots, h_k(\mathbf x)$, the objective of the PN decomposition algorithm is to derive a relaxation of the sum $\sum_{i=1}^k|h_i(\mathbf x)e_i|$ that eliminates both the absolute value operations and the explicit appearance of error variables $e_i$ (the relative error terms defined in Section~\ref{sec:floatingpoint}). 
This relaxation of the summation is used in subsequent steps; a precise computation would otherwise necessitate expensive computer algebra for identifying the individual positive and negative regions of each absolute value's argument.

The central idea of our PN decomposition algorithm is to introduce a pair of fresh variables $x^+,x^-$ for each existing variable $x$, defined by $x^+ :=\max\{x, 0\}$ and $x^-:=\max\{-x, 0\}$ (here we abuse the notation so that $x,x^+,x^-$ also refer to the random variables they indicate).
Two obvious properties follow:
(a) $x = x^+ - x^-$; (b) $x^+ \ge 0$ and $x^-\ge 0$.
We refer to $x$ as the \emph{original variable}, $x^+$ as the \emph{positive component}, and $x^-$ the \emph{negative component}.
In addition, the identity $x^+\cdot x^-=0$ holds for all $x$, which serves as a useful simplification property throughout the algorithm.

At a high level, the PN decomposition targets variables $x$ whose distribution spans both positive and negative values, i.e., variables that are not always non-negative or non-positive.
For such variables, odd powers $x^n$ are rewritten using the identity $x^n = x^{n-1}(x^+-x^-)$, which can be further simplified to $(x^+)^n - (x^-)^n$, leveraging the non-negativity of $x^{n-1}$. 
For $x$ always non-negative, no transformation is needed; for $x$ always non-positive, $x^n$ is replaced by $(-x^-)^n$.

The pseudocode of our PN decomposition algorithm is given in Algorithm~\ref{alg:posneg}. 
The algorithm takes as input a list of $k$ polynomial functions $h_1(\mathbf x), \ldots, h_k(\mathbf x)$, and outputs a polynomial expression $p(\mathbf x)$ that satisfies $\sum_{i=1}^k|h_i(\mathbf x) e_i| \le p(\mathbf x) \cdot \epsilon$ for all possible values of $\mathbf x$ and $\mathbf e$. 
Note that the polynomial $p$ explicitly depends on the positive and negative components of the variables in the vector $\mathbf x$, but since they depend on $\mathbf x$, we use $p(\mathbf x)$ as our notation for simplicity.

Algorithm~\ref{alg:posneg} first accepts the input polynomials and initializes the output $p(\mathbf x)$.
Then, the for loop at line 2 processes each $h_i$. 
During each loop iteration, lines 4 -- 10 replace each factor $x^a$ in $h_i$ by its PN decomposition 
as mentioned previously.
Lines 12 -- 20 collect the positive parts of the $h_i$'s into $h_i^+$ and the negative parts into $h_i^-$. 
Finally, line 21 takes the over-approximation of the original $|h_i|$ as $h_i^+ + h_i^-$ and adds it into the current over-approximation $p(\mathbf{x})$. 
After the loop, the returned polynomial $p(\mathbf{x})$ over-approximates $\sum_{i=1}^k|h_i(\mathbf x)e_i|$. 
The correctness is given in Theorem~\ref{thm:posnegcorrect}.

\begin{algorithm}[ht]
\caption{Positive-Negative Decomposition}
\label{alg:posneg}
\begin{algorithmic}[1]
\Require A list of $k$ polynomial functions $h_1(\mathbf x), \ldots, h_k(\mathbf x)$
\Ensure A polynomial expression $p(\mathbf x)$ such that $\forall \mathbf x, \mathbf e, \sum_{i=1}^k|h_i(\mathbf x)e_i| \le p(\mathbf x) \cdot \epsilon$
\State $p(\mathbf x) \gets 0$
\For {$i \gets 1 \text{ to } k$}
    \State Expand $h_i(\mathbf x)$ into polynomial form
    \For{each factor $x_j^a$ in each term in $h_i(\mathbf x)$}
        \If{$a$ is odd and $x_j$'s range spans over positive and negative values}
            \State Replace $x_j^a$ by $(x_j^+)^a - (x_j^-)^a$ 
            % \Comment $x_j^{a-1}$ can be omitted if $a = 1$
        \EndIf
        \If{$a$ is odd and $x_j$ is always non-positive}
            \State Replace $x_j^a$ by $-(x_j^-)^{a}$
        \EndIf
    \EndFor
    \State Expand $h_i(\mathbf x)$ again into polynomial form
    \State Initialize $h_i^+(\mathbf x) \gets 0; \quad h_i^-(\mathbf x) \gets 0$
    \For{each term $t$ in $h_i(\mathbf x)$}
        \If{term $t$ has positive coefficient}
            \State $h_i^+(\mathbf x) \gets h_i^+(\mathbf x) + t$
        \Else 
            \State $h_i^-(\mathbf x) \gets h_i^-(\mathbf x) - t$
        \EndIf
    \EndFor
    \State $p(\mathbf x) \gets p(\mathbf x) + h_i^+(\mathbf x) + h_i^-(\mathbf x)$
\EndFor
\State \textbf{Return} $p(\mathbf x)$
\end{algorithmic}
\end{algorithm}

\begin{theorem}
Let $p(\mathbf x)$ be the output of Algorithm \ref{alg:posneg}. Then for all possible values of $\mathbf x$ 
and $\mathbf e$,
$\sum_{i=1}^k|h_i(\mathbf x) e_i| \le p(\mathbf x) \cdot \epsilon$ holds.
\label{thm:posnegcorrect}
\end{theorem}

{\sc Proof Sketch} (full proof in Appendix~\ref{sec:fullproof}).
After the execution reaches line 12, $h_i(\mathbf x)$ remains equivalent to the original input expression, as substituting $x_j^a$ with $x_j^{a-1}(x_j^+ - x_j^-)$ preserves the equality. 
Besides, in each term in $h_i(\mathbf x)$, all factors ($x_j^+$, $x_j^-$, or even powers) are non-negative, ensuring non-negative monomial parts, and thus, all terms in $h_i(\mathbf x)$ have the same sign as its coefficient.

The decomposition of $h_i(\mathbf x)$ into $h_i^+(\mathbf x)$ and $h_i^-(\mathbf x)$ after the for loop at line 14 can be proved to have two properties similar to the decomposition of individual variables: (a) both $h_i^+(\mathbf x)$ and $h_i^-(\mathbf x)$ are non-negative; (b) $h_i(\mathbf x) = h_i^+(\mathbf x) - h_i^-(\mathbf x)$.

Using the triangular inequality, $|h_i(\mathbf x)e_i| = |h_i^+(\mathbf x)e_i - h_i^-(\mathbf x)e_i| \le |h_i^+(\mathbf x)e_i| + |h_i^-(\mathbf x)e_i|$. Since $h_i^+(\mathbf x)$ and $h_i^-(\mathbf x)$ are non-negative, this simplifies to $|h_i(\mathbf x)e_i| \le (h_i^+(\mathbf x) + h_i^-(\mathbf x))\epsilon$.
Summing over all $i$, $F(\mathbf x, \mathbf e) = \sum_{i=1}^k|h_i(\mathbf x)e_i| \le \sum_{i=1}^k(h_i^+(\mathbf x) + h_i^-(\mathbf x))\epsilon = p(\mathbf x) \cdot \epsilon$, as required. \qed

\full{\begin{proof}
We first establish the following two facts after the execution reaches line 11:
(a) $h_i(\mathbf x)$ is equivalent to the original input expression; 
(b) each term in $h_i(\mathbf x)$ has the same sign as its coefficient (the constant factor of the term).
Fact (a) is evident because the only two possible modifications to $h_i(\mathbf x)$ are replacing $x_j^a$ by either $(x_j^+)^a - (x_j^-)^a$ or $(-x_j^-)^a$.
Given that $x_j = x_j^+ - x_j^-$ and $x_j^{a-1}$ is non-negative for odd values of $a$, $x_j^{a} = x_j^{a-1}x_j^+ - x_j^{a-1}x_j^- = (x_j^+)^a - (x_j^-)^a$, and thus the first possible replacement preserves the value.
When $x_j$ is always non-positive, $x_j = x_j^-$ holds, and thus the second possible replacement is an equivalent transformation.
Then we examine the monomial part of each term. Note that all factors fall into one of the three categories: an even power of the original variable, the positive component, or the negative component. 
Each of the three is non-negative, and thus their product can be proved non-negative.
Thus, the monomial part of all terms in $h_i(\mathbf x)$ is non-negative, and fact (b) immediately follows.

Then we analyze the values of $h_i^+(\mathbf x)$ and $h_i^-(\mathbf x)$ after the execution of line 22, and show that it has the two properties similar to the decomposition of variables:
\begin{itemize}
    \item Both $h_i^+(\mathbf x)$ and $h_i^-(\mathbf x)$ are non-negative.
    When a term $t$ has positive coefficient, indicating that $t$ is non-negative (by fact (b)), it is added to $h_i^+$, thereby preserving the non-negativity of $h_i^+$.
    On the other hand, when $t$ has negative coefficient, the term $t$ is non-positive. Subtracting such a term form $h_i^-$ ensures that its non-negativity is maintained.
    \item $h_i(\mathbf x) = h_i^+(\mathbf x) - h_i^-(\mathbf x)$. Each term in $h_i(\mathbf x)$ is assigned either to $h_i^+(\mathbf x)$ as it appears, or to $h_i^-(\mathbf x)$ as its negation. Thus, subtracting $h_i^-(\mathbf x)$ from $h_i^+(\mathbf x)$ reconstructs the original value of $h_i(\mathbf x)$, thereby verifying the equality.
\end{itemize}

Therefore, we have the following relaxation:
$$ \begin{aligned}
|h_i(\mathbf x)e_i| &= |h_i^+(\mathbf x)e_i  - h_i^-(\mathbf x)e_i| 
\le |h_i^+(\mathbf x)e_i| + |h_i^-(\mathbf x)e_i|  \\
&= (h_i^+(\mathbf x) + h_i^-(\mathbf x))|e_i|\text{}
\le (h_i^+(\mathbf x) + h_i^-(\mathbf x))\epsilon.
\end{aligned}$$
The second inequality is an application of the triangular inequality. 
The third equality is because $h_i^+(\mathbf x)$ and $h_i^-(\mathbf x)$ are non-negative.
At the end of the algorithm, $p(\mathbf x) = \sum_{i=1}^k(h_i^+(\mathbf x) + h_i^-(\mathbf x))$.
Thus, for the entire first-order term we have
$$\sum_{i=1}^k|h_i(\mathbf x)e_i| \le \sum_{i=1}^k(h_i^+(\mathbf x) + h_i^-(\mathbf x))\epsilon = p(\mathbf x) \cdot \epsilon.$$
\end{proof}}

\begin{remark}
We choose to restrict the input functions to the PN decomposition algorithm to division-free expressions only. The algorithm expands each of the input functions into polynomial form, and subsequent steps (such as selectively transforming only odd-degree factors) operate on the granularity of individual terms and do not naturally extend to more general algebraic structures.
\end{remark}

% \begin{remark}
% For a single term or a sum of terms, computing the expectation with PN decomposition produces the same result as simply applying the triangular inequality and factorizing the expectation under independence. 
% However, the simple approach offers no direct path to further refinement, while PN decomposition enables direct improvement as it is a much finer representation.
% This improvement plays a bigger role in the interaction with range partition (Section~\ref{sec:partition}).
% \end{remark}

\subsection{Algorithms for Division-Free Expressions}
\label{sec:poly}

By feeding polynomials $h_i(\mathbf x) = \frac{\partial \tilde f}{\partial e_i}(\mathbf x, \mathbf 0, \mathbf 0)$ ($1\le i \le k$) from a division-free first-order error term  $F(\mathbf x, \mathbf e) = \sum_{i=1}^k\left|\frac{\partial \tilde f}{\partial e_i}(\mathbf x, \mathbf 0, \mathbf 0) e_i\right|$ as input to Algorithm~\ref{alg:posneg}, one obtains a sound polynomial over-approximation $p(\mathbf x)$ for 
$F(\mathbf x, \mathbf e) = \sum_{i=1}^k\left|\frac{\partial \tilde f}{\partial e_i}(\mathbf x, \mathbf 0, \mathbf 0) e_i\right|$.
By using Algorithm~\ref{alg:posneg} as a subroutine, we further develop two algorithms below 
that solve the probabilistic round-off analysis problem by synthesizing a threshold that fulfills $\mathbb P[\err(f, \mathbf x) < U_c] > c$.

\subsubsection{Naive Markov (NM) Algorithm}
\label{sec:polynm}

The Naive Markov Algorithm, presented in Algorithm~\ref{alg:nm}, computes an error threshold by utilizing an over-approximation produced by Algorithm~\ref{alg:posneg} and a direct application of Markov's inequality in a higher-moment setting.

Given a confidence level $c$, an analysis order $n$~\footnote{The analysis order $n$ is a hyper-parameter. Analyses with a larger $n$ usually generate more accurate thresholds at the cost of increased analysis time. The parameter $n$ in Section~\ref{sec:polycmb} and \ref{sec:div} serves the same role.}, and a division-free arithmetic expression $f(\mathbf{x})$ to be analyzed, the algorithm first computes the partial derivatives $\frac{\partial \tilde f}{\partial e_i}(\mathbf x, \mathbf 0, \mathbf 0)$ (lines 1 -- 2) and passes it to Algorithm~\ref{alg:posneg} to generate an over-approximation $p(\mathbf x)$ of $F(\mathbf x, \mathbf e)$ (line 3). 
Then, the algorithm gets a probabilistic threshold $U_c^{(1)}$ for first-order error term $F(\mathbf x, \mathbf e)$ by applying Markov's inequality with analysis order $n$ over $p(\mathbf x)$ (line 4, see also Theorem~\ref{thm:markov}).
Next, the algorithm employs a sound global optimization tool (e.g. GELPIA) to get a deterministic upper bound $U^{(2)}$ for the second-order error term (line 6). 
Finally, the sum of the two thresholds is returned as the overall threshold $U_c^*$. 
The soundness of the algorithm is given in Theorem~\ref{thm:nmcorrect}. 

\begin{algorithm}[ht]
\caption{The Naive Markov Algorithm}
\label{alg:nm}
\begin{algorithmic}[1]
\Require A division-free expression $f(\mathbf x)$, distribution of the input variables, confidence level $c$, unit round-off $\epsilon$, smallest representable number $\delta$, analysis order $n$
\Ensure A probabilistic upper bound $U_c^*$ for the round-off error $\err(f, \mathbf x)$
\State Compute the floating-point version of the expression $\tilde f(\mathbf x, \mathbf e, \mathbf d)$. 
\State Compute the partial derivatives $h_i(\mathbf x) = \frac{\partial \tilde f}{\partial e_i}(\mathbf x, \mathbf 0, \mathbf 0)$.
\State Input the partial derivatives $h_i(\mathbf x)$ into Algorithm~\ref{alg:posneg} and get the expression for $p(\mathbf x)$ 
\State $U_c^{(1)} \gets \epsilon \cdot \sqrt[n]{\mathbb E[(p(\mathbf x))^n]/(1-c)}$
\State Compute the expression for second-order error term $|R_2(\mathbf x, \mathbf e, \mathbf d)|$.
\State Get a deterministic upper bound $U^{(2)}$ for $|R_2(\mathbf x, \mathbf e, \mathbf d)|$ using a sound global optimization tool
\State \textbf{Return} $U_c^* \gets U_c^{(1)} + U^{(2)}$
\end{algorithmic}
\end{algorithm}

\begin{theorem}
Let $U_c^*$ be the output of Algorithm~\ref{alg:nm}, then $\mathbb P[\err(f, \mathbf x) \ge U_c^*] \le 1-c$.
\label{thm:nmcorrect}
\end{theorem}

\begin{proof}
By Theorem~\ref{thm:posnegcorrect}, we have 
$F(\mathbf x, \mathbf e) = \sum_{i=1}^k\left|\frac{\partial\tilde f}{\partial e_i}(\mathbf x, \mathbf 0, \mathbf 0)e_i\right| \le p(\mathbf x) \cdot \epsilon$.
Thus, for any value of $U_c^{(1)}\in \mathbb R$, we can relax the first-order error threshold violation probability
$\mathbb P[F(\mathbf x, \mathbf e) \ge U_c^{(1)}]  \le \mathbb P[p(\mathbf x) \cdot \epsilon \ge U_c^{(1)}] = \mathbb P[(p(\mathbf x))^n \cdot \epsilon ^n \ge (U_c^{(1)})^n]$.
Since $F(\mathbf x, \mathbf e)$ is non-negative, its over-approximation $(p(\mathbf x))^n \cdot \epsilon^n$ is also non-negative, and thus Markov's inequality is applicable. Therefore,
$\mathbb P[(p(\mathbf x))^n \cdot \epsilon ^n \ge (U_c^{(1)})^n] \le {\mathbb E[(p(\mathbf x))^n] \cdot \epsilon^n }/{(U_c^{(1)})^n}$.
Subsequently, we substitute $U_c^{(1)} = \epsilon \cdot \sqrt[n]{\mathbb E[(p(\mathbf x))^n]/(1-c)}$ in the formula to get:
${\mathbb E[(p(\mathbf x))^n] \cdot \epsilon^n }/{(U_c^{(1)})^n} = {\mathbb E[(p(\mathbf x))^n] \cdot \epsilon^n}/{(\mathbb E[(p(\mathbf x))^n] \cdot \epsilon^n / (1-c))} = 1-c$.
Combining the derivations above, we arrive at the following inequality
$\mathbb P[F(\mathbf x, \mathbf e) \ge U_c^{(1)}]  \le 1-c$.
Since the global optimization tool gives a sound and deterministic upper bound of the second-order error, $|R_2(\mathbf x, \mathbf e, \mathbf d)| \le U^{(2)}$ holds. Therefore,
$\mathbb P [\err (f, \mathbf x) \ge U_c^*] =\mathbb P[F(\mathbf x, \mathbf e) + |R_2(\mathbf x, \mathbf e, \mathbf d)| \ge U_c^{(1)} + U^{(2)}]\le \mathbb P[F(\mathbf x, \mathbf e) \ge U_c^{(1)}]  \le 1-c$.
\end{proof}

\subsubsection{Central-Moment-Based (CMB) Algorithm}
\label{sec:polycmb}

The NM algorithm determines an optimal threshold $U_c^*$ for a specified confidence level $c$.
In contrast, the CMB algorithm adopts an inverse approach -- for a fixed threshold $u$, it computes an upper bound $\flag_u$ on the first-order violation probability $\mathbb P[F(\mathbf x, \mathbf e) \ge u]$.
The CMB algorithm then seeks a close-to-minimum value of $u$ such that $\flag_u \le 1-c$, thereby ensuring the desired probabilistic guarantee. 
To achieve this, the CMB algorithm employs a binary search over an interval $(\ell, r)$, whose endpoints will be formally introduced later. 
It is worth noting that $\flag_u$ is monotonically decreasing with respect to $u$ once $u$ exceeds the lower endpoint $\ell$~\footnote{A proof of this monotonicity property is provided in Appendix~\ref{sec:flagmono}.}, which justifies the correctness of the binary search procedure.

We first introduce the definition of $\flag_u$ and prove the validity of $\flag_u$ as an upper bound in Theorem~\ref{thm:cmb-flag}, and then present the complete CMB algorithm in Algorithm~\ref{alg:cmb}. 
The correctness for Algorithm~\ref{alg:cmb} is proved in Theorem~\ref{thm:cmbcorrect}.

\begin{theorem}
Let $u$ be a fixed threshold.
Let $p(\mathbf x)$ be a function satisfying $F(\mathbf x, \mathbf e) \le p(\mathbf x) \cdot \epsilon$ for all possible values of $\mathbf x$ and $\mathbf e$.
Let $K_u(\mathbf x)\triangleq p(\mathbf x)\cdot \epsilon - u$, and $\mu \triangleq \mathbb E[K_u(\mathbf x)]$. 
Suppose that $\mu < 0$.
For an even analysis order $n$, if we define
$\flag_u \triangleq {\mathbb E[(K_u(\mathbf x) - \mu)^n]}/{\mu^n}$,
then the first-order violation probability admits the following upper bound: $\mathbb P[F(\mathbf x, \mathbf e) \ge u] \le \flag_u$.
\label{thm:cmb-flag}
\end{theorem}

\begin{proof}
By assumption, we have $F(\mathbf x, \mathbf e) \le p(\mathbf x) \cdot \epsilon$. 
Therefore, for a fixed value of $u$,
$\mathbb P[F(\mathbf x, e)\ge u] \le \mathbb P[p(\mathbf x)\cdot \epsilon \ge u] = \mathbb P[K_u(\mathbf x) \ge 0]$. 
Denote the event $K_u(\mathbf x) \ge 0 $ by $E_1$, and the event $|K_u(\mathbf x) - \mu| \ge |\mu|$ by $E_2$. 
We claim that the occurrence of event $E_1$ implies the occurrence of event $E_2$, i.e., $E_1\subseteq E_2$, based on the following reasoning:
whenever $K_u(\mathbf x) \ge0$ ($E_1$ occurs) and given that $\mu<0$, it holds that $K_u(\mathbf x) - \mu \ge0$, and therefore,
$|K_u(\mathbf x) - \mu| = K_u(\mathbf x) - \mu \ge -\mu = |\mu|$,
which satisfies the condition defining event $E_2$.
Hence, $\mathbb P[K_u(\mathbf x) \ge 0] =\mathbb P[E_1] \le \mathbb P[E_2] = \mathbb P[|K_u(\mathbf x) - \mu| \ge \mu] = \mathbb P[(K_u(\mathbf x) - \mu)^n \ge \mu^n]$. 
Then we apply Markov's inequality (applicable since $n$ is even and thus $\mu^n > 0$) and get
$\mathbb P[(K_u(\mathbf x) - \mu)^n \ge \mu^n] \le {\mathbb E[(K_u(\mathbf x) - \mu)^n]}/{\mu^n} = \flag_u$.
Combining the derivations above would lead us to the desired conclusion.
\end{proof}

\begin{algorithm}[ht]
\caption{The Central-Moment-Based Algorithm}
\label{alg:cmb}
\begin{algorithmic}[1]
\Require A division-free expression $f(\mathbf x)$, distribution of the input variables, confidence level $c$, unit round-off $\epsilon$, smallest representable number $\delta$, and even analysis order $n$
\Ensure A probabilistic upper bound $U_c^*$ for the round-off error $\err(f, \mathbf x)$
\State Compute the floating-point version of the expression $\tilde f(\mathbf x, \mathbf e, \mathbf d)$.
\State Compute the partial derivatives $h_i(\mathbf x) = \frac{\partial \tilde f}{\partial e_i}(\mathbf x, \mathbf 0, \mathbf 0)$.
\State Input the partial derivatives $h_i(\mathbf x)$ into Algorithm~\ref{alg:posneg} and get the expression for $p(\mathbf x)$ 
\State $\ell\gets \mathbb E[p(\mathbf x)] \cdot \epsilon$
\State $r\gets\epsilon \cdot \sqrt[n]{\mathbb E[(p(\mathbf x))^n]/(1-c)}$
\If {$\flag_{r} < 1-c$}
    \State Use binary search to find the close-to-minimum $U^{(1)}_c\in (\ell, r)$ such that $\flag_{U^{(1)}_c} \le 1 - c$
\Else 
    \State $U_c^{(1)} \gets r$
\EndIf
\State Get a deterministic upper bound $U^{(2)}$ for $|R_2(\mathbf x, \mathbf e, \mathbf d)|$ using a sound global optimization tool
\State \textbf{Return} $U_c^* \gets U_c^{(1)} + U^{(2)}$
\end{algorithmic}
\end{algorithm}

\begin{theorem}
Let $U_c^*$ be the output of Algorithm~\ref{alg:cmb}, then $\mathbb P[\err(f, \mathbf x)\ge U_c^*] \le 1-c$.
\label{thm:cmbcorrect}
\end{theorem}

\begin{proof}

As established in the proof of Theorem~\ref{thm:nmcorrect}, $p(\mathbf x)$ computed in line 3 should satisfy 
$F(\mathbf x, \mathbf e) \le p(\mathbf x) \cdot \epsilon$.
Note that the initial value of $r$ coincides with the value of $U^{(1)}_c$ produced by Algorithm~\ref{alg:nm}. 
To verify that $U^{(1)}_c$ is a valid threshold, we take a closer look at the two branches of the conditional statement in line 6:
\begin{itemize}
    \item When $\flag_r < 1 - c$, the \texttt{then} branch is taken. 
    In this case, the value of $\ell = \mathbb E[p(\mathbf x)] \cdot \epsilon$ is chosen so that for any candidate threshold $u \in (\ell, r)$, we have $\mathbb E[K_u(\mathbf x)] = \mathbb E[p(\mathbf x) \cdot \epsilon - u] = \ell - u \le 0$. 
    Consequently, the conditions of Theorem~\ref{thm:cmb-flag} is satisfied, and we conclude that 
    $\mathbb P[F(\mathbf x, \mathbf e) \ge U_c^{(1)}] \le \flag_{U_c^{(1)}} \le 1 -c .$
    \item When the condition $\flag_r < 1 - c$ is not satisfied, the algorithm directly sets $U_c^{(1)} := r$, which, by construction, is identical to the first-order threshold computed by Algorithm~\ref{alg:nm}. By Theorem~\ref{thm:nmcorrect}, it follows that $\mathbb P[F(\mathbf x, \mathbf e) \ge U_c^{(1)}] \le 1-c$.
\end{itemize}
In both cases, the value assigned to $U_c^{(1)}$ ensures that $\mathbb P[F(\mathbf x, \mathbf e) \ge U_c^{(1)}] \le 1 -c$, making it a valid threshold on the first-order error term.
Since the second-order error term is handled identically to Algorithm~\ref{alg:nm}, it follows that the final threshold $U_c^*$, including both first- and second-order error, satisfies $\mathbb P[\err(f, \mathbf x) \ge U_c^*] \le 1 - c$.
\end{proof}

\subsection{Algorithm for Fractional Expressions}
\label{sec:div}

In Section~\ref{sec:poly}, two algorithms are introduced to address division-free expressions. 
However, when dealing with expressions with division by non-constant expressions, the partial derivatives $\frac{\partial \tilde f}{\partial e_i}(\mathbf x, \mathbf 0, \mathbf 0)$ generally contain divisions, rendering them incompatible with the PN decomposition algorithm, which assumes division-free inputs.
In this subsection, we explain how the CMB algorithm (Section~\ref{sec:polycmb}) can be extended to handle fractional expressions in the form of $N(\mathbf x)/Q(\mathbf x)$ where both the numerator $N(\mathbf x)$ and the denominator $Q(\mathbf x)$ are division-free. 

Similar to Section~\ref{sec:polycmb}, we first state and prove the applicability of PN decomposition to our modified inputs in Lemma~\ref{lem:fractionalpn}, then prove a theorem for bounding the first-order violation probability of a fixed threshold $u$, denoted by $\flag_u$, in Theorem~\ref{thm:div-flag}.
Finally we introduce the complete algorithm in Algorithm ~\ref{alg:div} and prove its soundness in Theorem~\ref{thm:divcorrect}.

\begin{lemma}\label{lem:fractionalpn}
Suppose the input function $f(\mathbf x) = N(\mathbf x) / Q(\mathbf x)$ is a fractional expression, where both the numerator $N(\mathbf x)$ and the denominator $Q(\mathbf x)$ are division-free, and $Q(\mathbf x) \neq 0$ for all possible values of $\mathbf x$ in the specified distribution. Let $h_i(\mathbf x) = \frac{\partial \tilde f}{\partial e_i}(\mathbf x, \mathbf 0, \mathbf 0) \cdot (Q(\mathbf x))^2$, then $h_i(\mathbf x)$ can be simplified to a polynomial expression.
\end{lemma}

{\sc Proof Sketch} (full proof in Appendix~\ref{sec:fullproof}).
Given the structure of the input function $f(\mathbf x)$, the floating-point model should have the form
$\tilde f(\mathbf x, \mathbf e, \mathbf d) = \frac{\tilde N(\mathbf x, \mathbf e, \mathbf d)}{\tilde Q(\mathbf x, \mathbf e, \mathbf d)}  (1 + e_k) + d_k$, 
\full{$\tilde f(\mathbf x, \mathbf e, \mathbf d) = \tilde N(\mathbf x, \mathbf e, \mathbf d) / \tilde Q(\mathbf x, \mathbf e, \mathbf d) \cdot (1 + e_k) + d_k$, }
where $k$ is the total number of arithmetic operations in $f(\mathbf x)$ and $e_k, d_k$ are the error variables for the outmost division. Each $e_i$ should appear exactly once. We consider the three possibilities where $e_i$ appears, and describe how the corresponding $h_i(\mathbf x)$ should be reducible to polynomials: 
(a) $e_i$ appears in $\tilde N(\mathbf x, \mathbf e, \mathbf d)$: differentiation yields a factor $1 / Q(\mathbf x)$; multiplying by $(Q(\mathbf x))^2$ produces a division-free term; 
(b) $e_i$ appears in $\tilde Q(\mathbf x, \mathbf e, \mathbf d)$: differentiation gives $1/(Q(\mathbf x))^2$, which is cancelled by multiplying $(Q(\mathbf x))^2$; 
(c) $e_i = e_k$: differentiation gives $N(\mathbf x) / Q(\mathbf x)$, and again multiplying $(Q(\mathbf x))^2$ yields $N(\mathbf x) Q(\mathbf x)$. 
In all cases, the result involves only sums and products of division-free expressions, so each $h_i(\mathbf x)$ reduces to a polynomial. 
\qed

\full{
\begin{proof}
Given that $f(\mathbf x)$ has the form $N(\mathbf x) / Q(\mathbf x)$, its floating-point model should observe the following structure: 
$\tilde f(\mathbf x, \mathbf e, \mathbf d) = {\tilde N(\mathbf x, \mathbf e, \mathbf d)}/{\tilde Q(\mathbf x, \mathbf e, \mathbf d)} \cdot (1 + e_k) + d_k$, 
where $k$ is the total number of arithmetic operations in $f(\mathbf x)$ and {$e_k, d_k$ are the error variables for the outmost division}.
We observe the fact that each relative error term $e_i$ appears exactly once in the symbolic expression of $\tilde f(\mathbf x, \mathbf e, \mathbf d)$. More specifically, each $e_i$ contributes to one of the following components: (i) the floating-point model of the numerator $\tilde N(\mathbf x, \mathbf e, \mathbf d)$; (ii) the floating-point model of the denominator $\tilde Q(\mathbf x, \mathbf e, \mathbf d)$; or (iii) the error term $e_k$ related to division $N(\mathbf x) / Q(\mathbf x)$.
To demonstrate that for any index $i$, $\frac{\partial \tilde f}{\partial e_i}(\mathbf x, \mathbf 0, \mathbf 0) \cdot (Q(\mathbf x))^2$ can be simplified into polynomial form, we analyze the three cases. 
Throughout, we make frequent use of the identity $\tilde f(\mathbf x, \mathbf 0, \mathbf 0) = f(\mathbf x)$.
\begin{enumerate}[(i)]
    \item When $e_i$ appears in the numerator $\tilde N(\mathbf x, \mathbf e, \mathbf d)$, we have $\frac{\partial \tilde f}{\partial e_i}(\mathbf x, \mathbf 0, \mathbf 0) \cdot (Q(\mathbf x))^2 = \frac{1}{Q(\mathbf x)} \cdot \frac{\partial \tilde N}{\partial e_i}(\mathbf x, \mathbf 0, \mathbf 0)  \cdot (Q(\mathbf x))^2 = \frac{\partial \tilde N}{\partial e_i}(\mathbf x, \mathbf 0, \mathbf 0)  \cdot Q(\mathbf x)$ which is division-free since $N(\mathbf x)$ and $Q(\mathbf x)$ are both division-free.
    \item When $e_i$ appears in the denominator $\tilde Q(\mathbf x, \mathbf e, \mathbf d)$, we have 
    $\frac{\partial \tilde f}{\partial e_i}(\mathbf x, \mathbf 0, \mathbf 0) \cdot (Q(\mathbf x))^2 =
    \frac{- N(\mathbf x)}{(Q(\mathbf x))^2} \cdot \frac{\partial \tilde Q}{\partial e_i}(\mathbf x, \mathbf 0, \mathbf 0) \cdot (Q(\mathbf x))^2 = - N(\mathbf x) \cdot \frac{\partial \tilde Q}{\partial e_i}(\mathbf x, \mathbf 0, \mathbf 0)$, which is also division-free.
    \item When $e_i$ is just $e_k$, then $\frac{\partial \tilde f}{\partial e_i}(\mathbf x, \mathbf 0, \mathbf 0) \cdot (Q(\mathbf x))^2  = \frac{N(\mathbf x)}{Q(\mathbf x)} \cdot (Q(\mathbf x))^2 = N(\mathbf x) \cdot Q(\mathbf x)$, clearly division-free.
\end{enumerate}
Therefore, in each of the three possible cases, $h_i(\mathbf x) = \frac{\partial \tilde f}{\partial e_i}(\mathbf x, \mathbf 0, \mathbf 0) \cdot (Q(\mathbf x))^2$ is division-free, and is thus reducible to a polynomial expression.
\end{proof}
}

\begin{theorem}\label{thm:fractionalflagu}
Let $u$ be a fixed threshold. Suppose the input function $f(\mathbf x) = N(\mathbf x) / Q(\mathbf x)$ where both the numerator $N(\mathbf x)$ and the denominator $Q(\mathbf x)$ are division-free, and $Q(\mathbf x) \neq 0$ for all possible values of $\mathbf x$ in the specified distributions.
Let $p(\mathbf x)$ be a polynomial ~\footnote{It will be described after the proof of Theorem~\ref{thm:div-flag} how to derive such a polynomial function $p(\mathbf x)$.} satisfying 
$\sum_{i=1}^k\left| \frac{\partial \tilde f}{\partial e_i}(\mathbf x, \mathbf 0, \mathbf 0)\cdot (Q(\mathbf x))^2 \cdot e_i\right| \le p(\mathbf x) \cdot \epsilon$
for all possible values of $\mathbf x$ and $\mathbf e$. 
Let $K_u(\mathbf x) \triangleq p(\mathbf x) \cdot \epsilon - (Q(\mathbf x))^2 \cdot u$, and $\mu \triangleq \mathbb E[K_u(\mathbf x)]$.
Suppose that $\mu < 0$. 
For an even analysis order $n$, if we define
$\flag_u \triangleq {\mathbb E[(K_u(\mathbf x) - \mu)^n]}/{\mu^n}$,
then the first-order violation probability admits the following upper bound: $\mathbb P[F(\mathbf x, \mathbf e) \ge u] \le \flag_u$.
\label{thm:div-flag}
\end{theorem}

\begin{proof}
We first exploit the property of $p(\mathbf x)$ to relax $F(\mathbf x, \mathbf e)$ and eliminate explicit occurrences of error variables $\mathbf e$:
$F(\mathbf x, \mathbf e) = \sum_{i=1}^k\left| \frac{\partial \tilde f}{\partial e_i}(\mathbf x, \mathbf 0, \mathbf 0)\cdot e_i\right|  = \frac{1}{(Q(\mathbf x))^2}\sum_{i=1}^k\left| \frac{\partial \tilde f}{\partial e_i}(\mathbf x, \mathbf 0, \mathbf 0)\cdot (Q(\mathbf x))^2 \cdot e_i\right| \le {p(\mathbf x) \cdot \epsilon}/{(Q(\mathbf x))^2}$.
Accordingly, the violation probability can be relaxed by
$\mathbb P[F(\mathbf x, \mathbf e) \ge u] \le \mathbb P \left[ {p(\mathbf x) \cdot \epsilon}/{(Q(\mathbf x))^2} \ge u\right] = \mathbb P[p(\mathbf x) \cdot \epsilon - (Q(\mathbf x))^2 \cdot u \ge 0] = \mathbb P[K_u(\mathbf x) \ge 0].$
The rest of the proof is analogous to that in Theorem~\ref{thm:cmb-flag}. 
The event $K_u(\mathbf x) \ge 0$ would imply the event $|K_u(\mathbf x) - \mu| \ge |\mu|$, and therefore,
$\mathbb P[K_u(\mathbf x) \ge 0] \le \mathbb P[|K_u(\mathbf x) - \mu| \ge \mu| = \mathbb P[(K_u(\mathbf x ) - \mu)^n \ge \mu^n] \le {\mathbb E[(K_u(\mathbf x) - \mu)^n]}/{\mu^n} = \flag_u.$
Combining the inequalities, we conclude $\mathbb P[F(\mathbf x, \mathbf e)\ge u] \le \flag_u$.
\end{proof}

We now describe how PN decomposition can be employed to produce the $p(\mathbf x)$ as required in Theorem~\ref{thm:div-flag}. 
As is proved in Lemma~\ref{lem:fractionalpn}, $h_i(\mathbf x) = \frac{\partial \tilde f}{\partial e_i}(\mathbf x, \mathbf 0, \mathbf 0) \cdot (Q(\mathbf x))^2$ can always be simplified to a polynomial expression.
We may thus provide $h_1(\mathbf x), \dots, h_k(\mathbf x)$ as inputs to the PN decomposition algorithm (Algorithm~\ref{alg:posneg}), and by Theorem~\ref{thm:posnegcorrect}, the algorithm returns a function $p(\mathbf x)$ such that
$\sum_{i=1}^k\left| \frac{\partial \tilde f}{\partial e_i}(\mathbf x, \mathbf 0, \mathbf 0)\cdot (Q(\mathbf x))^2 \cdot e_i\right| \le p(\mathbf x) \cdot \epsilon$,
which precisely satisfies the requirement 
%stipulated in 
of
Theorem~\ref{thm:div-flag}.

\begin{algorithm}[ht]
\caption{The Extended Central-Moment-Based Algorithm for Fractional Expressions}
\label{alg:div}
\begin{algorithmic}[1]
\Require A fractional expression $f(\mathbf x) = N(\mathbf x) / Q(\mathbf x)$ (where $N(\mathbf x)$ and $(Q(\mathbf x)$ are division-free), distribution of the input variables, confidence level $c$, unit round-off $\epsilon$, smallest representable number $\delta$, and even analysis order $n$
\Ensure A probabilistic upper bound $U_c^*$ for the round-off error $\err(f, \mathbf x)$
\State Compute the floating-point version of the expression $\tilde f(\mathbf x, \mathbf e, \mathbf d)$.
\State Compute the expressions for $h_i(\mathbf x) = \frac{\partial \tilde f}{\partial e_i}(\mathbf x, \mathbf 0, \mathbf 0) \cdot (Q(\mathbf x))^2$.
\State Input the list of $h_i(\mathbf x)$ into Algorithm~\ref{alg:posneg} and get the expression for $p(\mathbf x)$ 
\State $\ell\gets \mathbb E[p(\mathbf x)] \cdot \epsilon / \mathbb E[(Q(\mathbf x))^2] $
\State $r\gets\ell \times 10^5$ \Comment{$10^5$ is the large multiple mentioned in the description of the algorithm.}
\State Use binary search to find the close-to-minimum $U^{(1)}_c\in (\ell, r)$ such that $\flag_{U^{(1)}_c} \le 1 - c$
\State Get a deterministic upper bound $U^{(2)}$ for $|R_2(\mathbf x, \mathbf e, \mathbf d)|$ using a sound global optimization tool
\State \textbf{Return} $U_c^* \gets U_c^{(1)} + U^{(2)}$
\end{algorithmic}
\end{algorithm}

We adopt a strategy analogous to the CMB algorithm (Section~\ref{sec:polycmb}, {Algorithm~\ref{alg:cmb}}), employing binary search~\footnote{Since the monotonicity of $\flag_u$ in the case of fractional expressions is not formally proved, binary search is a tentative plan here. It currently works across all evaluated benchmarks. As a fallback, if binary search fails to find a solution, we can instead enumerate candidate values of $u$ at each order of magnitude to identify a suitable threshold.} to find the close-to-minimum $U_c^{(1)}$ such that $\flag_{U_c^{(1)}} \le 1 - c$ {w.r.t Theorem~\ref{thm:fractionalflagu}}.
To ensure that $\mu = \mathbb E[p(\mathbf x) \cdot \epsilon - (Q(\mathbf x))^2 \cdot u]$ remains negative throughout the search, we initialize the left endpoint of the binary search to $\mathbb E[p(\mathbf x)] \cdot \epsilon / \mathbb E[(Q(\mathbf x))^2]$, and the right endpoint to a large multiple of the left endpoint. The detailed algorithm is presented in Algorithm~\ref{alg:div}.
% {\bf More explanations here. }

\begin{theorem}
Let $U_c^*$ be the output of Algorithm~\ref{alg:div}, then $\mathbb P[\err(f, \mathbf x) \ge U_c^*] \le 1-c$.
\label{thm:divcorrect}
\end{theorem}

\begin{proof}
We follow a similar line of reasoning as in the proof of Theorem~\ref{thm:cmbcorrect}.
As proved in Lemma~\ref{lem:fractionalpn}, each expression $h_i(\mathbf x)$ computed at line 2 of the algorithm can be simplified to a polynomial expression, and hence qualifies as input to Algorithm~\ref{alg:posneg}. 
Applying Algorithm~\ref{alg:posneg} to the list $h_i(\mathbf x)$ and invoking Theorem~\ref{thm:posnegcorrect}, we obtain an output polynomial $p(\mathbf x)$ such that $\sum_{i=1}^k \left| \frac{\partial \tilde f}{\partial e_i}(\mathbf x, \mathbf 0, \mathbf 0) \cdot (Q(\mathbf x))^2 \cdot e_i \right| \le p(\mathbf x) \cdot \epsilon$. 
Besides, all possible values of $U_c^{(1)}$ in the search space are greater than the left endpoint of our binary search $\ell = \mathbb E[p(\mathbf x)] \cdot \epsilon / \mathbb E[(Q(\mathbf x))^2]$, and hence $\mu = \mathbb E[p(\mathbf x) \cdot \epsilon - (Q(\mathbf x))^2 \cdot U_c^{(1)}] < 0$.
Thus the preconditions of Theorem~\ref{thm:fractionalflagu} are met, and by Theorem~\ref{thm:fractionalflagu}, $\flag_u$ serves as a sound over-approximation of the first-order violation probability.

Therefore, when the binary search at line 6 is able to output a value of $U_c^{(1)}$, it satisfies $\mathbb P[F(\mathbf x, \mathbf e) \ge U_{c}^{(1)}] \le \flag_{U_c^{(1)}} \le 1-c$.
Since the second order error term is handled identically to Algorithm~\ref{alg:nm}, it follows that the final threshold $U_c^*$, including both first- and second-order error, satisfies $\mathbb P[\err (f, \mathbf x) \ge U_c^*] \le 1-c$.
\end{proof}

\begin{remark}
Our approach can be easily extended for general expressions with division, by reducing $\tilde f(\mathbf x, \mathbf e, \mathbf d)$ to a common denominator $\tilde Q(\mathbf x, \mathbf e, \mathbf d)$, and then taking partial derivatives $\frac{\partial \tilde f}{\partial e_i}(\mathbf x, \mathbf 0, \mathbf 0)$. 
Though some $e_i$ may occur in both the numerator and the denominator, in all partial derivatives $\frac{\partial\tilde f}{\partial e_i}(\mathbf x, \mathbf 0, \mathbf 0)$ the denominator divides $(\tilde Q(\mathbf x,\mathbf 0, \mathbf 0))^2$, and thus $h_i(\mathbf x) = \frac{\partial \tilde f}{\partial e_i}(\mathbf x, \mathbf 0, \mathbf 0) \cdot (\tilde Q(\mathbf x, \mathbf 0, \mathbf 0))^2$ remains a polynomial, thereby ensuring the proposed algorithm is applicable in this setting.

\end{remark}

\subsection{Refinement with Range Partition}
\label{sec:partition}

A central component of the CMB algorithm, both in its original formulation and its extension to fractional expressions, is the computation of an upper bound on first-order violation probability, $\flag_u$, for a given threshold $u$. 
In order to tighten this bound, we introduce a \emph{range partition refinement} 
for the computation of $\flag_u$.
The key idea is to partition the domain of each variable in $\mathbf x$ into smaller sub-regions, compute a local upper bound on $\flag_u$ within every resulting sub-region of the input space, and then aggregate these local bounds to obtain the global estimate of $\flag_u$.

Several new notations and definitions related to range partition are introduced for readability.
Assume there are $m$ input variables, $x_1, \dots, x_m$, and the range of variable $x_i$ is $I_i\subset \mathbb R$.
For each $i\in \{1, \dots, m\}$, we partition the range $I_i$ into $b$ disjoint sub-ranges: $I_{i,1}, \dots, I_{i,b}$, where $I_i = \bigcup_{j=1}^bI_{i,j}$ and $I_{i,j}\cap I_{i,k} = \emptyset$ for $j\neq k$.
The weight of each sub-range $I_{i,j}$ is defined as $w_{i, j} \triangleq \mathbb P[x_i \in I_{i, j}]$.
Given a multi-index $(i_1, \dots, i_m)\in \{1,\dots, b\}^m$, a sub-region is defined as 
$B_{i_1, \dots, i_m} \triangleq I_{1, i_1}\times \cdots \times I_{m, i_m}$.
Assuming the independence of variables, the weight of the sub-region is given by the product $W_{i_1, \dots, i_m} = \mathbb P[\mathbf x \in B_{i_1, \dots, i_m}] = \prod_{j=1}^mw_{j, i_j}$.

\paragraph{Illustrative Example.} 
We illustrate range partition using Example~\ref{ex:poly}.
Suppose the number of sub-ranges is $b = 2$, and the ranges are partitioned uniformly. 
Then $I_{i,1} = [-1, 0], I_{i,2} = [0,1]$ for $i\in\{1,2,3\}$, each associated with an equal weight of $0.5$.
The eight sub-regions are different combinations of the sub-ranges. 
For example, sub-region $B_{1, 1, 1}$ is $[-1, 0]^3$, and its weight is $W_{1, 1, 1} = 0.5^3 = 0.125$.
For a given value of $u$, we compute the local probability bounds $\flag_{u, B_{\mathbf j}}$ for all $\mathbf{j} \in \{1,2\}^3$ and sub-region $B_{\mathbf j}$, where the expectations are taken w.r.t the normalized distribution on the sub-region, and then get the global bound by aggregating across all sub-regions $\flag_u = \sum_{\mathbf i\in \{1, 2\}^3}\flag_{u, B_\mathbf {i}}$.

We now present the general refined computation of $\flag_u$ in Algorithm~\ref{alg:partition}. 
The algorithm begins by partitioning the ranges into sub-ranges, and computing the corresponding weights (lines 1--4).
The main loop (lines 6--16) iterates over each sub-region $B$: it first computes the weight of the sub-region, $W$ (line 7), then derives the local upper bound $\flag_{u, B}$ (lines 8--14), and finally accumulates the weighted local upper bound onto the global bound (line 15).
We provide a proof of correctness in Theorem~\ref{thm:partitioncorrect}, and then explain why the range partition refinement could result in a tighter bound.

\begin{algorithm}[ht]
\caption{Computing $\flag_u$ with range partition}
\label{alg:partition}
\begin{algorithmic}[1]
\Require Input variables $x_1, \dots, x_m$ and their ranges $I_1, \dots, I_m$, threshold $u$, number of partitions $b$, expression of $K_u(\mathbf x)$ {(either in polynomial or fractional case)}, analysis order $n$
\Ensure An upper bound for first-order violation probability $\flag_u$
\For {$i \gets 1 \text{ to }m$}
    \State Partition range $I_i$ into sub-ranges $I_{i, 1},\dots,I_{i, b}$
    \State Compute weights: $w_{i,j} \gets \mathbb P[x_i\in I_{i,j}]$ for $j = 1$ to $b$
\EndFor
\State $\flag_u \gets 0$
\For {each multi-index $\mathbf i\in \{1, \ldots, b\}^m$}
    \State Compute sub-region weight $W_{\mathbf i} \gets \prod_{j=1}^mw_{j, i_j}$ 
    \State $\mu \gets \mathbb E[K_u(\mathbf x)\mid \mathbf x\in B_{\mathbf i}]$ 
        % \Comment Expectation based on normalized distribution of $\mathbf x$.
    \If {$\mu < 0$}
        \State $\flag_{u, B_{\mathbf i}} \gets \mathbb E[(K_u(\mathbf x) - \mu)^n \mid \mathbf x \in B_{\mathbf i}]/\mu^n$
        % \Comment Subscript of $B$ omitted for brevity.
    \Else 
        \State $\flag_{u, B_{\mathbf i}} \gets 1$
    \EndIf 
    \State Set $\flag_{u, B_{\mathbf i}}$ to $1$ if $\flag_{u, B_{\mathbf i}} > 1$
    \State $\flag_u \gets \flag_u + W_{\mathbf i} \cdot \flag_{u, B_{\mathbf i}}$
\EndFor
\State \textbf{Return} $\flag_u$
\end{algorithmic}
\end{algorithm} 

% \begin{algorithm}[ht]
% \caption{Computing $\flag_u$ with range partition}
% \label{alg:partition}
% \begin{algorithmic}[1]
% \Require Input variables $x_1, \dots, x_m$ and their ranges $I_1, \dots, I_m$, threshold $u$, number of partitions $b$, expression of $K_u(\mathbf x)$ {\color{blue} (either in polynomial or division case)}, analysis order $n$
% \Ensure An upper bound for first-order violation probability $\flag_u$
% \For {$i \gets 1 \text{ to }m$}
%     \State Partition range $I_i$ into sub-ranges $I_{i, 1},\dots,I_{i, b}$
%     \State Compute weights: $w_{i,j} \gets \mathbb P[x_i\in I_{i,j}]$ for $j = 1$ to $b$
% \EndFor
% \State $\flag_u \gets 0$
% \For {each sub-region $B$ (subscripts are omitted for brevity)}
%     \State Compute sub-region weight $W \gets \prod_{j=1}^mw_{j, i_j}$ {\color{blue} ($W$ or $W_i$ or $W_b$ ?)}
%     % \State Normalize distribution of $\mathbf x$ to $B$
%     \State $\mu \gets \mathbb E[K_u(\mathbf x)\mid \mathbf x\in B]$ 
%         \Comment Expectation based on normalized distribution of $\mathbf x$
%     \If {$\mu \le 0$}
%         \State $\flag_{u, B} \gets \mathbb E[(K_u(\mathbf x) - \mu)^n \mid \mathbf x \in B]/\mu^n$
%     \Else 
%         \State $\flag_{u, B} \gets 1$
%     \EndIf 
%     \State Set $\flag_{u, B}$ to $1$ if $\flag_{u, B} > 1$
%     \State $\flag_u \gets \flag_u + W \cdot \flag_{u, B}$
% \EndFor
% \State \textbf{Return} $\flag_u$
% \end{algorithmic}
% \end{algorithm} 

\begin{theorem}
Let $\flag_u$ be the output of Algorithm~\ref{alg:partition}, then $\mathbb P[K_u(\mathbf x) \ge 0] \le \flag _u$.
\label{thm:partitioncorrect}
\end{theorem}

\begin{proof}
We first establish the correctness of each iteration of the for-loop (line 6 -- 16) by showing that the $\flag_{u, B_{\mathbf i}}$ generated in line 14 is a sound local upper bound 
$\mathbb P[K_u \ge 0 \mid \mathbf x \in B_{\mathbf i}] \le \flag_{u, B_{\mathbf i}}.$
If $\mu > 0$  or the value computed in line 10 exceeds $1$, the algorithm sets $\flag_{u,B_{\mathbf i}}$ to $1$, and the desired inequality holds trivially. 
Otherwise, the validity of the bound follows by an argument analogous to that in Theorem~\ref{thm:cmb-flag}.
Specifically, under the normalized distribution of $\mathbf x$ conditioned on $\mathbf x\in B_{\mathbf i}$,
$$\begin{aligned}
\mathbb P[K_u(\mathbf x) \ge 0\mid \mathbf x \in B_{\mathbf i}] &\le \mathbb P[|K_u(\mathbf x) - \mu| \ge |\mu|\mid \mathbf x \in B_{\mathbf i}] =\mathbb P[(K_u(\mathbf x) - \mu)^n \ge \mu^n\mid \mathbf x \in B_{\mathbf i}] \\
&\le {\mathbb E[(K_u(\mathbf x) - \mu)^n\mid \mathbf x \in B_{\mathbf i}]}/{\mu^n} = \flag_{u, B_{\mathbf i}}.
\end{aligned}$$
This confirms that the local bound $\flag_{u, B_{\mathbf i}}$ indeed bounds the probability. 

Note that the returned $\flag_u$ is equal to the weighted sum of all local bounds, i.e.,
$\flag_u = \sum_{\mathbf i \in (1, \dots, b)^m}\flag_{u, B_{\mathbf i}}W_{\mathbf i}.$
Then we extend the local guarantee to a global one.
Since $\{B_{\mathbf i}\}_{\mathbf i\in (1, \dots, b)^m}$ is a set of disjoint sub-regions covering the domain of $\mathbf x$,
by the law of total probability,
$$\begin{aligned}
\mathbb P[K_u(\mathbf x) \ge 0] &= \textstyle\sum_{\mathbf i \in (1, \dots, b)^m} \mathbb P[K_u(\mathbf x) \ge 0 \mid \mathbf x \in B_{\mathbf i}] \cdot \mathbb P[\mathbf x \in B_{\mathbf i}] 
\le \textstyle\sum_{\mathbf i \in (1, \dots, b)^m} \mathbb  \flag_{u, B_{\mathbf i}}\cdot W_{\mathbf i} = \flag_u.
\end{aligned}$$
\end{proof}

By the proofs in Theorem~\ref{thm:cmb-flag} and Theorem~\ref{thm:div-flag}, we know that $\mathbb P[F(\mathbf x, \mathbf e) \ge u] \le \mathbb P[K_u(\mathbf x) \ge 0]$ for the CMB algorithm, both in its original form and its extension for fractional expressions. 
Although the specific form of the auxiliary expression $K_u(\mathbf x)$ differs between the two cases, the range partition refinement applies to both settings.
Consequently, in either case, the first-order violation probability satisfies
$\mathbb P[F(\mathbf x, \mathbf e) \ge u] \le \mathbb P[K_u(\mathbf x) > 0] \le \flag_u$
for the $\flag_u$ given by Algorithm~\ref{alg:partition}.

\begin{remark}
Range partition enhances the CMB method because the mean value of $K_u(\mathbf x)$ differs across sub-regions, enabling the computation of CMB bounds with sub-region-specific means rather than one global mean. 
Since CMB bounds are fractional, these differences influence the final aggreagated result.
By contract, partition offers no benefit for the NM method, since local bounds are combined by direct summation, which ultimately depends on the global mean over the entire range.
\qed
\end{remark}

\begin{remark}
Our current approach does not support transcendental functions. 
While our implementation can handle polynomial approximations of transcendental functions (such as benchmarks \texttt{ksin} and \texttt{kcos}), extending our approach to handle transcendental functions directly would require substantial effort. 
In particular, this would require precise knowledge of the polynomial or piecewise-polynomial approximations used in specific languages, libraries, or hardware, as well as accounting for the additional error introduced by replacing the transcendental functions with its (piecewise) polynomial approximation. Handling this approximation error falls outside the scope of this paper.
\qed
\end{remark}

\section{Implementation and Evaluation}
\label{sec:implementation}

\noindent {\bf Implementation.}
We have implemented our algorithms (Algorithm~\ref{alg:nm}, Algorithm \ref{alg:cmb} along with its extension for fractional expressions, Algorithm~\ref{alg:div}, and their refinement using Algorithm~\ref{alg:partition}) in a prototype tool named ProbTaylor. 
The implementation comprises approximately 2000 lines of code in OCaml, relying only on standard libraries -- \texttt{List}, \texttt{ocamllex}, and \texttt{ocamlyacc}.
We obtain second-order errors using the off-the-shelf sound optimizer, GELPIA~\cite{gelpia}, to carry out global optimization.
ProbTaylor offers three operational modes: \texttt{nm}, \texttt{cmb}, and \texttt{div}, corresponding to the Naive Markov algorithm, the Central-Moment-Based algorithm, and its extension for fractional expressions, respectively. 
ProbTaylor accepts as input a text file specifying the distribution of the input variables and the floating-point expression to be analyzed. 
Additionally, the user may also specify the desired precision parameters $\epsilon$ and $\delta$, the target confidence level $c$, the analysis order $n$, and the number of partitions $b$ (the latter applicable only to \texttt{cmb} and \texttt{div} modes).
Three types of input distributions are currently supported: uniform, normal, and double exponential (a.k.a. Laplace distribution) with various variance values.
The standard normal and double exponential distributions are truncated to a certain range specified by the input file.
Four basic arithmetic operations are supported: addition, subtraction, multiplication, and division.
The output of ProbTaylor is a valid threshold that is satisfied with probability of at least the given confidence level $c$.

It is worth mentioning that the expectation of given expressions is computed through exact symbolic derivation instead of numerical approximation.
This not only enhances the computational efficiency, but also produces a more accurate result. 
See Appendix~\ref{sec:symbolicexp} for details.

\smallskip
\noindent{\bf Benchmarks.}
The benchmark set we use to evaluate our prototype includes three parts -- polynomial benchmarks, fractional benchmarks, and scalability benchmarks.
The polynomial benchmarks include all 24 polynomial benchmarks from PAF~\cite{constantinides2021rigorous}, which are originally adapted from FPBench~\cite{damouche2017toward}, 
and 2 benchmarks adapted from \texttt{fdlibm}~\footnote{\texttt{fdlibm}: freely distributable math library for C programming language originally developed by Sun Microsystems, available at \url{http://www.netlib.org/fdlibm/}.} (\texttt{ksin}, \texttt{kcos}).
The number of operations in these polynomial benchmarks ranges from 1 to 31.
The fractional benchmarks include all 3 fractional benchmarks from PAF, and 4 additional ones from FPBench (\texttt{nonlin1}, \texttt{nonlin2}, \texttt{predator}, \texttt{verhulst}).
The number of operations in these fractional benchmarks ranges from 2 to 12. 
The scalability benchmarks are explained in Section~\ref{sec:scalability}.
Additional benchmarks are derived by widening the input distributions of existing benchmarks.
We assume single precision computation and $99\%$ target confidence level for all experiments in this section.

\smallskip
\noindent {\bf Evaluation Criteria.}
In the evaluation, we consider the following four research questions:
\begin{itemize}
    \item \textbf{RQ1:} How does the NM algorithm compare to the SOTA 
    tools, PAF and PrAn, in terms of accuracy and time efficiency on polynomial benchmarks?
    \item \textbf{RQ2:} How well does our approach scale when analyzing expressions with a large number of variables and operations?
    \item \textbf{RQ3:} How does the CMB algorithm, equipped with range partition refinement, perform relative to the NM algorithm and the existing tools on polynomial benchmarks? How does its extension for fractional expressions perform compared to existing tools?
    \item \textbf{RQ4:} To what extent can ProbTaylor improve upon the results obtained from conservative deterministic analysis tools (such as FPTaylor)?
\end{itemize}

\smallskip
\noindent {\bf Experiment Setup.}
The experiments were run on a 2017 iMac Pro with a 2.3 GHz 18-core Intel Xeon W processor and 128 GB of memory, running MacOS 14.4.1. 
Since PAF is only compatible with Ubuntu 18.04, all experiments, including those with PAF and PrAn, and with ProbTaylor, were run on a virtual machine configured with Ubuntu 18.04 hosted on the aforementioned iMac Pro.

\smallskip
\noindent {\bf Underflow and Exceptions.}
Due to numerical limits of OCaml and Python, underflow arises during the analysis in both ProbTaylor (e.g. on \texttt{classids} benchmarks with Laplace distribution with $\sigma = 0.01$), and in existing tools -- PAF (e.g. on \texttt{nonlin2}) and PrAn (e.g. on \texttt{doppler}).
When such underflow occurs in the denominator, it could potentially lead to division by zero, and produce NaN values.
In these cases, the corresponding entries in the results are marked as ``DZ''.
Following the same assumption of PAF, we assume no overflow or other exceptions during computations. 
This assumption is further validated since every benchmark passes the exception checker of FPTaylor.

\begin{table}[ht]
\centering
\caption{Experimental results for uniform distribution. The thresholds are produced assuming computations done in single precision and $99\%$ confidence level. The optimal $n$ is the analysis order with which the NM algorithm produces the tightest threshold. The threshold ratios are the ratios of thresholds given by the NM algorithm to those given by PAF/PrAn. A threshold ratio smaller than $100\%$ (marked in bold) indicates that the NM algorithm produces a tighter threshold, and vice versa. ``TO'' indicates timeout for competing tool.}
\label{tab:rq1uniform}
\begin{tabular}{|c|ccccccc|}
\hline
 & \multicolumn{7}{c|}{Uniform distribution} \\ \hline
 & \multicolumn{3}{c|}{Naive Markov algorithm} & \multicolumn{2}{c|}{Comparison with PAF} & \multicolumn{2}{c|}{Comparison with PrAn} \\ \hline
Benchmark & \multicolumn{1}{c|}{Time (s)} & \multicolumn{1}{c|}{Threshold} & \multicolumn{1}{c|}{Optimal $n$} & \multicolumn{1}{c|}{Speedup} & \multicolumn{1}{c|}{\begin{tabular}[c]{@{}c@{}}Threshold\\ ratio\end{tabular}} & \multicolumn{1}{c|}{Speedup} & \begin{tabular}[c]{@{}c@{}}Threshold\\ ratio\end{tabular} \\ \hline
bsplines0 & \multicolumn{1}{c|}{0.84} & \multicolumn{1}{c|}{4.85E-07} & \multicolumn{1}{c|}{36} & \multicolumn{1}{c|}{1,466x} & \multicolumn{1}{c|}{849.39\%} & \multicolumn{1}{c|}{21x} & 558.11\% \\ \hline
bsplines1 & \multicolumn{1}{c|}{0.74} & \multicolumn{1}{c|}{5.61E-07} & \multicolumn{1}{c|}{36} & \multicolumn{1}{c|}{2,210x} & \multicolumn{1}{c|}{301.61\%} & \multicolumn{1}{c|}{37x} & 296.83\% \\ \hline
bsplines2 & \multicolumn{1}{c|}{1.47} & \multicolumn{1}{c|}{5.60E-07} & \multicolumn{1}{c|}{36} & \multicolumn{1}{c|}{1,869x} & \multicolumn{1}{c|}{288.66\%} & \multicolumn{1}{c|}{24x} & 264.15\% \\ \hline
bsplines3 & \multicolumn{1}{c|}{0.26} & \multicolumn{1}{c|}{3.96E-08} & \multicolumn{1}{c|}{36} & \multicolumn{1}{c|}{3,565x} & \multicolumn{1}{c|}{\textbf{93.84\%}} & \multicolumn{1}{c|}{58x} & \textbf{69.35\%} \\ \hline
classids0 & \multicolumn{1}{c|}{98.12} & \multicolumn{1}{c|}{1.17E-05} & \multicolumn{1}{c|}{36} & \multicolumn{1}{c|}{\textgreater{}36x} & \multicolumn{1}{c|}{168.83\%} & \multicolumn{1}{c|}{1.08x} & 135.26\% \\ \hline
classids1 & \multicolumn{1}{c|}{99.76} & \multicolumn{1}{c|}{6.56E-06} & \multicolumn{1}{c|}{36} & \multicolumn{1}{c|}{\textgreater{}36x} & \multicolumn{1}{c|}{176.82\%} & \multicolumn{1}{c|}{1.46x} & 139.87\% \\ \hline
classids2 & \multicolumn{1}{c|}{98.29} & \multicolumn{1}{c|}{1.02E-05} & \multicolumn{1}{c|}{36} & \multicolumn{1}{c|}{\textgreater{}36x} & \multicolumn{1}{c|}{195.03\%} & \multicolumn{1}{c|}{1.22x} & 134.56\% \\ \hline
filters1 & \multicolumn{1}{c|}{0.19} & \multicolumn{1}{c|}{8.58E-08} & \multicolumn{1}{c|}{36} & \multicolumn{1}{c|}{3,200x} & \multicolumn{1}{c|}{\textbf{68.64\%}} & \multicolumn{1}{c|}{18x} & \textbf{42.27\%} \\ \hline
filters2 & \multicolumn{1}{c|}{0.75} & \multicolumn{1}{c|}{8.07E-07} & \multicolumn{1}{c|}{36} & \multicolumn{1}{c|}{4,248x} & \multicolumn{1}{c|}{101.77\%} & \multicolumn{1}{c|}{86x} & \textbf{79.90\%} \\ \hline
filters3 & \multicolumn{1}{c|}{19.49} & \multicolumn{1}{c|}{2.57E-06} & \multicolumn{1}{c|}{36} & \multicolumn{1}{c|}{\textgreater{}184x} & \multicolumn{1}{c|}{109.83\%} & \multicolumn{1}{c|}{10x} & \textbf{89.86\%} \\ \hline
filters4 & \multicolumn{1}{c|}{168.95} & \multicolumn{1}{c|}{5.62E-06} & \multicolumn{1}{c|}{36} & \multicolumn{1}{c|}{\textgreater{}21x} & \multicolumn{1}{c|}{135.42\%} & \multicolumn{1}{c|}{2.54x} & 108.08\% \\ \hline
rigidbody1 & \multicolumn{1}{c|}{97.02} & \multicolumn{1}{c|}{1.75E-04} & \multicolumn{1}{c|}{18} & \multicolumn{1}{c|}{\textgreater{}37x} & \multicolumn{1}{c|}{100.57\%} & \multicolumn{1}{c|}{0.28x} & 101.16\% \\ \hline
rigidbody2 & \multicolumn{1}{c|}{137.21} & \multicolumn{1}{c|}{1.55E-02} & \multicolumn{1}{c|}{8} & \multicolumn{1}{c|}{\textgreater{}26x} & \multicolumn{1}{c|}{\textbf{79.08\%}} & \multicolumn{1}{c|}{0.33x} & 159.79\% \\ \hline
sine & \multicolumn{1}{c|}{3.9} & \multicolumn{1}{c|}{5.52E-07} & \multicolumn{1}{c|}{36} & \multicolumn{1}{c|}{\textgreater{}923x} & \multicolumn{1}{c|}{232.91\%} & \multicolumn{1}{c|}{34x} & 229.05\% \\ \hline
solvecubic & \multicolumn{1}{c|}{31.03} & \multicolumn{1}{c|}{1.91E-05} & \multicolumn{1}{c|}{36} & \multicolumn{1}{c|}{\textgreater{}116x} & \multicolumn{1}{c|}{107.30\%} & \multicolumn{1}{c|}{5.82x} & \textbf{95.02\%} \\ \hline
sqrt & \multicolumn{1}{c|}{1.18} & \multicolumn{1}{c|}{1.39E-04} & \multicolumn{1}{c|}{36} & \multicolumn{1}{c|}{\textgreater{}3,050x} & \multicolumn{1}{c|}{\textbf{90.26\%}} & \multicolumn{1}{c|}{47x} & \textbf{90.26\%} \\ \hline
traincars1 & \multicolumn{1}{c|}{12.82} & \multicolumn{1}{c|}{2.47E-03} & \multicolumn{1}{c|}{36} & \multicolumn{1}{c|}{267x} & \multicolumn{1}{c|}{140.34\%} & \multicolumn{1}{c|}{3.92x} & 126.02\% \\ \hline
traincars2 & \multicolumn{1}{c|}{185.36} & \multicolumn{1}{c|}{1.84E-03} & \multicolumn{1}{c|}{18} & \multicolumn{1}{c|}{\textgreater{}19x} & \multicolumn{1}{c|}{176.92\%} & \multicolumn{1}{c|}{1.10x} & 138.35\% \\ \hline
traincars3 & \multicolumn{1}{c|}{138.94} & \multicolumn{1}{c|}{2.80E-02} & \multicolumn{1}{c|}{8} & \multicolumn{1}{c|}{\textgreater{}25x} & \multicolumn{1}{c|}{160.00\%} & \multicolumn{1}{c|}{0.13x} & 122.27\% \\ \hline
traincars4 & \multicolumn{1}{c|}{115.24} & \multicolumn{1}{c|}{3.53E-01} & \multicolumn{1}{c|}{6} & \multicolumn{1}{c|}{\textgreater{}31x} & \multicolumn{1}{c|}{195.03\%} & \multicolumn{1}{c|}{0.22x} & 153.48\% \\ \hline
trid1 & \multicolumn{1}{c|}{158.33} & \multicolumn{1}{c|}{6.78E-03} & \multicolumn{1}{c|}{18} & \multicolumn{1}{c|}{\textgreater{}22x} & \multicolumn{1}{c|}{112.81\%} & \multicolumn{1}{c|}{0.43x} & 110.78\% \\ \hline
trid2 & \multicolumn{1}{c|}{195.14} & \multicolumn{1}{c|}{1.35E-02} & \multicolumn{1}{c|}{10} & \multicolumn{1}{c|}{\textgreater{}18x} & \multicolumn{1}{c|}{131.07\%} & \multicolumn{1}{c|}{0.30x} & 115.38\% \\ \hline
trid3 & \multicolumn{1}{c|}{116.82} & \multicolumn{1}{c|}{2.46E-02} & \multicolumn{1}{c|}{6} & \multicolumn{1}{c|}{\textgreater{}30x} & \multicolumn{1}{c|}{140.57\%} & \multicolumn{1}{c|}{2.67x} & 126.15\% \\ \hline
trid4 & \multicolumn{1}{c|}{100.17} & \multicolumn{1}{c|}{4.64E-02} & \multicolumn{1}{c|}{4} & \multicolumn{1}{c|}{\textgreater{}35x} & \multicolumn{1}{c|}{172.49\%} & \multicolumn{1}{c|}{10x} & 161.11\% \\ \hline
ksin & \multicolumn{1}{c|}{8.09} & \multicolumn{1}{c|}{7.78E-08} & \multicolumn{1}{c|}{36} & \multicolumn{1}{c|}{\textgreater{}444x} & \multicolumn{1}{c|}{TO} & \multicolumn{1}{c|}{253x} & 112.27\% \\ \hline
kcos & \multicolumn{1}{c|}{6.72} & \multicolumn{1}{c|}{2.05E-07} & \multicolumn{1}{c|}{36} & \multicolumn{1}{c|}{\textgreater{}535x} & \multicolumn{1}{c|}{TO} & \multicolumn{1}{c|}{273x} & 170.83\% \\ \hline
\end{tabular}
\end{table}
\vspace{\tableskip}

\begin{table}[ht]
\centering
\caption{Experimental results for truncated standard normal distribution. The setup is the same as in Table~\ref{tab:rq1uniform}.}
\label{tab:rq1normal}
\begin{tabular}{|c|ccccccc|}
\hline
 & \multicolumn{7}{c|}{Truncated standard normal distribution} \\ \hline
 & \multicolumn{3}{c|}{Naive Markov algorithm} & \multicolumn{2}{c|}{Comparison with PAF} & \multicolumn{2}{c|}{Comparison with PrAn} \\ \hline
Benchmark & \multicolumn{1}{c|}{Time (s)} & \multicolumn{1}{c|}{Threshold} & \multicolumn{1}{c|}{Optimal $n$} & \multicolumn{1}{c|}{Speedup} & \multicolumn{1}{c|}{\begin{tabular}[c]{@{}c@{}}Threshold\\ ratio\end{tabular}} & \multicolumn{1}{c|}{Speedup} & \begin{tabular}[c]{@{}c@{}}Threshold\\ ratio\end{tabular} \\ \hline
bsplines0 & \multicolumn{1}{c|}{0.86} & \multicolumn{1}{c|}{7.56E-07} & \multicolumn{1}{c|}{8} & \multicolumn{1}{c|}{1,426x} & \multicolumn{1}{c|}{1323.99\%} & \multicolumn{1}{c|}{33x} & 869.97\% \\ \hline
bsplines1 & \multicolumn{1}{c|}{0.74} & \multicolumn{1}{c|}{7.56E-07} & \multicolumn{1}{c|}{6} & \multicolumn{1}{c|}{2,306x} & \multicolumn{1}{c|}{406.45\%} & \multicolumn{1}{c|}{72x} & 360.00\% \\ \hline
bsplines2 & \multicolumn{1}{c|}{1.47} & \multicolumn{1}{c|}{7.59E-07} & \multicolumn{1}{c|}{6} & \multicolumn{1}{c|}{1,889x} & \multicolumn{1}{c|}{391.24\%} & \multicolumn{1}{c|}{42x} & 358.02\% \\ \hline
bsplines3 & \multicolumn{1}{c|}{0.26} & \multicolumn{1}{c|}{6.18E-08} & \multicolumn{1}{c|}{4} & \multicolumn{1}{c|}{3,515x} & \multicolumn{1}{c|}{146.45\%} & \multicolumn{1}{c|}{86x} & 108.23\% \\ \hline
classids0 & \multicolumn{1}{c|}{98.77} & \multicolumn{1}{c|}{4.54E-06} & \multicolumn{1}{c|}{18} & \multicolumn{1}{c|}{\textgreater{}36x} & \multicolumn{1}{c|}{102.02\%} & \multicolumn{1}{c|}{2.83x} & \textbf{51.01\%} \\ \hline
classids1 & \multicolumn{1}{c|}{98.34} & \multicolumn{1}{c|}{2.67E-06} & \multicolumn{1}{c|}{18} & \multicolumn{1}{c|}{\textgreater{}36x} & \multicolumn{1}{c|}{\textbf{99.63\%}} & \multicolumn{1}{c|}{3.02x} & \textbf{56.09\%} \\ \hline
classids2 & \multicolumn{1}{c|}{98.89} & \multicolumn{1}{c|}{4.47E-06} & \multicolumn{1}{c|}{18} & \multicolumn{1}{c|}{\textgreater{}36x} & \multicolumn{1}{c|}{116.10\%} & \multicolumn{1}{c|}{2.74x} & \textbf{58.82\%} \\ \hline
filters1 & \multicolumn{1}{c|}{0.18} & \multicolumn{1}{c|}{8.37E-08} & \multicolumn{1}{c|}{24} & \multicolumn{1}{c|}{3,227x} & \multicolumn{1}{c|}{\textbf{67.50\%}} & \multicolumn{1}{c|}{30x} & \textbf{41.23\%} \\ \hline
filters2 & \multicolumn{1}{c|}{0.76} & \multicolumn{1}{c|}{7.48E-07} & \multicolumn{1}{c|}{24} & \multicolumn{1}{c|}{4,236x} & \multicolumn{1}{c|}{122.02\%} & \multicolumn{1}{c|}{159x} & \textbf{74.06\%} \\ \hline
filters3 & \multicolumn{1}{c|}{19.76} & \multicolumn{1}{c|}{2.30E-06} & \multicolumn{1}{c|}{18} & \multicolumn{1}{c|}{\textgreater{}182x} & \multicolumn{1}{c|}{112.20\%} & \multicolumn{1}{c|}{14x} & \textbf{80.14\%} \\ \hline
filters4 & \multicolumn{1}{c|}{168.12} & \multicolumn{1}{c|}{4.89E-06} & \multicolumn{1}{c|}{18} & \multicolumn{1}{c|}{\textgreater{}21x} & \multicolumn{1}{c|}{117.83\%} & \multicolumn{1}{c|}{4.91x} & \textbf{94.04\%} \\ \hline
rigidbody1 & \multicolumn{1}{c|}{87.27} & \multicolumn{1}{c|}{4.19E-06} & \multicolumn{1}{c|}{6} & \multicolumn{1}{c|}{\textgreater{}41x} & \multicolumn{1}{c|}{\textbf{68.24\%}} & \multicolumn{1}{c|}{0.52x} & \textbf{2.42\%} \\ \hline
rigidbody2 & \multicolumn{1}{c|}{133.20} & \multicolumn{1}{c|}{2.35E-05} & \multicolumn{1}{c|}{4} & \multicolumn{1}{c|}{\textgreater{}27x} & \multicolumn{1}{c|}{\textbf{39.23\%}} & \multicolumn{1}{c|}{0.71x} & \textbf{0.24\%} \\ \hline
sine & \multicolumn{1}{c|}{3.77} & \multicolumn{1}{c|}{6.85E-07} & \multicolumn{1}{c|}{6} & \multicolumn{1}{c|}{\textgreater{}954x} & \multicolumn{1}{c|}{289.03\%} & \multicolumn{1}{c|}{126x} & 284.23\% \\ \hline
solvecubic & \multicolumn{1}{c|}{31.37} & \multicolumn{1}{c|}{4.94E-06} & \multicolumn{1}{c|}{18} & \multicolumn{1}{c|}{\textgreater{}114x} & \multicolumn{1}{c|}{\textbf{72.22\%}} & \multicolumn{1}{c|}{7.34x} & \textbf{24.46\%} \\ \hline
sqrt & \multicolumn{1}{c|}{1.19} & \multicolumn{1}{c|}{2.98E-06} & \multicolumn{1}{c|}{4} & \multicolumn{1}{c|}{\textgreater{}3,025x} & \multicolumn{1}{c|}{1211.38\%} & \multicolumn{1}{c|}{97x} & \textbf{1.94\%} \\ \hline
traincars1 & \multicolumn{1}{c|}{12.75} & \multicolumn{1}{c|}{9.53E-04} & \multicolumn{1}{c|}{18} & \multicolumn{1}{c|}{235x} & \multicolumn{1}{c|}{115.38\%} & \multicolumn{1}{c|}{7.32x} & \textbf{48.62\%} \\ \hline
traincars2 & \multicolumn{1}{c|}{180.03} & \multicolumn{1}{c|}{4.07E-04} & \multicolumn{1}{c|}{12} & \multicolumn{1}{c|}{\textgreater{}19x} & \multicolumn{1}{c|}{112.43\%} & \multicolumn{1}{c|}{2.64x} & \textbf{30.60\%} \\ \hline
traincars3 & \multicolumn{1}{c|}{138.61} & \multicolumn{1}{c|}{5.92E-03} & \multicolumn{1}{c|}{8} & \multicolumn{1}{c|}{\textgreater{}25x} & \multicolumn{1}{c|}{\textbf{61.92\%}} & \multicolumn{1}{c|}{0.13x} & \textbf{26.55\%} \\ \hline
traincars4 & \multicolumn{1}{c|}{114.32} & \multicolumn{1}{c|}{6.95E-02} & \multicolumn{1}{c|}{6} & \multicolumn{1}{c|}{\textgreater{}31x} & \multicolumn{1}{c|}{\textbf{78.35\%}} & \multicolumn{1}{c|}{0.23x} & \textbf{30.22\%} \\ \hline
trid1 & \multicolumn{1}{c|}{158.87} & \multicolumn{1}{c|}{7.32E-06} & \multicolumn{1}{c|}{8} & \multicolumn{1}{c|}{\textgreater{}22x} & \multicolumn{1}{c|}{\textbf{46.33\%}} & \multicolumn{1}{c|}{1.59x} & \textbf{0.12\%} \\ \hline
trid2 & \multicolumn{1}{c|}{195.4} & \multicolumn{1}{c|}{1.21E-05} & \multicolumn{1}{c|}{8} & \multicolumn{1}{c|}{\textgreater{}18x} & \multicolumn{1}{c|}{\textbf{50.00\%}} & \multicolumn{1}{c|}{0.71x} & \textbf{0.10\%} \\ \hline
trid3 & \multicolumn{1}{c|}{116.59} & \multicolumn{1}{c|}{1.88E-05} & \multicolumn{1}{c|}{6} & \multicolumn{1}{c|}{\textgreater{}30x} & \multicolumn{1}{c|}{\textbf{27.65\%}} & \multicolumn{1}{c|}{4.60x} & \textbf{0.10\%} \\ \hline
trid4 & \multicolumn{1}{c|}{100.29} & \multicolumn{1}{c|}{3.34E-05} & \multicolumn{1}{c|}{4} & \multicolumn{1}{c|}{\textgreater{}35x} & \multicolumn{1}{c|}{\textbf{12.65\%}} & \multicolumn{1}{c|}{15.71x} & \textbf{0.11\%} \\ \hline
ksin & \multicolumn{1}{c|}{8.3} & \multicolumn{1}{c|}{1.49E-07} & \multicolumn{1}{c|}{4} & \multicolumn{1}{c|}{\textgreater{}433x} & \multicolumn{1}{c|}{TO} & \multicolumn{1}{c|}{\textgreater{}433x} & TO \\ \hline
kcos & \multicolumn{1}{c|}{7.18} & \multicolumn{1}{c|}{4.69E-07} & \multicolumn{1}{c|}{4} & \multicolumn{1}{c|}{\textgreater{}501x} & \multicolumn{1}{c|}{TO} & \multicolumn{1}{c|}{\textgreater{}501x} & TO \\ \hline
\end{tabular}
\end{table}

\begin{table}[ht]
\centering
\caption{Experimental results for truncated double exponential distribution. The setup is the same as in Table~\ref{tab:rq1uniform}. For benchmarks \texttt{classids0}, \texttt{classids1}, \texttt{classids2}, and \texttt{solvecubic}, the variance value is set to $1$. For all other benchmarks, the variance value is set to $0.01$.}
\label{tab:rq1laplace}
\begin{tabular}{|c|ccccc|}
\hline
 & \multicolumn{5}{c|}{Truncated double exponential distribution} \\ \hline
 & \multicolumn{3}{c|}{Naive Markov algorithm} & \multicolumn{2}{c|}{Comparison with PAF} \\ \hline
Benchmark & \multicolumn{1}{c|}{Time (s)} & \multicolumn{1}{c|}{Threshold} & \multicolumn{1}{c|}{Optimal $n$} & \multicolumn{1}{c|}{Speedup} & \begin{tabular}[c]{@{}c@{}}Threshold\\ ratio\end{tabular} \\ \hline
bsplines0 & \multicolumn{1}{c|}{0.84} & \multicolumn{1}{c|}{7.41E-08} & \multicolumn{1}{c|}{30} & \multicolumn{1}{c|}{3,129x} & 129.77\% \\ \hline
bsplines1 & \multicolumn{1}{c|}{0.74} & \multicolumn{1}{c|}{9.04E-08} & \multicolumn{1}{c|}{36} & \multicolumn{1}{c|}{2,974x} & 130.07\% \\ \hline
bsplines2 & \multicolumn{1}{c|}{1.47} & \multicolumn{1}{c|}{2.96E-08} & \multicolumn{1}{c|}{18} & \multicolumn{1}{c|}{\textgreater{}2,448x} & 140.28\% \\ \hline
bsplines3 & \multicolumn{1}{c|}{0.25} & \multicolumn{1}{c|}{1.07E-11} & \multicolumn{1}{c|}{2} & \multicolumn{1}{c|}{3,752x} & 140.42\% \\ \hline
classids0 & \multicolumn{1}{c|}{381.09} & \multicolumn{1}{c|}{6.76E-06} & \multicolumn{1}{c|}{18} & \multicolumn{1}{c|}{64x} & 135.20\% \\ \hline
classids1 & \multicolumn{1}{c|}{382.22} & \multicolumn{1}{c|}{3.90E-06} & \multicolumn{1}{c|}{18} & \multicolumn{1}{c|}{42x} & 126.62\% \\ \hline
classids2 & \multicolumn{1}{c|}{382.39} & \multicolumn{1}{c|}{6.23E-06} & \multicolumn{1}{c|}{18} & \multicolumn{1}{c|}{35x} & 150.12\% \\ \hline
filters1 & \multicolumn{1}{c|}{0.19} & \multicolumn{1}{c|}{2.69E-09} & \multicolumn{1}{c|}{6} & \multicolumn{1}{c|}{3,857x} & \textbf{49.54\%} \\ \hline
filters2 & \multicolumn{1}{c|}{0.74} & \multicolumn{1}{c|}{2.06E-08} & \multicolumn{1}{c|}{6} & \multicolumn{1}{c|}{\textgreater{}4,864x} & \textbf{71.03\%} \\ \hline
filters3 & \multicolumn{1}{c|}{19.88} & \multicolumn{1}{c|}{5.54E-08} & \multicolumn{1}{c|}{8} & \multicolumn{1}{c|}{\textgreater{}181x} & \textbf{50.83\%} \\ \hline
filters4 & \multicolumn{1}{c|}{168.43} & \multicolumn{1}{c|}{1.12E-07} & \multicolumn{1}{c|}{8} & \multicolumn{1}{c|}{\textgreater{}21x} & \textbf{24.30\%} \\ \hline
rigidbody1 & \multicolumn{1}{c|}{97.11} & \multicolumn{1}{c|}{9.01E-09} & \multicolumn{1}{c|}{6} & \multicolumn{1}{c|}{\textgreater{}37x} & \textbf{1.88\%} \\ \hline
rigidbody2 & \multicolumn{1}{c|}{132.93} & \multicolumn{1}{c|}{1.20E-08} & \multicolumn{1}{c|}{4} & \multicolumn{1}{c|}{\textgreater{}27x} & \textbf{1.26\%} \\ \hline
sine & \multicolumn{1}{c|}{3.82} & \multicolumn{1}{c|}{1.16E-08} & \multicolumn{1}{c|}{6} & \multicolumn{1}{c|}{\textgreater{}942x} & \textbf{77.85\%} \\ \hline
solvecubic & \multicolumn{1}{c|}{116.66} & \multicolumn{1}{c|}{1.52E-05} & \multicolumn{1}{c|}{12} & \multicolumn{1}{c|}{143x} & 107.04\% \\ \hline
sqrt & \multicolumn{1}{c|}{1.2} & \multicolumn{1}{c|}{2.91E-07} & \multicolumn{1}{c|}{24} & \multicolumn{1}{c|}{\textgreater{}3,000x} & 118.29\% \\ \hline
traincars1 & \multicolumn{1}{c|}{13.83} & \multicolumn{1}{c|}{3.10E-04} & \multicolumn{1}{c|}{36} & \multicolumn{1}{c|}{\textgreater{}260x} & \textbf{68.89\%} \\ \hline
traincars2 & \multicolumn{1}{c|}{181.78} & \multicolumn{1}{c|}{8.64E-06} & \multicolumn{1}{c|}{10} & \multicolumn{1}{c|}{\textgreater{}19x} & \textbf{30.53\%} \\ \hline
traincars3 & \multicolumn{1}{c|}{223.37} & \multicolumn{1}{c|}{1.17E-04} & \multicolumn{1}{c|}{8} & \multicolumn{1}{c|}{\textgreater{}16x} & \textbf{13.07\%} \\ \hline
traincars4 & \multicolumn{1}{c|}{114.92} & \multicolumn{1}{c|}{1.58E-03} & \multicolumn{1}{c|}{6} & \multicolumn{1}{c|}{\textgreater{}31x} & \textbf{21.56\%} \\ \hline
trid1 & \multicolumn{1}{c|}{158.35} & \multicolumn{1}{c|}{5.55E-07} & \multicolumn{1}{c|}{18} & \multicolumn{1}{c|}{\textgreater{}22x} & \textbf{3.51\%} \\ \hline
trid2 & \multicolumn{1}{c|}{195.13} & \multicolumn{1}{c|}{1.07E-06} & \multicolumn{1}{c|}{10} & \multicolumn{1}{c|}{\textgreater{}18x} & \textbf{4.40\%} \\ \hline
trid3 & \multicolumn{1}{c|}{117.04} & \multicolumn{1}{c|}{2.25E-06} & \multicolumn{1}{c|}{6} & \multicolumn{1}{c|}{\textgreater{}30x} & \textbf{3.32\%} \\ \hline
trid4 & \multicolumn{1}{c|}{99.84} & \multicolumn{1}{c|}{4.86E-06} & \multicolumn{1}{c|}{4} & \multicolumn{1}{c|}{\textgreater{}36x} & \textbf{1.84\%} \\ \hline
ksin & \multicolumn{1}{c|}{8.36} & \multicolumn{1}{c|}{3.87E-09} & \multicolumn{1}{c|}{6} & \multicolumn{1}{c|}{\textgreater{}430x} & TO \\ \hline
kcos & \multicolumn{1}{c|}{6.76} & \multicolumn{1}{c|}{1.89E-07} & \multicolumn{1}{c|}{10} & \multicolumn{1}{c|}{\textgreater{}532x} & TO \\ \hline
\end{tabular}
\end{table}

\subsection{Comparing the Naive Markov algorithm to PAF and PrAn (Answering RQ1)}
\label{sec:nmvspaf}

The NM algorithm (Algorithm~\ref{alg:nm}) allows customization of the analysis order $n$, thereby offering flexibility.
It is intended as a light-weight static analysis algorithm capable of producing results with relatively low computational overhead.
Therefore, in the evaluation, we execute the NM algorithm on orders 2, 4, 6, 8, 10, 12, 18, 24, 30, 36, with a timeout of 90 seconds for each individual analysis order, and terminate the execution upon a timeout.~\footnote{We are not aware of an existing method to determine the optimal analysis order in advance. Nonetheless, evaluating multiple values of $n$ remains computationally inexpensive.}
We report the tightest probabilistic threshold generated among these configurations, and report the running time as the total time elapsed across analysis orders, plus the time spent computing and bounding second-order error via GELPIA \footnote{For benchmarks \texttt{ksin} and \texttt{kcos}, GELPIA crashes when analyzing the second-order error. Therefore, we developed a simple brute-force algorithm to bound the second-order error in those cases. The brute-force algorithm is explained in Appendix~\ref{sec:bruteforce}}.

Since the NM algorithm only applies to division-free expressions, 
the experiments in this subsection exclude benchmarks involving division by non-constants.
We compare our thresholds and running time with those obtained using PAF and PrAn (settings C and D in \cite{lohar2019sound}). 
We consider the case where all operations are performed in single precision and the target confidence level is $99\%$.

We document the experimental results on the benchmarks with the original distributions from PAF in Table~\ref{tab:rq1uniform} (uniform distribution), Table~\ref{tab:rq1normal} (standard normal distribution), and Table~\ref{tab:rq1laplace} (Laplace distribution). 
In the case of Laplace distribution, we adopt a variance value of  $\sigma = 0.01$ to maintain consistency with PAF, except for four benchmarks (\texttt{classids0}, \texttt{classids1}, \texttt{classids2}, and \texttt{solvecubic}) where numerical underflow occurs.
For these four benchmarks, we instead set $\sigma=1$.
We impose a timeout of $3600$ seconds per benchmark.
For the PAF benchmarks where PAF exceeds this time limit,  we report the output thresholds from PAF as documented in  \cite{constantinides2021rigorous} and run PAF exhaustively to get thresholds for the $\sigma=1$ double exponential benchmarks as mentioned above.

Due to page limits, the output thresholds and running times of PAF and PrAn are not presented in the main tables.
These data can be found in Appendix~\ref{sec:fullexperiment}.
The tables focus on comparing our method against PAF and PrAn.
The columns labeled ``Speedup'' indicate how much faster (or, in a few cases, slower) the NM algorithm is compared to PAF/PrAn, expressed as the ratio of their respective running times.
For those benchmarks where PAF fails to complete within the $3600$-second timeout, we report a lower bound on the speedup.
The columns labeled ``Threshold ratio'' provide a quantitative comparison of the thresholds, expressing the thresholds derived using the NM algorithm as a percentage of those obtained with PAF/PrAn.
For uniform (Table~\ref{tab:rq1uniform}) and truncated normal distribution (Table~\ref{tab:rq1normal}), we compare the NM method to both PAF and PrAn. 
For truncated Laplace distribution (Table~\ref{tab:rq1laplace}), which is not supported by PrAn, we compare only to PAF.

The NM algorithm successfully completes all benchmarks within $200$ seconds, whereas PAF exceeds the time limit on 56 of 78 benchmarks.
This shows that our approach is significantly more time-efficient compared to PAF.
Furthermore, PrAn fails to finish analysis on 2 of 52 benchmarks.
Our approach outperforms PrAn in terms of runtime, achieving an average speedup of 49x and completing the analysis faster on 41 of 52 benchmarks.
Importantly, the advantage in time efficiency does not sacrifice the accuracy of the derived threshold, as the NM algorithm produces probabilistic thresholds that are within the same order of magnitude as those produced by PrAn on all benchmarks, and by PAF on all but two.
Moreover, the NM algorithm yields tighter thresholds than PAF on 30 of 72 benchmarks, 
and achieves thresholds more than an order of magnitude tighter on 6 of the benchmarks.
In comparison to PrAn, the NM algorithm produces tighter thresholds on 25 of 50 benchmarks, with improvements exceeding an order of magnitude on 7 benchmarks.
Overall, these results indicate that the NM algorithm generally derives comparable or better thresholds in substantially less time. 
Additionally, the accuracy of the results can be further improved through the CMB algorithm and range partition.

\subsection{Scalability Experiments (Answering RQ2)}
\label{sec:scalability}

\begin{table}[ht]
\centering
\caption{Scalability experiments. The benchmarks are constructed as the inner product of two vectors. All entries are independently and uniformly distributed over the interval $(0, 1)$. Analysis assumes single precision computation and $99\%$ target confidence level. Computational times for first-order and second-order error threshold are reported separately in columns labeled ``FO time (s)'' and ``SO time (s)''. ``TO'' stands for timeout.}
\label{tab:rq2scalability}
\begin{tabular}{|c|c|c|c|c|c|c|c|}
\hline
\begin{tabular}[c]{@{}c@{}}Vector\\ length\end{tabular} & Threshold & \begin{tabular}[c]{@{}c@{}}FO \\ time (s)\end{tabular} & \begin{tabular}[c]{@{}c@{}}SO\\ time (s)\end{tabular} & \begin{tabular}[c]{@{}c@{}}Vector\\ length\end{tabular} & Threshold & \begin{tabular}[c]{@{}c@{}}FO\\ time (s)\end{tabular} & \begin{tabular}[c]{@{}c@{}}SO\\ time(s)\end{tabular} \\ \hline
25 & 5.30E-05 & 0.09 & 53.64 & 175 & 2.33E-03 & 14.5 & TO \\ \hline
50 & 1.99E-04 & 0.45 & 1779.53 & 200 & 3.03E-03 & 21.58 & TO \\ \hline
75 & 4.39E-04 & 1.23 & 13965.87 & 225 & 3.83E-03 & 31.32 & TO \\ \hline
100 & 7.71E-04 & 2.88 & TO & 250 & 4.72E-03 & 42.38 & TO \\ \hline
125 & 1.20E-03 & 5.13 & TO & 275 & 5.71E-03 & 57.4 & TO \\ \hline
150 & 1.72E-03 & 9.12 & TO & 300 & 6.78E-03 & 75.37 & TO \\ \hline
\end{tabular}
\end{table}

To evaluate the scalability of the NM algorithm, we designed a set of large-scale benchmarks, constructed as the inner product of two vectors whose entries are independently and uniformly distributed over $(0, 1)$.
The analysis is performed with analysis order $n=2$.
Due to GELPIA failing to bound the second-order errors on these large benchmarks, the brute-force method is employed instead.
A 4-hour timeout is enforced on the analysis of second-order errors.
We report the running time for analyzing first- and second-order error separately.
We report the results in Table~\ref{tab:rq2scalability}.

The experimental results demonstrate that the NM algorithm maintains reasonable computational performance on large-scale inputs -- the first-order analysis completes within 0.1 to 76 seconds as the vector length increases from 25 to 300.
However, the computation of the second-order error expression and the brute-force method for bounding second-order error are considerably more expensive, consuming significantly more time or even being infeasibly slow. \jiawei{What is the timeout for deriving SO errors?}
Overall, these findings suggest that the NM algorithm scales effectively for first-order error analysis, while efficiency improvements for second-order error estimation remains an important direction for future work.

\subsection{Evaluating the Central-Moment-Based Algorithm (Answering RQ3)}
\label{sec:ablation}

\begin{table}[ht]
\centering
\caption{Comparison of the thresholds generated by the CMB algorithm and by the NM algorithm. The data reported are the ratio of the threshold generated by the Central-Moment-Based algorithm with range partition refinement to those by the Naive Markov Method. The experiments use an analysis order of $n=4$. For the Central-Moment-Based algorithm, the number of partitions is $b=8$ for benchmarks with fewer than four input variables and fewer than ten operations, and range partitioning is disabled otherwise. }
\label{tab:rq3polyratio}
\begin{tabular}{|c|c|c|c|c|c|c|c|}
\hline
Benchmark  & Uniform & Normal  & Laplace  & Benchmark  & Uniform & Normal  & Laplace  \\ \hline
bsplines0  & 52.54\% & 55.64\% & 56.13\%  & sine       & 59.60\% & 18.68\% & 100.00\% \\ \hline
bsplines1  & 55.53\% & 24.12\% & 61.00\%  & solvecubic & 53.23\% & 32.37\% & 67.01\%  \\ \hline
bsplines2  & 57.03\% & 22.66\% & 63.58\%  & sqrt       & 67.57\% & 16.04\% & 24.52\%  \\ \hline
bsplines3  & 62.39\% & 14.13\% & 70.29\%  & traincars1 & 47.88\% & 30.64\% & 58.57\%  \\ \hline
classids0  & 60.58\% & 51.23\% & 61.21\%  & traincars2 & 64.89\% & 71.70\% & 81.63\%  \\ \hline
classids1  & 62.33\% & 54.88\% & 62.93\%  & traincars3 & 64.79\% & 69.86\% & 76.64\%  \\ \hline
classids2  & 57.75\% & 51.31\% & 58.23\%  & traincars4 & 60.93\% & 61.85\% & 69.57\%  \\ \hline
filters1   & 51.53\% & 21.67\% & 59.38\%  & trid1      & 64.44\% & 22.14\% & 12.94\%  \\ \hline
filters2   & 51.06\% & 25.04\% & 62.06\%  & trid2      & 76.47\% & 80.13\% & 98.02\%  \\ \hline
filters3   & 51.67\% & 27.04\% & 59.59\%  & trid3      & 75.54\% & 77.22\% & 94.69\%  \\ \hline
filters4   & 65.14\% & 69.46\% & 71.03\%  & trid4      & 77.58\% & 74.85\% & 93.11\%  \\ \hline
rigidbody1 & 68.82\% & 16.23\% & 91.15\%  & ksin       & 77.40\% & 93.96\% & 66.14\%  \\ \hline
rigidbody2 & 93.82\% & 97.45\% & 100.00\% & kcos       & 49.79\% & 51.60\% & 16.20\%  \\ \hline
\end{tabular}
\end{table}

\begin{table}[ht]
\centering 
\caption{Comparison between thresholds generated by the CMB algorithm and by existing tools. The setup is the same as in Table~\ref{tab:rq3polyratio}. The CMB algorithm is run with analysis order $n=4$ and number of partitions $b=8$ for benchmarks with fewer than four input variables and fewer than ten operations, and range partition is disabled otherwise. The columns marked ``CMB/PAF'' are ratios of thresholds given by the CMB algorithm to those given by PAF, and ``CMB/PrAn'' is defined similarly. Ratios below $100\%$ indicate that the CMB algorithm yields a tighter threshold than the competing tool, and is highlighted in bold. ``TO'' indicates timeout for competing tool.}
\label{tab:rq3polycomp}
\begin{tabular}{|c|cc|cc|c|}
\hline
           & \multicolumn{2}{c|}{Uniform}                             & \multicolumn{2}{c|}{Normal}                              & Laplace          \\ \hline
Benchmarks & \multicolumn{1}{c|}{CMB/PAF}          & CMB/PrAn         & \multicolumn{1}{c|}{CMB/PAF}          & CMB/PrAn         & CMB/PAF          \\ \hline
bsplines0  & \multicolumn{1}{c|}{868.65\%}         & 570.77\%         & \multicolumn{1}{c|}{872.15\%}         & 573.07\%         & 112.61\%         \\ \hline
bsplines1  & \multicolumn{1}{c|}{307.53\%}         & 302.65\%         & \multicolumn{1}{c|}{125.27\%}         & 110.95\%         & 119.86\%         \\ \hline
bsplines2  & \multicolumn{1}{c|}{296.91\%}         & 271.70\%         & \multicolumn{1}{c|}{110.82\%}         & 101.42\%         & 104.74\%         \\ \hline
bsplines3  & \multicolumn{1}{c|}{\textbf{97.87\%}} & \textbf{72.33\%} & \multicolumn{1}{c|}{\textbf{20.69\%}} & \textbf{15.29\%} & \textbf{3.27\%}  \\ \hline
classids0  & \multicolumn{1}{c|}{210.68\%}         & 168.79\%         & \multicolumn{1}{c|}{107.64\%}         & \textbf{53.82\%} & 142.00\%         \\ \hline
classids1  & \multicolumn{1}{c|}{223.45\%}         & 176.76\%         & \multicolumn{1}{c|}{104.85\%}         & \textbf{59.03\%} & 131.17\%         \\ \hline
classids2  & \multicolumn{1}{c|}{235.18\%}         & 162.27\%         & \multicolumn{1}{c|}{122.34\%}         & \textbf{61.97\%} & 158.55\%         \\ \hline
filters1   & \multicolumn{1}{c|}{\textbf{72.56\%}} & \textbf{44.68\%} & \multicolumn{1}{c|}{\textbf{25.16\%}} & \textbf{15.37\%} & \textbf{7.99\%}  \\ \hline
filters2   & \multicolumn{1}{c|}{103.66\%}         & \textbf{81.39\%} & \multicolumn{1}{c|}{\textbf{51.88\%}} & \textbf{31.49\%} & \textbf{14.90\%} \\ \hline
filters3   & \multicolumn{1}{c|}{112.39\%}         & \textbf{91.96\%} & \multicolumn{1}{c|}{\textbf{51.71\%}} & \textbf{36.93\%} & \textbf{52.48\%} \\ \hline
filters4   & \multicolumn{1}{c|}{174.22\%}         & 139.04\%         & \multicolumn{1}{c|}{141.93\%}         & 113.27\%         & \textbf{25.38\%} \\ \hline
rigidbody1 & \multicolumn{1}{c|}{104.02\%}         & 104.62\%         & \multicolumn{1}{c|}{\textbf{12.05\%}} & \textbf{0.43\%}  & \textbf{0.41\%}  \\ \hline
rigidbody2 & \multicolumn{1}{c|}{\textbf{85.20\%}} & 172.16\%         & \multicolumn{1}{c|}{\textbf{38.23\%}} & \textbf{0.24\%}  & \textbf{0.92\%}  \\ \hline
sine       & \multicolumn{1}{c|}{254.01\%}         & 249.79\%         & \multicolumn{1}{c|}{\textbf{67.93\%}} & \textbf{66.80\%} & \textbf{12.48\%} \\ \hline
solvecubic & \multicolumn{1}{c|}{111.24\%}         & \textbf{98.51\%} & \multicolumn{1}{c|}{\textbf{45.91\%}} & \textbf{15.54\%} & 114.08\%         \\ \hline
sqrt       & \multicolumn{1}{c|}{\textbf{97.40\%}} & \textbf{97.40\%} & \multicolumn{1}{c|}{194.31\%}         & \textbf{0.31\%}  & 102.03\%         \\ \hline
traincars1 & \multicolumn{1}{c|}{140.91\%}         & 126.53\%         & \multicolumn{1}{c|}{\textbf{63.80\%}} & \textbf{26.89\%} & \textbf{63.33\%} \\ \hline
traincars2 & \multicolumn{1}{c|}{222.12\%}         & 173.68\%         & \multicolumn{1}{c|}{124.59\%}         & \textbf{33.91\%} & \textbf{30.99\%} \\ \hline
traincars3 & \multicolumn{1}{c|}{164.00\%}         & 125.33\%         & \multicolumn{1}{c|}{\textbf{65.69\%}} & \textbf{28.16\%} & \textbf{12.85\%} \\ \hline
traincars4 & \multicolumn{1}{c|}{166.30\%}         & 130.87\%         & \multicolumn{1}{c|}{\textbf{67.08\%}} & \textbf{25.87\%} & \textbf{17.33\%} \\ \hline
trid1      & \multicolumn{1}{c|}{115.81\%}         & 113.73\%         & \multicolumn{1}{c|}{\textbf{12.59\%}} & \textbf{0.03\%}  & \textbf{2.84\%}  \\ \hline
trid2      & \multicolumn{1}{c|}{151.46\%}         & 133.33\%         & \multicolumn{1}{c|}{\textbf{51.65\%}} & \textbf{0.11\%}  & \textbf{3.13\%}  \\ \hline
trid3      & \multicolumn{1}{c|}{139.43\%}         & 125.13\%         & \multicolumn{1}{c|}{\textbf{26.91\%}} & \textbf{0.09\%}  & \textbf{1.67\%}  \\ \hline
trid4      & \multicolumn{1}{c|}{128.62\%}         & 120.14\%         & \multicolumn{1}{c|}{\textbf{9.47\%}}  & \textbf{0.08\%}  & \textbf{0.63\%}  \\ \hline
ksin       & \multicolumn{1}{c|}{TO}               & 163.06\%         & \multicolumn{1}{c|}{TO}               & TO               & TO               \\ \hline
kcos       & \multicolumn{1}{c|}{TO}               & 197.50\%         & \multicolumn{1}{c|}{TO}               & TO               & TO               \\ \hline
\end{tabular}
\end{table}

\begin{table}[ht]
\centering
\caption{Extended CMB algorithm for fractional expressions. The experiments use an analysis order of $n=2$, and number of partitions $b=16$. ``DZ'' indicates that the competing tool encounters underflow-related NaNs (PAF, PrAn). ``N/A'' indicates lack of support for the corresponding distribution (PrAn laplace).}
\label{tab:rq3div}
\begin{tabular}{|c|c|cc|cc|cc|}
\hline
 &  & \multicolumn{2}{c|}{ProbTaylor} & \multicolumn{2}{c|}{\begin{tabular}[c]{@{}c@{}}Comparison\\ with PAF\end{tabular}} & \multicolumn{2}{c|}{\begin{tabular}[c]{@{}c@{}}Comparison\\ with PrAn\end{tabular}} \\ \hline
Benchmark & Distribution & \multicolumn{1}{c|}{Time (s)} & Threshold & \multicolumn{1}{c|}{Speedup} & \begin{tabular}[c]{@{}c@{}}Threshold\\ ratio\end{tabular} & \multicolumn{1}{c|}{Speedup} & \begin{tabular}[c]{@{}c@{}}Threshold\\ ratio\end{tabular} \\ \hline
doppler1 & uniform & \multicolumn{1}{c|}{758.96} & 7.01E-06 & \multicolumn{1}{c|}{\textgreater{}4.74x} & \textbf{8.82\%} & \multicolumn{1}{c|}{DZ} & DZ \\ \hline
doppler2 & uniform & \multicolumn{1}{c|}{760.92} & 1.86E-05 & \multicolumn{1}{c|}{\textgreater{}4.73x} & \textbf{13.01\%} & \multicolumn{1}{c|}{DZ} & DZ \\ \hline
doppler3 & uniform & \multicolumn{1}{c|}{1487.33} & 4.29E-06 & \multicolumn{1}{c|}{\textgreater{}2.42x} & \textbf{9.43\%} & \multicolumn{1}{c|}{DZ} & DZ \\ \hline
nonlin1 & uniform & \multicolumn{1}{c|}{2.5} & 6.33E-08 & \multicolumn{1}{c|}{336x} & \textbf{94.34\%} & \multicolumn{1}{c|}{9.33x} & \textbf{83.95\%} \\ \hline
nonlin2 & uniform & \multicolumn{1}{c|}{26.19} & 7.20E-06 & \multicolumn{1}{c|}{DZ} & DZ & \multicolumn{1}{c|}{8.60x} & 196.72\% \\ \hline
predator & uniform & \multicolumn{1}{c|}{23.22} & 4.42E-07 & \multicolumn{1}{c|}{41x} & 444.22\% & \multicolumn{1}{c|}{3.66x} & 442.00\% \\ \hline
verhulst & uniform & \multicolumn{1}{c|}{5.9} & 2.43E-07 & \multicolumn{1}{c|}{59x} & 141.28\% & \multicolumn{1}{c|}{5.09x} & 135.00\% \\ \hline
nonlin1 & normal & \multicolumn{1}{c|}{2.52} & 6.02E-08 & \multicolumn{1}{c|}{334x} & \textbf{89.72\%} & \multicolumn{1}{c|}{6.32x} & \textbf{79.84\%} \\ \hline
nonlin2 & normal & \multicolumn{1}{c|}{25.74} & 7.94E-05 & \multicolumn{1}{c|}{DZ} & DZ & \multicolumn{1}{c|}{26.48x} & 2577.92\% \\ \hline
predator & normal & \multicolumn{1}{c|}{22.26} & 4.42E-07 & \multicolumn{1}{c|}{43x} & 444.22\% & \multicolumn{1}{c|}{7.81x} & 442.00\% \\ \hline
verhulst & normal & \multicolumn{1}{c|}{6.12} & 2.43E-07 & \multicolumn{1}{c|}{57x} & 141.28\% & \multicolumn{1}{c|}{7.97x} & 135.00\% \\ \hline
nonlin1 & laplace & \multicolumn{1}{c|}{2.52} & 6.04E-08 & \multicolumn{1}{c|}{333x} & \textbf{90.01\%} & \multicolumn{1}{c|}{N/A} & N/A \\ \hline
nonlin2 & laplace & \multicolumn{1}{c|}{27.81} & 1.72E-05 & \multicolumn{1}{c|}{DZ} & DZ & \multicolumn{1}{c|}{N/A} & N/A \\ \hline
predator & laplace & \multicolumn{1}{c|}{22.9} & 4.42E-07 & \multicolumn{1}{c|}{41x} & 446.92\% & \multicolumn{1}{c|}{N/A} & N/A \\ \hline
verhulst & laplace & \multicolumn{1}{c|}{6.42} & 2.43E-07 & \multicolumn{1}{c|}{54x} & 143.79\% & \multicolumn{1}{c|}{N/A} & N/A \\ \hline
\end{tabular}
\end{table}

In this subsection, we compare the NM method with the CMB method (Algorithm~\ref{alg:cmb}) with range partition (Algorithm~\ref{alg:partition}) on polynomial benchmarks.
We also evaluate the extended CMB method for fractional expressions (Algorithm~\ref{alg:div}) by comparing with those obtained by PAF and PrAn.

We run both the NM algorithm and the CMB algorithm on the polynomial benchmarks with analysis order $n=4$. 
For the CMB algorithm, we configure the number of partitions to $b=8$ for benchmarks with fewer than four input variables and ten operations, while disabling range partition for the others. 
To avoid numerical underflow with double exponential distributions, we set the variance to $\sigma = 1$ when comparing with the NM algorithm.
The resulting output threshold ratios are documented in Table~\ref{tab:rq3polyratio}.
The threshold values and running times can be found in Appendix~\ref{sec:fullexperiment}.
% For sake of space, the running time 

We use the results from the NM algorithm as a baseline for comparison to evaluate the amount of refinement that can be achieved with the CMB algorithm.
From the experimental data, we conclude that the CMB algorithm with range partition refinement significantly improves the derived thresholds.
Among the 78 polynomial benchmarks,
the CMB algorithm yields thresholds less than half of the baseline values in 16 cases. 
Furthermore, the average threshold refinement, defined as the ratio of the refined threshold to the baseline, is observed to be $62.7\%$ for uniform distribution benchmarks, $47.4\%$ for normal distribution benchmarks, and $66.8\%$ for double exponential benchmarks. 
The results underscore the effectiveness of the CMB algorithm and the range partition technique in deriving tighter probabilistic error thresholds. 
Besides, the CMB's time efficiency is not much worse than the NM's, with all but three benchmarks finishing within 90 seconds.

We then perform a comparison between the overall results produced by ProbTaylor and the other tools over all polynomial benchmarks.
The threshold ratios are presented in Table~\ref{tab:rq3polycomp}~\footnote{
Readers might have noticed that ProbTaylor exhibits relatively better performance on the normal and Laplace distributions. This behavior is expected, as our approach fundamentally relies on the application of ``concentration'' inequalities. Consequently, the approach is particularly effective for distributions that are more ``concentrated'' -- namely, those with relatively small variance.
}.
ProbTaylor achieves more accurate thresholds than PAF on 35 out of 72 polynomial benchmarks (PAF times out on 6), improving by more than an order of magnitude in 9 instances. 
Compared to PrAn, ProbTaylor produces tighter thresholds on 26 out of 50 polynomial benchmarks (PrAn times out on 2), improving by more than an order of magnitude in 7 instances.
In terms of threshold accuracy, ProbTaylor never underperforms PAF or PrAn by more than an order of magnitude.
The CMB algorithm with the range partition refinement remains faster than competing tools, with all benchmarks except \texttt{rigidbody1} finishing within 2 minutes.

Next, we evaluate the extension of the CMB algorithm for handling fractional expressions and compare the results of our tool to PAF and PrAn.
For these experiments, we set the analysis order to $n=2$ and number of partitions to $b=16$.
We run the extended CMB algorithm, PAF, and PrAn on the three \texttt{doppler} benchmarks from PAF with uniform distribution due to underflow issues in OCaml with the normal and Laplace distributions. We run four additional benchmarks using all three distributions.
We document the running time and resulting thresholds, as well as comparison with PAF and PrAn in Table~\ref{tab:rq3div}.
``Speedup'' and ``threshold ratio'' are defined in the same way as in Section~\ref{sec:nmvspaf}.
Cases where PAF or PrAn encounters NaNs caused by underflow are indicated with ``DZ''.
PrAn is not run in Laplace distribution cases due to its lack of support.

From the experimental results on fractional expressions, ProbTaylor's CMB algorithm is substantially more time-efficient than PAF, achieving an average speedup of 109x.
Moreover, ProbTaylor produces tighter thresholds than PAF on 6 out of 12 benchmarks that PAF is able to analyze, with 2 of these improvements exceeding an order of magnitude.
When compared to PrAn, ProbTaylor also exhibits significant time-efficiency advantage, with an average speedup of 9.4x.
Across all 11 benchmarks evaluated against PrAn, ProbTaylor successfully produces result for 3 benchmarks where PrAn crashes, and yields tighter threshold on 2 of them.

\subsection{Comparison to FPTaylor (Answering RQ4)}
\label{sec:fptaylor}

\begin{table}[ht]
\centering 
\caption{Comparison with FPTaylor on benchmarks with widened input range. The column ``original'' thresholds given on benchmarks with original input variable range. 
Here the ``threshold ratio'' is defined as the ratio of threshold generated by ProbTaylor with optimal analysis order $n$ and $99\%$ confidence on normally distributed input variables and the deterministic bound generated by FPTaylor.
The column ``doubled'' and ``quadrupled'' correspond to the threshold ratios on benchmarks with doubled and quadrupled input ranges respectively. Lower numbers indicate better result by probabilistic analysis.
}
\label{tab:rq4fptaylor}
\begin{tabular}{|cccccccc|}
\hline
\multicolumn{8}{|c|}{\begin{tabular}[c]{@{}c@{}}Threshold ratio $\left( \frac{\text{ProbTaylor}}{\text{FPTaylor}}\right)$ on widened input ranges\\ on truncated standard normal distribution\end{tabular}} \\ \hline
\multicolumn{1}{|c|}{benchmark} & \multicolumn{1}{c|}{original} & \multicolumn{1}{c|}{doubled} & \multicolumn{1}{c|}{quadrupled} & \multicolumn{1}{c|}{benchmark} & \multicolumn{1}{c|}{original} & \multicolumn{1}{c|}{doubled} & quadrupled \\ \hline
\multicolumn{1}{|c|}{bsplines0} & \multicolumn{1}{c|}{1321.68\%} & \multicolumn{1}{c|}{1910.24\%} & \multicolumn{1}{c|}{245.03\%} & \multicolumn{1}{c|}{rigidbody2} & \multicolumn{1}{c|}{0.12\%} & \multicolumn{1}{c|}{0.01\%} & 0.00\% \\ \hline
\multicolumn{1}{|c|}{bsplines1} & \multicolumn{1}{c|}{404.28\%} & \multicolumn{1}{c|}{320.28\%} & \multicolumn{1}{c|}{88.17\%} & \multicolumn{1}{c|}{sine} & \multicolumn{1}{c|}{287.82\%} & \multicolumn{1}{c|}{234.48\%} & 14.50\% \\ \hline
\multicolumn{1}{|c|}{bsplines2} & \multicolumn{1}{c|}{389.23\%} & \multicolumn{1}{c|}{272.48\%} & \multicolumn{1}{c|}{74.80\%} & \multicolumn{1}{c|}{solvecubic} & \multicolumn{1}{c|}{30.68\%} & \multicolumn{1}{c|}{DZ} & DZ \\ \hline
\multicolumn{1}{|c|}{bsplines3} & \multicolumn{1}{c|}{146.45\%} & \multicolumn{1}{c|}{95.56\%} & \multicolumn{1}{c|}{43.70\%} & \multicolumn{1}{c|}{sqrt} & \multicolumn{1}{c|}{1.97\%} & \multicolumn{1}{c|}{0.12\%} & 0.01\% \\ \hline
\multicolumn{1}{|c|}{classids0} & \multicolumn{1}{c|}{66.28\%} & \multicolumn{1}{c|}{45.69\%} & \multicolumn{1}{c|}{DZ} & \multicolumn{1}{c|}{traincars1} & \multicolumn{1}{c|}{54.77\%} & \multicolumn{1}{c|}{26.62\%} & 13.63\% \\ \hline
\multicolumn{1}{|c|}{classids1} & \multicolumn{1}{c|}{73.76\%} & \multicolumn{1}{c|}{47.38\%} & \multicolumn{1}{c|}{DZ} & \multicolumn{1}{c|}{traincars2} & \multicolumn{1}{c|}{43.02\%} & \multicolumn{1}{c|}{21.59\%} & 10.82\% \\ \hline
\multicolumn{1}{|c|}{classids2} & \multicolumn{1}{c|}{86.80\%} & \multicolumn{1}{c|}{59.42\%} & \multicolumn{1}{c|}{DZ} & \multicolumn{1}{c|}{traincars3} & \multicolumn{1}{c|}{32.89\%} & \multicolumn{1}{c|}{16.52\%} & 8.25\% \\ \hline
\multicolumn{1}{|c|}{filters1} & \multicolumn{1}{c|}{66.96\%} & \multicolumn{1}{c|}{51.60\%} & \multicolumn{1}{c|}{26.15\%} & \multicolumn{1}{c|}{traincars4} & \multicolumn{1}{c|}{38.40\%} & \multicolumn{1}{c|}{19.81\%} & 9.90\% \\ \hline
\multicolumn{1}{|c|}{filters2} & \multicolumn{1}{c|}{94.33\%} & \multicolumn{1}{c|}{65.41\%} & \multicolumn{1}{c|}{33.44\%} & \multicolumn{1}{c|}{trid1} & \multicolumn{1}{c|}{0.12\%} & \multicolumn{1}{c|}{0.03\%} & 0.01\% \\ \hline
\multicolumn{1}{|c|}{filters3} & \multicolumn{1}{c|}{103.14\%} & \multicolumn{1}{c|}{67.42\%} & \multicolumn{1}{c|}{33.89\%} & \multicolumn{1}{c|}{trid2} & \multicolumn{1}{c|}{0.12\%} & \multicolumn{1}{c|}{0.03\%} & 0.01\% \\ \hline
\multicolumn{1}{|c|}{filters4} & \multicolumn{1}{c|}{128.35\%} & \multicolumn{1}{c|}{82.18\%} & \multicolumn{1}{c|}{41.24\%} & \multicolumn{1}{c|}{trid3} & \multicolumn{1}{c|}{0.11\%} & \multicolumn{1}{c|}{0.03\%} & 0.01\% \\ \hline
\multicolumn{1}{|c|}{rigidbody1} & \multicolumn{1}{c|}{2.65\%} & \multicolumn{1}{c|}{0.66\%} & \multicolumn{1}{c|}{0.17\%} & \multicolumn{1}{c|}{trid4} & \multicolumn{1}{c|}{0.13\%} & \multicolumn{1}{c|}{0.03\%} & 0.01\% \\ \hline
\end{tabular}
\end{table}

To demonstrate that ProbTaylor is capable of producing thresholds in situations where deterministic analysis is too conservative, we adopt the standard deterministic analysis tool FPTaylor~\cite{solovyev2018rigorous} as a baseline. 
We conduct experiments using benchmarks with original distribution ranges from Table~\ref{tab:rq1uniform}, Table~\ref{tab:rq1normal} and Table~\ref{tab:rq1laplace} (excluding two realistic benchmarks from \texttt{fdlibm}, which are designed specifically for certain input ranges).
We derive additional benchmarks by doubling and quadrupling (the end-points of) the ranges of the distributions.
For example, for a uniform distribution with range $(-2, 4)$, we double the range to get the uniform distribution over $(-4, 8)$, and quadruple the range to get the uniform distribution over $(-8, 16)$. 
The cases for truncated normal and double exponential distributions are similar: we widen the range, and keep mean and variance unchanged.

The experimental results for standard normal distribution are collected in Table~\ref{tab:rq4fptaylor}.
Here, ``DZ'' indicates NaN produced by ProbTaylor.
The ProbTaylor configuration follows the same settings as described in Section~\ref{sec:nmvspaf}, with single precision, target confidence level $99\%$, and the same strategy -- reporting the optimal threshold among a set of analysis orders with a 90-second timeout per order.

We observe that ProbTaylor produces significantly tighter thresholds compared to FPTaylor with expanded input ranges.
On average, doubling the input range results in a $45\%$ decrease in threshold ratio, and quadrupling the range leads to a $82\%$ decrease compared to original benchmarks.
We focus only on comparing accuracy rather than runtime as (i) both tools run within 200 seconds for every benchmark and (ii) they have different objectives (deterministic vs. probabilistic analysis). \jiawei{Why would the different objectives preclude runtime analysis?}

\smallskip
\noindent {\bf Summary of Experimental Evaluations. }
When employing the NM algorithm, ProbTaylor is significantly faster than both PAF and PrAn while producing thresholds of comparable precision.
The NM algorithm itself scales well, capable of handling polynomial expressions with hundreds of operations within minutes.
However, the analysis of second-order error remains a scalability bottleneck in the overall toolchain.
Nevertheless, considering the CMB algorithm and its extensions, ProbTaylor consistently remains faster than existing tools. 
In terms of accuracy, the extended variants yield tighter thresholds than PAF and PrAn on roughly half of the benchmarks.
Moreover, comparisons against deterministic analysis show our probabilistic approach yields tighter thresholds under widened input ranges, demonstrating robustness with large input distribution's support.

\section{Related Works}
\label{sec:related}

Static analysis of round-off errors is an active research area. Most existing approaches do not involve probabilistic inputs, but rather are deterministic methods and produce bounds valid even in worst-case scenarios.
Many of these, including Gappa~\cite{daumas2010certification}, Gappa++~\cite{linderman2010towards}, FLUCTUAT~\cite{delmas2009towards}, RangeLab~\cite{martel2011rangelab}, Rosa~\cite{darulova2014sound} and Daisy~\cite{darulova2018daisy},
use abstract interpretation techniques~\cite{cousot1977abstract}, and employ abstract domains such as intervals~\cite{moore1966interval}, affine forms~\cite{ld1993affine}, or polyhedra~\cite{chen2008sound} to analyze errors.

An alternative approach, used by PRECiSA~\cite{titolo2018abstract} and Real2Float~\cite{magron2017certified}, formulates the bounding of round-off errors as an optimization problem.
Lee et al.~\cite{lee2016verifying} also employ optimization-based techniques and deal with the interplay between floating-point and bit-level operations.
FPTaylor~\cite{solovyev2018rigorous} leverages symbolic Taylor expansion, and optimizes its lower-order terms.
Our approach draws on FPTaylor's idea for the initial transformation and over-approximation of the target problem.

Recently, researchers have started to focus on probabilistic analysis of round-off error, aiming to produce bounds that are less pessimistic yet still valid with high probability.
PrAn~\cite{lohar2019sound} is the first work we are aware of that provides probabilistic analysis of round-off errors with specified input distributions. 
It extends probabilistic affine arithmetic~\cite{bouissou2012generalization} and utilizes probabilistic interval subdivision techniques. 
PAF~\cite{constantinides2021rigorous} is the current state-of-the-art, deriving tight bounds on most benchmarks, but is very slow even for moderate size benchmarks.
There are also results on analyzing probabilistic programs via concentration inequalities~\cite{DBLP:conf/cav/SunFCG23,DBLP:journals/toplas/ChatterjeeFNH18,DBLP:conf/pldi/WangS0CG21,DBLP:conf/pldi/Wang0R21,DBLP:conf/tacas/KuraUH19,DBLP:conf/popl/ChatterjeeNZ17,DBLP:journals/pacmpl/ChatterjeeGMZ24,DBLP:conf/cav/ChakarovS13,bouissou2016uncertainty}. Our result is orthogonal to them as we focus on floating-point arithmetic and deal with absolute values in the Taylor expansion.

\section{Conclusion and Future Work}
\label{sec:conclusion}

In this paper, we propose a novel approach to sound probabilistic analysis of floating-point round-off errors.
Our method employs Taylor expansion and concentration inequalities, combined with a positive-negative decomposition technique, to compute rigorous probabilistic thresholds of round-off errors.
We implemented our algorithms in a prototype tool, ProbTaylor, and evaluated its performance on a wide range of benchmarks.
Experimental results demonstrate that ProbTaylor yields thresholds  at least comparable to, and sometimes significantly more accurate than, those produced by existing tools such as PAF and PrAn, while substantially more time-efficient.
Besides, it potentially scales well to large benchmarks and can rule out unlikely worst-case scenarios that often hinder deterministic analysis.
For future work, we plan to expand our supported range of operations, including trigonometrics, logarithms, exponentials, etc. 
Another possible direction would be to use concentration bounds other than Markov's inequality, such as Chernoff bounds. 

\paragraph{Data Availability Statement.}
Our artifact can be found at \url{doi.org/10.5281/zenodo.18499134}.
It supports an implementation of ProbTaylor, and its comparison with PrAn (PAF is out of scope because it takes several days to run), and experiments for Sections~\ref{sec:nmvspaf}, \ref{sec:scalability} and \ref{sec:ablation}. Experiments in Section~\ref{sec:fptaylor} is out of scope, but FPTaylor can be run at \url{monadius.github.io/FPTaylorJS}.

%%
%% The acknowledgments section is defined using the "acks" environment
%% (and NOT an unnumbered section). This ensures the proper
%% identification of the section in the article metadata, and the
%% consistent spelling of the heading.
\begin{acks}
We thank Karthik Duraisamy, Sahil Bhola and Daisuke Uchida for fruitful discussions, as well as the anonymous reviewers for helpful comments. This research was supported in part by NSF grants CCF-2219997, CCF-2348706 and CCF-2446214.
\end{acks}

%%
%% The next two lines define the bibliography style to be used, and
%% the bibliography file.
\bibliographystyle{ACM-Reference-Format}
\bibliography{bib}

\newpage
%%
%% If your work has an appendix, this is the place to put it.
\appendix

\appendix

\section{Symbolic Computation of Expectations}
\label{sec:symbolicexp}

As is mentioned in Section~\ref{sec:implementation}, ProbTaylor computes expectations symbolically rather than numerically.
In this section, we provide a detailed explanation of the mathematical derivations underlying these symbolic computations.

Our symbolic expectation computation deals with expressions represented in polynomial form. 
We compute the expectation of each term individually, and sum them to obtain the expectation of the entire expression.
Within each term, the factors are grouped based on their original variables (e.g., $x$ and $x^+$ has the same original variable, and thus belong to the same factor group). 
Each group can then be simplified to just one factor (or even zero).
If the group involves more than one factor, it has one of the following four forms:
\begin{itemize}
    \item $x^n(x^+)^m$. This can be simplified to $(x^+)^{n+m}$ since $x^n$ just gets multiplied by zero when $x<0$.
    \item $x^n(x^-)^m$. This can be simplified to $(-1)^n(x^-)^{n+m}$ for similar reason, and the fact that $x^- = -x$ when $x<0$ contributes to the $(-1)^n$ factor.
    \item $(x^+)^n(x^-)^m$. This term is always equal to zero, since at least one factor is zero no matter what value $x$ takes.
    \item $(x^+)^n(x^-)^mx^k$. This term is also always equal to zero, for the same reason as above.
\end{itemize}
Therefore, our task reduces to calculating the expectation or higher moments or a single variable and multiply them together to obtain the expectation of the whole term.
The computation of higher moments for single factors following different distributions are presented in the following subsections.

\subsection{Uniform Distribution}
When $x$ is uniformly distributed on the interval $[a,b]$, its expectation and higher moments can be computed as follows.
$$ E[x^k] = \int_a^b\frac{x^k}{b - a}dx = \frac{b ^{k+1} -a^{k+1}}{(k+1)(b-a)}. $$
If the range $[a,b]$ crosses zero, i.e., $a < 0 < b$, then we can compute the higher moments of the positive component and negative component as follows.
$$E[x_+^k] = \int_a^0 \frac{0}{b-a}dx + \int_0^b\frac{x^k}{b-a}dx = \frac{b^{k+1}}{(k+1)(b-a)}.$$
$$E[x_-^k] = \int_a^0\frac{(-x)^k}{b-a}dx + \int_0^b\frac{0}{b-a}dx = \frac{(-a)^{k+1}}{(k+1)(b-a)}.$$
If the interval $[a,b]$ does not cross zero, then either $x = x^+$ or $x=-x^-$ holds, and the computation can be simplified accordingly.

\subsection{Truncated Standard Normal Distribution}
When $x$ follows the standard normal distribution truncated to interval $[a,b]$, its probabilistic distribution function (PDF) is expressed as
$$f(x) = \left\{
\begin{aligned}
\frac{\varphi(x)}{\Phi(b) - \Phi(a)}, &\quad x \in [a,b]\\
0, &\quad \text{otherwise}
\end{aligned}
\right.,
$$
where $\varphi(x) = \frac{1}{\sqrt{2\pi}}\exp(-\frac{x^2}{2})$ is the probabilistic distribution function of standard normal distribution, and $\Phi(x) = \int_{-\infty}^{x}\varphi(t)dt$ is the cummulative distribution function (CDF) of the standard normal distribution.
It is noteworthy that the derivative of $\varphi(x)$ satisfies $\varphi'(x) = \frac{1}{\sqrt{2\pi}}\exp(-\frac{x^2}{2})(-x) = -x\varphi(x)$.

Let $I_k = \int_a^bx^k\varphi(x)dx$, then the $k$-th moment of $x$ is given by
$$\mathbb E[x^k] = \int_a^b x^kf(x)dx = \frac{1}{\Phi(b) - \Phi(a)}\int_a^bx^k\varphi(x)dx = \frac{I_k}{\Phi(b) - \Phi(a)}.$$
Using integration by parts, let $u=x^{k-1}$ and $dv = x\varphi(x)dx = -\varphi'(x)dx$, and we have $du=(k-1)x^{k-2}$ and $v = -\varphi(x)$.
Apply the integration by parts formula, we derive
$$I_k = uv{\Large|}^b_a - \int^b_avdu = \left[-x^{k-1}\varphi(x)\right]^b_a + (k-1)I_{k-2}.$$
Therefore, the $k$-th moment of $x$ can be computed recursively with 
$$\mathbb E[x^k] = \frac{a^{k-1}\varphi(a) - b^{k-1}\varphi(b)}{\Phi(b) - \Phi(a)} + (k-1)\mathbb E[x^{k-2}].$$
The base cases to start with are as follows.
$$\mathbb E[x^0] = E[1] = 1.$$
$$\mathbb E[x^1] = \frac{1}{\Phi(b) - \Phi(a)}\int_a^bx\varphi(x)dx = \frac{1}{\Phi(b) - \Phi(a)}\left[-\varphi(x)\right]^b_a = \frac{\varphi(a) - \varphi(b)}{\Phi(b) - \Phi(a)}.$$
Using this recursive formula, the moments of $x$ for any order $k$ can be efficiently computed for the truncated standard normal distribution.

Next we turn to compute the higher moment for the positive components $x^+$ and negative components $x^-$. 
Only the cases where $a < 0 < b$ are discussed here, for otherwise, the computation of $\mathbb E[x^k]$ can be used directly.
The derivations are similar as for computing $\mathbb E[x^k]$ and we only present the recursion formula and the base cases below.
$$\mathbb E[x_+^k] = \frac{-b^{k-1}\varphi(b)}{\Phi(b) - \Phi(a)} + (k-1)\mathbb E[x_+^{k-2}].$$
$$\mathbb E[x_+^0] = \frac{\Phi(b) - \Phi(0)}{\Phi(b) - \Phi(a)}, \quad \mathbb E[x_+^1] = \frac{-\varphi(b) + \varphi(0)}{\Phi(b) - \Phi(a)}.$$
$$\mathbb E[x_-^k] = \frac{(-1)^ka^{k-1}\varphi(a)}{\Phi(b) - \Phi(a)} + (k-1)\mathbb E[x_-^{k-2}].$$
$$\mathbb E[x_-^0] = \frac{\Phi(0) - \Phi(a)}{\Phi(b) - \Phi(a)},\quad \mathbb E[x_-^1] = \frac{\varphi(0)-\varphi(a)}{\Phi(b) - \Phi(a)}.$$

\subsection{Truncated Double Exponential Distribution}
When a random variable $x$ follows the double exponenetial distribution (also known as Laplace distribution) parameterized by $\sigma$ truncated to the interval $[a,b]$, its PDF is expressed as
$$f(x) = \left\{
\begin{aligned}
\frac{\varphi(x)}{\Phi(b) - \Phi(a)}, &\quad x \in [a,b]\\
0, &\quad \text{otherwise}
\end{aligned}
\right.,
$$
where $\varphi(x) = \frac{1}{2\sigma}\exp\left( -\frac{|x|}{\sigma}\right)$ is the PDF of double exponential distribution, and the CDF is 
$$\Phi(x) = \left\{
\begin{aligned}
    \frac 12 \exp\left( \frac{x}{b} \right), &\quad x \le 0 \\
    1 - \frac12 \exp\left( - \frac{x}{b} \right), &\quad x>0
\end{aligned}
\right..$$

We first define a helper integration and study its computation.
Define $I(c,k) =\int_0^cx^ke^{-x}dx$.
Let $u = x^k$ and $dv = e^{-x}dx$, and therefore $du = kx^{k-1}$ and $v = -e^{-x}$.
By applying integration by parts,
$$I(c,k) = \int_0^cudv = [uv]_0^c - \int_0^cvdu = kI(c,k-1) - c^ke^{-c}.$$
We have a recursive relation for computing $I(c,k)$. 
The base case can be solved as follows:
$$I(c,0) = \int_0^c e^{-x}dx = 1 - e^{-c}.$$

Then we use $I(c,k)$ as a tool for computing $\mathbb E[x^k]$.
We first discuss the case when $a < 0 < b$.
Let $I_1 = \int_a^0x^k \varphi(x)dx$ and $I_2 = \int_0^b x^k\varphi(x)dx$, and the $k$-th moment of $x$ is given by
$$\mathbb E[x^k] = \frac{1}{\Phi(b) - \Phi(a)} \int_a^bx^k\varphi(x)dx = \frac{1}{\Phi(b) - \Phi(a)}(I_1 + I_2).$$
$I_1$ can be computed as
$$\begin{aligned}
I_1 &=\int_a^0 x^k\varphi(x)dx = \frac{1}{2\sigma}\int_a^0x^k\exp\left(\frac{x}{\sigma}\right)dx\\
&=\frac12(-1)^{k+1}\sigma^k\int_{-a/\sigma}^0u^k\exp(-u)du \quad (u \triangleq-\frac{x}{\sigma})\\
&=\frac12(-\sigma)^k\int_0^{-a/\sigma}u^k\exp(-u)du = \frac12(-\sigma)^k\cdot I\left(-\frac{a}{\sigma}, k\right).
\end{aligned}$$
$I_2$ can be similarly computed as 
$$I_2 = \frac12 \sigma^k \cdot I\left(\frac{b}{\sigma}, k\right).$$
By computing the corresponding $I(c,k)$ values and plugging them back, we can obtain the result of $\mathbb E[x^k]$.

When $0<a<b$, define and compute $I_1, I_2$ as follows.
\begin{itemize}
    \item $I_1 \triangleq \int_0^bx^k\varphi(x)dx = \frac12\sigma^k\cdot I\left(\frac{b}{\sigma},k\right)$.
    \item $I_2 \triangleq \int_0^ax^k\varphi(x)dx = \frac12\sigma^k\cdot I\left(\frac{a}{\sigma},k\right)$.
\end{itemize}
Thus, $\mathbb E[x^k] = \frac{1}{\Phi(b) - \Phi(a)}(I_1-I_2)$.

When $a<b<0$, define and compute $I_1, I_2$ as follows.
\begin{itemize}
    \item $I_1 \triangleq \int_a^0x^k\varphi(x)dx = \frac12(-\sigma)^k\cdot I\left(-\frac{a}{\sigma},k\right)$.
    \item $I_2 \triangleq \int_b^0x^k\varphi(x)dx = \frac12(-\sigma)^k\cdot I\left(-\frac{b}{\sigma},k\right)$.
\end{itemize}
Thus, $\mathbb E[x^k] = \frac{1}{\Phi(b) - \Phi(a)}(I_1-I_2)$.

Again, we only talk about the computation of reasoning about positive and negative parts when $[a,b]$ crosses zero. The derivations are similar as for the original variable, and the result is as follows.
$$E[x_+^k] = \frac{1}{\Phi(b) - \Phi(a)} \int_0^bx^kf(x)dx = \frac{1}{\Phi(b) - \Phi(a)} \cdot \left( \frac12 \sigma^k \cdot I\left( \frac{b}{\sigma}, k\right) \right).$$
$$E[x_-^k] = \frac{1}{\Phi(b) - \Phi(a)}\int_a^0(-x)^kf(x)dx = \frac{1}{\Phi(b) - \Phi(a)} \cdot \left( \frac12 \sigma^k \cdot I\left( -\frac{a}{\sigma}, k\right) \right).$$

\section{Full Proofs from Section~\ref{sec:algorithm}}
\label{sec:fullproof}

\noindent{\sc Full Proof of Theorem~\ref{thm:posnegcorrect}.}
We first establish the following two facts after the execution reaches line 11:
(a) $h_i(\mathbf x)$ is equivalent to the original input expression; 
(b) each term in $h_i(\mathbf x)$ has the same sign as its coefficient (the constant factor of the term).
Fact (a) is evident because the only two possible modifications to $h_i(\mathbf x)$ are replacing $x_j^a$ by either $(x_j^+)^a - (x_j^-)^a$ or $(-x_j^-)^a$.
Given that $x_j = x_j^+ - x_j^-$ and $x_j^{a-1}$ is non-negative for odd values of $a$, $x_j^{a} = x_j^{a-1}x_j^+ - x_j^{a-1}x_j^- = (x_j^+)^a - (x_j^-)^a$, and thus the first possible replacement preserves the value.
When $x_j$ is always non-positive, $x_j = x_j^-$ holds, and thus the second possible replacement is an equivalent transformation.
Then we examine the monomial part of each term. Note that all factors fall into one of the three categories: an even power of the original variable, the positive component, or the negative component. 
Each of the three is non-negative, and thus their product can be proved non-negative.
Thus, the monomial part of all terms in $h_i(\mathbf x)$ is non-negative, and fact (b) immediately follows.

Then we analyze the values of $h_i^+(\mathbf x)$ and $h_i^-(\mathbf x)$ after the execution of line 22, and show that it has the two properties similar to the decomposition of variables:
\begin{itemize}
    \item Both $h_i^+(\mathbf x)$ and $h_i^-(\mathbf x)$ are non-negative.
    When a term $t$ has positive coefficient, indicating that $t$ is non-negative (by fact (b)), it is added to $h_i^+$, thereby preserving the non-negativity of $h_i^+$.
    On the other hand, when $t$ has negative coefficient, the term $t$ is non-positive. Subtracting such a term form $h_i^-$ ensures that its non-negativity is maintained.
    \item $h_i(\mathbf x) = h_i^+(\mathbf x) - h_i^-(\mathbf x)$. Each term in $h_i(\mathbf x)$ is assigned either to $h_i^+(\mathbf x)$ as it appears, or to $h_i^-(\mathbf x)$ as its negation. Thus, subtracting $h_i^-(\mathbf x)$ from $h_i^+(\mathbf x)$ reconstructs the original value of $h_i(\mathbf x)$, thereby verifying the equality.
\end{itemize}

Therefore, we have the following relaxation:
$$ \begin{aligned}
|h_i(\mathbf x)e_i| &= |h_i^+(\mathbf x)e_i  - h_i^-(\mathbf x)e_i| 
\le |h_i^+(\mathbf x)e_i| + |h_i^-(\mathbf x)e_i|  \\
&= (h_i^+(\mathbf x) + h_i^-(\mathbf x))|e_i|\text{}
\le (h_i^+(\mathbf x) + h_i^-(\mathbf x))\epsilon.
\end{aligned}$$
The second inequality is an application of the triangular inequality. 
The third equality is because $h_i^+(\mathbf x)$ and $h_i^-(\mathbf x)$ are non-negative.
At the end of the algorithm, $p(\mathbf x) = \sum_{i=1}^k(h_i^+(\mathbf x) + h_i^-(\mathbf x))$.
Thus, for the entire first-order term we have
$$\sum_{i=1}^k|h_i(\mathbf x)e_i| \le \sum_{i=1}^k(h_i^+(\mathbf x) + h_i^-(\mathbf x))\epsilon = p(\mathbf x) \cdot \epsilon.$$
\qed

{\sc Full Proof of Lemma~\ref{lem:fractionalpn}.}
Given that $f(\mathbf x)$ has the form $N(\mathbf x) / Q(\mathbf x)$, its floating-point model should observe the following structure: 
$\tilde f(\mathbf x, \mathbf e, \mathbf d) = {\tilde N(\mathbf x, \mathbf e, \mathbf d)}/{\tilde Q(\mathbf x, \mathbf e, \mathbf d)} \cdot (1 + e_k) + d_k$, 
where $k$ is the total number of arithmetic operations in $f(\mathbf x)$ and {$e_k, d_k$ are the error variables for the outmost division}.
We observe the fact that each relative error term $e_i$ appears exactly once in the symbolic expression of $\tilde f(\mathbf x, \mathbf e, \mathbf d)$. More specifically, each $e_i$ contributes to one of the following components: (i) the floating-point model of the numerator $\tilde N(\mathbf x, \mathbf e, \mathbf d)$; (ii) the floating-point model of the denominator $\tilde Q(\mathbf x, \mathbf e, \mathbf d)$; or (iii) the error term $e_k$ related to division $N(\mathbf x) / Q(\mathbf x)$.
To demonstrate that for any index $i$, $\frac{\partial \tilde f}{\partial e_i}(\mathbf x, \mathbf 0, \mathbf 0) \cdot (Q(\mathbf x))^2$ can be reduced into polynomial form, we analyze the three cases. 
Throughout, we make frequent use of the identity $\tilde f(\mathbf x, \mathbf 0, \mathbf 0) = f(\mathbf x)$.
\begin{enumerate}[(i)]
    \item When $e_i$ appears in the numerator $\tilde N(\mathbf x, \mathbf e, \mathbf d)$, we have $\frac{\partial \tilde f}{\partial e_i}(\mathbf x, \mathbf 0, \mathbf 0) \cdot (Q(\mathbf x))^2 = \frac{1}{Q(\mathbf x)} \cdot \frac{\partial \tilde N}{\partial e_i}(\mathbf x, \mathbf 0, \mathbf 0)  \cdot (Q(\mathbf x))^2 = \frac{\partial \tilde N}{\partial e_i}(\mathbf x, \mathbf 0, \mathbf 0)  \cdot Q(\mathbf x)$ which is division-free since $N(\mathbf x)$ and $Q(\mathbf x)$ are both division-free.
    \item When $e_i$ appears in the denominator $\tilde Q(\mathbf x, \mathbf e, \mathbf d)$, we have 
    $\frac{\partial \tilde f}{\partial e_i}(\mathbf x, \mathbf 0, \mathbf 0) \cdot (Q(\mathbf x))^2 =
    \frac{- N(\mathbf x)}{(Q(\mathbf x))^2} \cdot \frac{\partial \tilde Q}{\partial e_i}(\mathbf x, \mathbf 0, \mathbf 0) \cdot (Q(\mathbf x))^2 = - N(\mathbf x) \cdot \frac{\partial \tilde Q}{\partial e_i}(\mathbf x, \mathbf 0, \mathbf 0)$, which is also division-free.
    \item When $e_i$ is just $e_k$, then $\frac{\partial \tilde f}{\partial e_i}(\mathbf x, \mathbf 0, \mathbf 0) \cdot (Q(\mathbf x))^2  = \frac{N(\mathbf x)}{Q(\mathbf x)} \cdot (Q(\mathbf x))^2 = N(\mathbf x) \cdot Q(\mathbf x)$, clearly division-free.
\end{enumerate}
Therefore, in each of the three possible cases, $h_i(\mathbf x) = \frac{\partial \tilde f}{\partial e_i}(\mathbf x, \mathbf 0, \mathbf 0) \cdot (Q(\mathbf x))^2$ is division-free, and is thus reducible to a polynomial expression.
\qed

\section{Monotonicity of $\flag_u$ in Algorithm~\ref{alg:cmb} with regard to $u$}
\label{sec:flagmono}

As stated in Theorem~\ref{thm:cmb-flag}, $\flag_u$ is defined by
$$\flag_u = \frac{\mathbb E[K_u(\mathbf x) - \mu)^n]}{\mu^n},$$
where $K_u(\mathbf x) = p(\mathbf x) \cdot \epsilon - u$ and $\mu = \mathbb E[K_u(\mathbf x)]$.
We perform the following transformations on $\flag_u$:
$$\begin{aligned}
\flag_u &= \frac{\mathbb E[K_u(\mathbf x) - \mu)^n]}{\mu^n} = \frac{\mathbb E[(p(\mathbf x) \cdot \epsilon - \mathbb E[p(\mathbf x) \cdot \epsilon] + \mathbb E[u])^n]}{(\mathbb E[p(\mathbf x) \cdot \epsilon] - u)^n}\\ 
&=\frac{\mathbb E[(p(\mathbf x) \cdot \epsilon - \mathbb E[p(\mathbf x) \cdot \epsilon])^n]}{(\mathbb E[p(\mathbf x) \cdot \epsilon] - u)^n}.
\end{aligned}$$
Observe that the numerator is constant with respect to $u$, non-negative since $n\geq 0$, while the denominator is monotonically increasing in $u$ when $u > \mathbb E[p(\mathbf x) \cdot \epsilon] = \ell$. 
Therefore, $\flag_u$ is monotonically decreasing when $u > \ell$.

\section{Brute-Force algorithm for bounding the second-order error}
\label{sec:bruteforce}
The second-order error expressions we aim to bound do not involve division by non-constant terms. 
Therefore, for a given expression $f(\mathbf x)$, we derive a deterministic upper bound $\text{bound}(f)$ for $|f|$ via structural recursion and pattern matching on the syntactic form of $f$:
\begin{itemize}
    \item If $f(\mathbf x) = c$, where $c$ is a constant, then $\text{bound}(f) = |c|$.
    \item If $f(\mathbf x) = x$, where $x$ is a variable with range $[a, b]$, then $\text{bound}(f) = \max\{|a|, |b|\}$.
    \item If $f(\mathbf x) = f_1(\mathbf x) + f_2(\mathbf x)$, then $\text{bound}(f) = \text{bound}(f_1) + \text{bound}(f_2)$.
    \item If $f(\mathbf x) = f_1(\mathbf x) - f_2(\mathbf x)$, then $\text{bound}(f) = \text{bound}(f_1) + \text{bound}(f_2)$. %\jiawei{Possible typo $f-2(x)$ corrected to $f_2(x)$}
    \item If $f(\mathbf x) = f_1(\mathbf x)\cdot f_2(\mathbf x)$, then $\text{bound}(f) = \text{bound}(f_1) \cdot \text{bound}(f_2)$.
    \item If $f(\mathbf x) = f_1(\mathbf x) / c$, where $c$ is a constant, then $\text{bound}(f) = \text{bound}(f_1) / |c|$.
\end{itemize}
The soundness of this algorithm follows directly from applications of triangular inequality and basic properties of absolute values under arithmetic operations.

\section{Full Experimental Results}
\label{sec:fullexperiment}

Due to space limitations, full experimental results are recorded here.
Tables~\ref{tab:fulluniform}, \ref{tab:fullnormal}, and \ref{tab:fulllaplace} report thresholds and running times of PAF and PrAn on benchmarks with uniform distribution, truncated standard normal distribution, and truncated double exponential distribution, respectively.
Table~\ref{tab:cmbpolytime} records the running time of ablation experiments on polynomial benchmarks.
Table~\ref{tab:fulldiv} records the thresholds and running times of PAF and PrAn on fractional benchmarks. PAF crashes on the \texttt{nonlin2} benchmarks. PrAn crashes on the \texttt{doppler} benchmarks, and does not support double exponential distribution.

\begin{remark}
Further refinement of the algorithm can be achieved through the use of SMT solvers.
Specifically, when analyzing each sub-region $B$, we may submit the logical query $\mathbf x \in B \Rightarrow K_u(\mathbf x) >0$ to an SMT solver. 
If the solver returns \texttt{UNSAT}, this indicates that $K_u(\mathbf x) \le 0$ holds within the sub-region $B$.
In such cases, we can soundly conclude that the violation probability over $B$ is zero, and safely assign $\flag_{u, B} := 0$. \qed 
\end{remark}
% Table~\ref{tab:widenedrangebound} reports thresholds produced by ProbTaylor with $99\%$ confidence level on benchmarks with widened input ranges.
% Table~\ref{tab:widenedfptaylor} reports thresholds produced by FPTaylor that are satisfied deterministically on benchmarks with widened input ranges.
% Table~\ref{tab:uniexpwidened} compares the threshold given by ProbTaylor to those given by FPTaylor on benchmarks with widened input ranges, with uniform input distribution and double exponential input distribution.
% Table~\ref{tab:expfull}, Table~\ref{tab:uniform2full}, Table~\ref{tab:normal2full}, and Table~\ref{tab:exp2full} are the full versions of Table~\ref{tab:exp}, Table~\ref{tab:uniform2}, Table~\ref{tab:normal2}, and Table~\ref{tab:exp2}, respectively.
% Table~\ref{tab:exp0.01} records the experimental results when the input variables follows double exponential distribution with $\sigma=0.01$. 
% The PAF thresholds are directly copied from \cite{constantinides2021rigorous}. 
% The NaNs produced by ProbTaylor are due to division by zero caused by computing the difference of two nearly identical numbers in OCaml.

\begin{table}[ht]
\centering
\caption{Experimental results on PAF and PrAn for uniform distribution. The thresholds are produced assuming computations are done in single precision and with a $99\%$ target confidence level. A timeout of 1 hour is enforced for all benchmarks, except for \texttt{ksin} and \texttt{kcos}, where the timeout is extended to 4 hours.}
\label{tab:fulluniform}
\begin{tabular}{|c|cccc|}
\hline
 & \multicolumn{4}{c|}{Uniform Distribution} \\ \hline
 & \multicolumn{2}{c|}{PAF} & \multicolumn{2}{c|}{PrAn} \\ \hline
Benchmark & \multicolumn{1}{c|}{Time (s)} & \multicolumn{1}{c|}{Threshold} & \multicolumn{1}{c|}{Time (s)} & Threshold \\ \hline
bsplines0 & \multicolumn{1}{c|}{1,232} & \multicolumn{1}{c|}{5.71E-08} & \multicolumn{1}{c|}{18.3} & 8.69E-08 \\ \hline
bsplines1 & \multicolumn{1}{c|}{1,636} & \multicolumn{1}{c|}{1.86E-07} & \multicolumn{1}{c|}{27.91} & 1.89E-07 \\ \hline
bsplines2 & \multicolumn{1}{c|}{2,748} & \multicolumn{1}{c|}{1.94E-07} & \multicolumn{1}{c|}{36.03} & 2.12E-07 \\ \hline
bsplines3 & \multicolumn{1}{c|}{927} & \multicolumn{1}{c|}{4.22E-08} & \multicolumn{1}{c|}{15.32} & 5.71E-08 \\ \hline
classids0 & \multicolumn{1}{c|}{Timeout} & \multicolumn{1}{c|}{6.93E-06} & \multicolumn{1}{c|}{106.36} & 8.65E-06 \\ \hline
classids1 & \multicolumn{1}{c|}{Timeout} & \multicolumn{1}{c|}{3.71E-06} & \multicolumn{1}{c|}{145.9} & 4.69E-06 \\ \hline
classids2 & \multicolumn{1}{c|}{Timeout} & \multicolumn{1}{c|}{5.23E-06} & \multicolumn{1}{c|}{120.77} & 7.58E-06 \\ \hline
filters1 & \multicolumn{1}{c|}{608} & \multicolumn{1}{c|}{1.25E-07} & \multicolumn{1}{c|}{3.58} & 2.03E-07 \\ \hline
filters2 & \multicolumn{1}{c|}{3,186} & \multicolumn{1}{c|}{7.93E-07} & \multicolumn{1}{c|}{65.21} & 1.01E-06 \\ \hline
filters3 & \multicolumn{1}{c|}{Timeout} & \multicolumn{1}{c|}{2.34E-06} & \multicolumn{1}{c|}{212.15} & 2.86E-06 \\ \hline
filters4 & \multicolumn{1}{c|}{Timeout} & \multicolumn{1}{c|}{4.15E-06} & \multicolumn{1}{c|}{429.3} & 5.20E-06 \\ \hline
rigidbody1 & \multicolumn{1}{c|}{Timeout} & \multicolumn{1}{c|}{1.74E-04} & \multicolumn{1}{c|}{27.55} & 1.73E-04 \\ \hline
rigidbody2 & \multicolumn{1}{c|}{Timeout} & \multicolumn{1}{c|}{1.96E-02} & \multicolumn{1}{c|}{46.5} & 9.70E-03 \\ \hline
sine & \multicolumn{1}{c|}{Timeout} & \multicolumn{1}{c|}{2.37E-07} & \multicolumn{1}{c|}{132.95} & 2.41E-07 \\ \hline
solvecubic & \multicolumn{1}{c|}{Timeout} & \multicolumn{1}{c|}{1.78E-05} & \multicolumn{1}{c|}{180.59} & 2.01E-05 \\ \hline
sqrt & \multicolumn{1}{c|}{Timeout} & \multicolumn{1}{c|}{1.54E-04} & \multicolumn{1}{c|}{56.01} & 1.54E-04 \\ \hline
traincars1 & \multicolumn{1}{c|}{3,434} & \multicolumn{1}{c|}{1.76E-03} & \multicolumn{1}{c|}{50.27} & 1.96E-03 \\ \hline
traincars2 & \multicolumn{1}{c|}{Timeout} & \multicolumn{1}{c|}{1.04E-03} & \multicolumn{1}{c|}{205.4} & 1.33E-03 \\ \hline
traincars3 & \multicolumn{1}{c|}{Timeout} & \multicolumn{1}{c|}{1.75E-02} & \multicolumn{1}{c|}{18.96} & 2.29E-02 \\ \hline
traincars4 & \multicolumn{1}{c|}{Timeout} & \multicolumn{1}{c|}{1.81E-01} & \multicolumn{1}{c|}{26.36} & 2.30E-01 \\ \hline
trid1 & \multicolumn{1}{c|}{Timeout} & \multicolumn{1}{c|}{6.01E-03} & \multicolumn{1}{c|}{69.03} & 6.12E-03 \\ \hline
trid2 & \multicolumn{1}{c|}{Timeout} & \multicolumn{1}{c|}{1.03E-02} & \multicolumn{1}{c|}{59.48} & 1.17E-02 \\ \hline
trid3 & \multicolumn{1}{c|}{Timeout} & \multicolumn{1}{c|}{1.75E-02} & \multicolumn{1}{c|}{312.22} & 1.95E-02 \\ \hline
trid4 & \multicolumn{1}{c|}{Timeout} & \multicolumn{1}{c|}{2.69E-02} & \multicolumn{1}{c|}{1018.98} & 2.88E-02 \\ \hline
ksin & \multicolumn{1}{c|}{Timeout} & \multicolumn{1}{c|}{N/A} & \multicolumn{1}{c|}{2049.07} & 6.93E-08 \\ \hline
kcos & \multicolumn{1}{c|}{Timeout} & \multicolumn{1}{c|}{N/A} & \multicolumn{1}{c|}{1838.64} & 1.20E-07 \\ \hline
\end{tabular}
\end{table}

\begin{table}[ht]
\centering
\caption{Experimental results on PAF and PrAn for truncated normal distribution. The thresholds are produced assuming computations are done in single precision and with a $99\%$ target confidence level. A timeout of 1 hour is enforced for all benchmarks, except for \texttt{ksin} and \texttt{kcos}, where the timeout is extended to 4 hours.}
\label{tab:fullnormal}
\begin{tabular}{|c|cccc|}
\hline
 & \multicolumn{4}{c|}{Truncated standard normal distribution} \\ \hline
 & \multicolumn{2}{c|}{PAF} & \multicolumn{2}{c|}{PrAn} \\ \hline
Benchmark & \multicolumn{1}{c|}{Time (s)} & \multicolumn{1}{c|}{Threshold} & \multicolumn{1}{c|}{Time (s)} & Threshold \\ \hline
bsplines0 & \multicolumn{1}{c|}{1,227} & \multicolumn{1}{c|}{5.71E-08} & \multicolumn{1}{c|}{28.99} & 8.69E-08 \\ \hline
bsplines1 & \multicolumn{1}{c|}{1,707} & \multicolumn{1}{c|}{1.86E-07} & \multicolumn{1}{c|}{53.31} & 2.10E-07 \\ \hline
bsplines2 & \multicolumn{1}{c|}{2,778} & \multicolumn{1}{c|}{1.94E-07} & \multicolumn{1}{c|}{62.68} & 2.12E-07 \\ \hline
bsplines3 & \multicolumn{1}{c|}{914} & \multicolumn{1}{c|}{4.22E-08} & \multicolumn{1}{c|}{22.45} & 5.71E-08 \\ \hline
classids0 & \multicolumn{1}{c|}{Timeout} & \multicolumn{1}{c|}{4.45E-06} & \multicolumn{1}{c|}{279.69} & 8.90E-06 \\ \hline
classids1 & \multicolumn{1}{c|}{Timeout} & \multicolumn{1}{c|}{2.68E-06} & \multicolumn{1}{c|}{297.81} & 4.76E-06 \\ \hline
classids2 & \multicolumn{1}{c|}{Timeout} & \multicolumn{1}{c|}{3.85E-06} & \multicolumn{1}{c|}{271.49} & 7.60E-06 \\ \hline
filters1 & \multicolumn{1}{c|}{581} & \multicolumn{1}{c|}{1.24E-07} & \multicolumn{1}{c|}{5.52} & 2.03E-07 \\ \hline
filters2 & \multicolumn{1}{c|}{3,220} & \multicolumn{1}{c|}{6.13E-07} & \multicolumn{1}{c|}{121.31} & 1.01E-06 \\ \hline
filters3 & \multicolumn{1}{c|}{Timeout} & \multicolumn{1}{c|}{2.05E-06} & \multicolumn{1}{c|}{293.62} & 2.87E-06 \\ \hline
filters4 & \multicolumn{1}{c|}{Timeout} & \multicolumn{1}{c|}{4.15E-06} & \multicolumn{1}{c|}{826.53} & 5.20E-06 \\ \hline
rigidbody1 & \multicolumn{1}{c|}{Timeout} & \multicolumn{1}{c|}{6.14E-06} & \multicolumn{1}{c|}{45.43} & 1.73E-04 \\ \hline
rigidbody2 & \multicolumn{1}{c|}{Timeout} & \multicolumn{1}{c|}{5.99E-05} & \multicolumn{1}{c|}{95.76} & 9.70E-03 \\ \hline
sine & \multicolumn{1}{c|}{Timeout} & \multicolumn{1}{c|}{2.37E-07} & \multicolumn{1}{c|}{475.34} & 2.41E-07 \\ \hline
solvecubic & \multicolumn{1}{c|}{Timeout} & \multicolumn{1}{c|}{6.84E-06} & \multicolumn{1}{c|}{230.44} & 2.02E-05 \\ \hline
sqrt & \multicolumn{1}{c|}{Timeout} & \multicolumn{1}{c|}{2.46E-07} & \multicolumn{1}{c|}{115.91} & 1.54E-04 \\ \hline
traincars1 & \multicolumn{1}{c|}{3,007} & \multicolumn{1}{c|}{8.26E-04} & \multicolumn{1}{c|}{93.36} & 1.96E-03 \\ \hline
traincars2 & \multicolumn{1}{c|}{Timeout} & \multicolumn{1}{c|}{3.62E-04} & \multicolumn{1}{c|}{476.4} & 1.33E-03 \\ \hline
traincars3 & \multicolumn{1}{c|}{Timeout} & \multicolumn{1}{c|}{9.56E-03} & \multicolumn{1}{c|}{19.18} & 2.23E-02 \\ \hline
traincars4 & \multicolumn{1}{c|}{Timeout} & \multicolumn{1}{c|}{8.87E-02} & \multicolumn{1}{c|}{26.86} & 2.30E-01 \\ \hline
trid1 & \multicolumn{1}{c|}{Timeout} & \multicolumn{1}{c|}{1.58E-05} & \multicolumn{1}{c|}{253.98} & 6.06E-03 \\ \hline
trid2 & \multicolumn{1}{c|}{Timeout} & \multicolumn{1}{c|}{2.42E-05} & \multicolumn{1}{c|}{140.63} & 1.17E-02 \\ \hline
trid3 & \multicolumn{1}{c|}{Timeout} & \multicolumn{1}{c|}{6.80E-05} & \multicolumn{1}{c|}{537.01} & 1.95E-02 \\ \hline
trid4 & \multicolumn{1}{c|}{Timeout} & \multicolumn{1}{c|}{2.64E-04} & \multicolumn{1}{c|}{1576.52} & 3.03E-02 \\ \hline
ksin & \multicolumn{1}{c|}{Timeout} & \multicolumn{1}{c|}{N/A} & \multicolumn{1}{c|}{Timeout} & N/A \\ \hline
kcos & \multicolumn{1}{c|}{Timeout} & \multicolumn{1}{c|}{N/A} & \multicolumn{1}{c|}{Timeout} & N/A \\ \hline
\end{tabular}
\end{table}

\begin{table}[ht]
\centering
\caption{Experimental results on PAF for truncated double exponential distribution. The thresholds are produced assuming computations are done in single precision and with a $99\%$ target confidence level. A timeout of 1 hour is enforced for all benchmarks, except for \texttt{ksin} and \texttt{kcos}, where the timeout is extended to 4 hours. For benchmarks \texttt{classids0}, \texttt{classids1}, \texttt{classids2}, and \texttt{solvecubic}, the variance value is set to $1$. For all other benchmarks, the variance value is set to $0.01$.}
\label{tab:fulllaplace}
\begin{tabular}{|c|cc|}
\hline
 & \multicolumn{2}{c|}{Truncated double exponential distribution} \\ \hline
 & \multicolumn{2}{c|}{PAF} \\ \hline
Benchmark & \multicolumn{1}{c|}{Time (s)} & Threshold \\ \hline
bsplines0 & \multicolumn{1}{c|}{2,629} & 5.71E-08 \\ \hline
bsplines1 & \multicolumn{1}{c|}{2,201} & 6.95E-08 \\ \hline
bsplines2 & \multicolumn{1}{c|}{Timeout} & 2.11E-08 \\ \hline
bsplines3 & \multicolumn{1}{c|}{938} & 7.62E-12 \\ \hline
classids0 & \multicolumn{1}{c|}{24,606} & 5.00E-06 \\ \hline
classids1 & \multicolumn{1}{c|}{16,145} & 3.08E-06 \\ \hline
classids2 & \multicolumn{1}{c|}{13,698} & 4.15E-06 \\ \hline
filters1 & \multicolumn{1}{c|}{733} & 5.43E-09 \\ \hline
filters2 & \multicolumn{1}{c|}{Timeout} & 2.90E-08 \\ \hline
filters3 & \multicolumn{1}{c|}{Timeout} & 1.09E-07 \\ \hline
filters4 & \multicolumn{1}{c|}{Timeout} & 4.61E-07 \\ \hline
rigidbody1 & \multicolumn{1}{c|}{Timeout} & 4.80E-07 \\ \hline
rigidbody2 & \multicolumn{1}{c|}{Timeout} & 9.55E-07 \\ \hline
sine & \multicolumn{1}{c|}{Timeout} & 1.49E-08 \\ \hline
solvecubic & \multicolumn{1}{c|}{16,795} & 1.42E-05 \\ \hline
sqrt & \multicolumn{1}{c|}{Timeout} & 2.46E-07 \\ \hline
traincars1 & \multicolumn{1}{c|}{Timeout} & 4.50E-04 \\ \hline
traincars2 & \multicolumn{1}{c|}{Timeout} & 2.83E-05 \\ \hline
traincars3 & \multicolumn{1}{c|}{Timeout} & 8.95E-04 \\ \hline
traincars4 & \multicolumn{1}{c|}{Timeout} & 7.33E-03 \\ \hline
trid1 & \multicolumn{1}{c|}{Timeout} & 1.58E-05 \\ \hline
trid2 & \multicolumn{1}{c|}{Timeout} & 2.43E-05 \\ \hline
trid3 & \multicolumn{1}{c|}{Timeout} & 6.77E-05 \\ \hline
trid4 & \multicolumn{1}{c|}{Timeout} & 2.64E-04 \\ \hline
ksin & \multicolumn{1}{c|}{Timeout} & N/A \\ \hline
kcos & \multicolumn{1}{c|}{Timeout} & N/A \\ \hline
\end{tabular}
\end{table}

\begin{table}[ht]
\centering
\caption{Ablation experiments on polynomial benchmarks. Columns labeled ``NM'' report thresholds generated by the Naive Markov method, while columns labeled ``CMB'' report thresholds generated by the Central-Moment-Based algorithm with range partition refinement. 
The experiments use an analysis order of $n=4$. For the Central-Moment-Based algorithm, the number of partitions is $b=8$ for benchmarks with fewer than four input variables and fewer than ten operations, and range partitioning is disabled otherwise.}
\label{tab:rq3polyfull}
\begin{tabular}{|c|cc|cc|cc|}
\hline
 & \multicolumn{2}{c|}{Uniform} & \multicolumn{2}{c|}{Normal} & \multicolumn{2}{c|}{Laplace} \\ \hline
Benchmark & \multicolumn{1}{c|}{NM} & CMB & \multicolumn{1}{c|}{NM} & CMB & \multicolumn{1}{c|}{NM} & CMB \\ \hline
bsplines0 & \multicolumn{1}{c|}{9.44E-07} & 4.96E-07 & \multicolumn{1}{c|}{8.95E-07} & 4.98E-07 & \multicolumn{1}{c|}{8.57E-07} & 4.81E-07 \\ \hline
bsplines1 & \multicolumn{1}{c|}{1.03E-06} & 5.72E-07 & \multicolumn{1}{c|}{9.66E-07} & 2.33E-07 & \multicolumn{1}{c|}{9.23E-07} & 5.63E-07 \\ \hline
bsplines2 & \multicolumn{1}{c|}{1.01E-06} & 5.76E-07 & \multicolumn{1}{c|}{9.49E-07} & 2.15E-07 & \multicolumn{1}{c|}{9.06E-07} & 5.76E-07 \\ \hline
bsplines3 & \multicolumn{1}{c|}{6.62E-08} & 4.13E-08 & \multicolumn{1}{c|}{6.18E-08} & 8.73E-09 & \multicolumn{1}{c|}{5.89E-08} & 4.14E-08 \\ \hline
classids0 & \multicolumn{1}{c|}{2.41E-05} & 1.46E-05 & \multicolumn{1}{c|}{9.35E-06} & 4.79E-06 & \multicolumn{1}{c|}{1.16E-05} & 7.10E-06 \\ \hline
classids1 & \multicolumn{1}{c|}{1.33E-05} & 8.29E-06 & \multicolumn{1}{c|}{5.12E-06} & 2.81E-06 & \multicolumn{1}{c|}{6.42E-06} & 4.04E-06 \\ \hline
classids2 & \multicolumn{1}{c|}{2.13E-05} & 1.23E-05 & \multicolumn{1}{c|}{9.18E-06} & 4.71E-06 & \multicolumn{1}{c|}{1.13E-05} & 6.58E-06 \\ \hline
filters1 & \multicolumn{1}{c|}{1.76E-07} & 9.07E-08 & \multicolumn{1}{c|}{1.44E-07} & 3.12E-08 & \multicolumn{1}{c|}{1.45E-07} & 8.61E-08 \\ \hline
filters2 & \multicolumn{1}{c|}{1.61E-06} & 8.22E-07 & \multicolumn{1}{c|}{1.27E-06} & 3.18E-07 & \multicolumn{1}{c|}{1.26E-06} & 7.82E-07 \\ \hline
filters3 & \multicolumn{1}{c|}{5.09E-06} & 2.63E-06 & \multicolumn{1}{c|}{3.92E-06} & 1.06E-06 & \multicolumn{1}{c|}{3.86E-06} & 2.30E-06 \\ \hline
filters4 & \multicolumn{1}{c|}{1.11E-05} & 7.23E-06 & \multicolumn{1}{c|}{8.48E-06} & 5.89E-06 & \multicolumn{1}{c|}{8.32E-06} & 5.91E-06 \\ \hline
rigidbody1 & \multicolumn{1}{c|}{2.63E-04} & 1.81E-04 & \multicolumn{1}{c|}{4.56E-06} & 7.40E-07 & \multicolumn{1}{c|}{1.13E-05} & 1.03E-05 \\ \hline
rigidbody2 & \multicolumn{1}{c|}{1.78E-02} & 1.67E-02 & \multicolumn{1}{c|}{2.35E-05} & 2.29E-05 & \multicolumn{1}{c|}{1.24E-04} & 1.24E-04 \\ \hline
sine & \multicolumn{1}{c|}{1.01E-06} & 6.02E-07 & \multicolumn{1}{c|}{8.62E-07} & 1.61E-07 & \multicolumn{1}{c|}{9.33E-07} & 9.33E-07 \\ \hline
solvecubic & \multicolumn{1}{c|}{3.72E-05} & 1.98E-05 & \multicolumn{1}{c|}{9.70E-06} & 3.14E-06 & \multicolumn{1}{c|}{1.97E-05} & 1.32E-05 \\ \hline
sqrt & \multicolumn{1}{c|}{2.22E-04} & 1.50E-04 & \multicolumn{1}{c|}{2.98E-06} & 4.78E-07 & \multicolumn{1}{c|}{3.98E-05} & 9.76E-06 \\ \hline
traincars1 & \multicolumn{1}{c|}{5.18E-03} & 2.48E-03 & \multicolumn{1}{c|}{1.72E-03} & 5.27E-04 & \multicolumn{1}{c|}{2.10E-03} & 1.23E-03 \\ \hline
traincars2 & \multicolumn{1}{c|}{3.56E-03} & 2.31E-03 & \multicolumn{1}{c|}{6.29E-04} & 4.51E-04 & \multicolumn{1}{c|}{9.09E-04} & 7.42E-04 \\ \hline
traincars3 & \multicolumn{1}{c|}{4.43E-02} & 2.87E-02 & \multicolumn{1}{c|}{8.99E-03} & 6.28E-03 & \multicolumn{1}{c|}{1.22E-02} & 9.35E-03 \\ \hline
traincars4 & \multicolumn{1}{c|}{4.94E-01} & 3.01E-01 & \multicolumn{1}{c|}{9.62E-02} & 5.95E-02 & \multicolumn{1}{c|}{1.16E-01} & 8.07E-02 \\ \hline
trid1 & \multicolumn{1}{c|}{1.08E-02} & 6.96E-03 & \multicolumn{1}{c|}{8.99E-06} & 1.99E-06 & \multicolumn{1}{c|}{2.21E-05} & 2.86E-06 \\ \hline
trid2 & \multicolumn{1}{c|}{2.04E-02} & 1.56E-02 & \multicolumn{1}{c|}{1.56E-05} & 1.25E-05 & \multicolumn{1}{c|}{3.54E-05} & 3.47E-05 \\ \hline
trid3 & \multicolumn{1}{c|}{3.23E-02} & 2.44E-02 & \multicolumn{1}{c|}{2.37E-05} & 1.83E-05 & \multicolumn{1}{c|}{5.08E-05} & 4.81E-05 \\ \hline
trid4 & \multicolumn{1}{c|}{4.46E-02} & 3.46E-02 & \multicolumn{1}{c|}{3.34E-05} & 2.50E-05 & \multicolumn{1}{c|}{6.82E-05} & 6.35E-05 \\ \hline
ksin & \multicolumn{1}{c|}{1.46E-07} & \multicolumn{1}{r|}{1.13E-07} & \multicolumn{1}{c|}{1.49E-07} & \multicolumn{1}{r|}{1.40E-07} & \multicolumn{1}{c|}{5.02E-07} & 3.32E-07 \\ \hline
kcos & \multicolumn{1}{c|}{4.76E-07} & \multicolumn{1}{r|}{2.37E-07} & \multicolumn{1}{c|}{4.69E-07} & \multicolumn{1}{r|}{2.42E-07} & \multicolumn{1}{c|}{1.42E-06} & 2.30E-07 \\ \hline
\end{tabular}\end{table}

\begin{table}[ht]
\centering
\caption{Running time of ablation experiments on polynomial benchmarks. Both the NM and CMB algorithm use an analysis order of $n=4$. For the CMB algorithm, the number of partitions if $b=8$ for benchmarks with fewer than four input variables and fewer than ten operations.}
\label{tab:cmbpolytime}
\begin{tabular}{|c|cc|cc|cc|}
\hline
 & \multicolumn{2}{c|}{Uniform} & \multicolumn{2}{c|}{Normal} & \multicolumn{2}{c|}{Laplace} \\ \hline
Benchmark & \multicolumn{1}{c|}{\begin{tabular}[c]{@{}c@{}}NM\\ time (s)\end{tabular}} & \begin{tabular}[c]{@{}c@{}}CMB\\ time(s)\end{tabular} & \multicolumn{1}{c|}{\begin{tabular}[c]{@{}c@{}}NM\\ time (s)\end{tabular}} & \begin{tabular}[c]{@{}c@{}}CMB\\ time (s)\end{tabular} & \multicolumn{1}{c|}{\begin{tabular}[c]{@{}c@{}}NM\\ time (s)\end{tabular}} & \begin{tabular}[c]{@{}c@{}}CMB\\ time (s)\end{tabular} \\ \hline
bsplines0 & \multicolumn{1}{c|}{0.59} & 0.67 & \multicolumn{1}{c|}{0.67} & 0.67 & \multicolumn{1}{c|}{0.62} & 0.68 \\ \hline
bsplines1 & \multicolumn{1}{c|}{0.54} & 0.61 & \multicolumn{1}{c|}{0.56} & 0.63 & \multicolumn{1}{c|}{0.55} & 0.63 \\ \hline
bsplines2 & \multicolumn{1}{c|}{1.21} & 1.32 & \multicolumn{1}{c|}{1.21} & 1.29 & \multicolumn{1}{c|}{1.11} & 1.28 \\ \hline
bsplines3 & \multicolumn{1}{c|}{0.25} & 0.3 & \multicolumn{1}{c|}{0.25} & 0.27 & \multicolumn{1}{c|}{0.25} & 0.28 \\ \hline
classids0 & \multicolumn{1}{c|}{1.23} & 1.3 & \multicolumn{1}{c|}{1.21} & 1.34 & \multicolumn{1}{c|}{1.23} & 1.21 \\ \hline
classids1 & \multicolumn{1}{c|}{1.13} & 1.25 & \multicolumn{1}{c|}{1.12} & 1.27 & \multicolumn{1}{c|}{1.23} & 1.3 \\ \hline
classids2 & \multicolumn{1}{c|}{1.24} & 1.36 & \multicolumn{1}{c|}{1.24} & 1.36 & \multicolumn{1}{c|}{1.25} & 1.34 \\ \hline
filters1 & \multicolumn{1}{c|}{0.18} & 0.23 & \multicolumn{1}{c|}{0.17} & 0.23 & \multicolumn{1}{c|}{0.17} & 0.22 \\ \hline
filters2 & \multicolumn{1}{c|}{0.24} & 2.64 & \multicolumn{1}{c|}{0.23} & 1.86 & \multicolumn{1}{c|}{0.23} & 2.07 \\ \hline
filters3 & \multicolumn{1}{c|}{1.08} & 48.62 & \multicolumn{1}{c|}{1.07} & 41.58 & \multicolumn{1}{c|}{1.06} & 75.33 \\ \hline
filters4 & \multicolumn{1}{c|}{6.06} & 6.05 & \multicolumn{1}{c|}{5.51} & 5.94 & \multicolumn{1}{c|}{5.81} & 5.91 \\ \hline
rigidbody1 & \multicolumn{1}{c|}{0.42} & 234.49 & \multicolumn{1}{c|}{0.42} & 235.61 & \multicolumn{1}{c|}{0.43} & 234.79 \\ \hline
rigidbody2 & \multicolumn{1}{c|}{1.28} & 2.97 & \multicolumn{1}{c|}{1.23} & 2.45 & \multicolumn{1}{c|}{1.32} & 1.41 \\ \hline
sine & \multicolumn{1}{c|}{2.35} & 3.03 & \multicolumn{1}{c|}{2.26} & 3.06 & \multicolumn{1}{c|}{2.28} & 2.39 \\ \hline
solvecubic & \multicolumn{1}{c|}{1.03} & 27.28 & \multicolumn{1}{c|}{1.01} & 26.74 & \multicolumn{1}{c|}{0.99} & 26.97 \\ \hline
sqrt & \multicolumn{1}{c|}{0.79} & 0.87 & \multicolumn{1}{c|}{0.76} & 0.9 & \multicolumn{1}{c|}{0.75} & 0.93 \\ \hline
traincars1 & \multicolumn{1}{c|}{0.39} & 10.16 & \multicolumn{1}{c|}{0.38} & 8.73 & \multicolumn{1}{c|}{0.39} & 8.56 \\ \hline
traincars2 & \multicolumn{1}{c|}{0.94} & 1.09 & \multicolumn{1}{c|}{0.9} & 1.22 & \multicolumn{1}{c|}{0.96} & 60.05 \\ \hline
traincars3 & \multicolumn{1}{c|}{3.38} & 4.89 & \multicolumn{1}{c|}{3.31} & 5.05 & \multicolumn{1}{c|}{3.23} & 5.68 \\ \hline
traincars4 & \multicolumn{1}{c|}{6.43} & 10.34 & \multicolumn{1}{c|}{6.06} & 10.98 & \multicolumn{1}{c|}{6.38} & 10.79 \\ \hline
trid1 & \multicolumn{1}{c|}{0.49} & 25.22 & \multicolumn{1}{c|}{0.47} & 22.04 & \multicolumn{1}{c|}{0.46} & 21.82 \\ \hline
trid2 & \multicolumn{1}{c|}{1.5} & 4.67 & \multicolumn{1}{c|}{1.46} & 3.43 & \multicolumn{1}{c|}{1.46} & 3.77 \\ \hline
trid3 & \multicolumn{1}{c|}{4.28} & 13.02 & \multicolumn{1}{c|}{4.16} & 16.24 & \multicolumn{1}{c|}{4.15} & 15.68 \\ \hline
trid4 & \multicolumn{1}{c|}{10.13} & 33.34 & \multicolumn{1}{c|}{10.25} & 34.51 & \multicolumn{1}{c|}{9.94} & 36.35 \\ \hline
ksin & \multicolumn{1}{c|}{7.11} & 7.35 & \multicolumn{1}{c|}{7.27} & 7.32 & \multicolumn{1}{c|}{7.32} & 7.17 \\ \hline
kcos & \multicolumn{1}{c|}{5.49} & 5.9 & \multicolumn{1}{c|}{5.59} & 5.63 & \multicolumn{1}{c|}{5.62} & 5.74 \\ \hline
\end{tabular}
\end{table}

\begin{table}[ht]
\centering
\caption{Experimental results on PAF and PrAn on fractional benchmarks. The settings are the same as in Table~\ref{tab:rq1uniform}. Entries marked ``N/A'' indicate that the corresponding tool either crashed or lacked support for the corresponding distribution (PrAn laplace).}
\label{tab:fulldiv}
\begin{tabular}{|c|c|cc|cc|}
\hline
 &  & \multicolumn{2}{c|}{PAF} & \multicolumn{2}{c|}{PrAn} \\ \hline
Benchmark & Distribution & \multicolumn{1}{c|}{Time (s)} & Threshold & \multicolumn{1}{c|}{Time (s)} & Threshold \\ \hline
doppler1 & uniform & \multicolumn{1}{c|}{Timeout} & 7.95E-05 & \multicolumn{1}{c|}{N/A} & N/A \\ \hline
doppler2 & uniform & \multicolumn{1}{c|}{Timeout} & 1.43E-04 & \multicolumn{1}{c|}{N/A} & N/A \\ \hline
doppler3 & uniform & \multicolumn{1}{c|}{Timeout} & 4.55E-05 & \multicolumn{1}{c|}{N/A} & N/A \\ \hline
nonlin1 & uniform & \multicolumn{1}{c|}{840.21} & 6.71E-08 & \multicolumn{1}{c|}{23.32} & 7.54E-08 \\ \hline
nonlin2 & uniform & \multicolumn{1}{c|}{N/A} & N/A & \multicolumn{1}{c|}{225.24} & 3.66E-06 \\ \hline
predator & uniform & \multicolumn{1}{c|}{972.16} & 9.95E-08 & \multicolumn{1}{c|}{84.93} & 1.00E-07 \\ \hline
verhulst & uniform & \multicolumn{1}{c|}{350.34} & 1.72E-07 & \multicolumn{1}{c|}{30.03} & 1.80E-07 \\ \hline
nonlin1 & normal & \multicolumn{1}{c|}{842.05} & 6.71E-08 & \multicolumn{1}{c|}{15.92} & 7.54E-08 \\ \hline
nonlin2 & normal & \multicolumn{1}{c|}{N/A} & N/A & \multicolumn{1}{c|}{681.51} & 3.08E-06 \\ \hline
predator & normal & \multicolumn{1}{c|}{965.61} & 9.95E-08 & \multicolumn{1}{c|}{173.84} & 1.00E-07 \\ \hline
verhulst & normal & \multicolumn{1}{c|}{353.44} & 1.72E-07 & \multicolumn{1}{c|}{48.75} & 1.80E-07 \\ \hline
nonlin1 & laplace & \multicolumn{1}{c|}{841.33} & 6.71E-08 & \multicolumn{1}{c|}{N/A} & N/A \\ \hline
nonlin2 & laplace & \multicolumn{1}{c|}{N/A} & N/A & \multicolumn{1}{c|}{N/A} & N/A \\ \hline
predator & laplace & \multicolumn{1}{c|}{957.49} & 9.89E-08 & \multicolumn{1}{c|}{N/A} & N/A \\ \hline
verhulst & laplace & \multicolumn{1}{c|}{346.78} & 1.69E-07 & \multicolumn{1}{c|}{N/A} & N/A \\ \hline
\end{tabular}
\end{table}

\end{document}